%% file: main.tex
\documentclass[conference]{IEEEtran}
\ifCLASSINFOpdf
\else
\fi
\usepackage{amsmath}
\usepackage{url}

\ifdefined\LONGVERSION
  \relax
\else
 \newcommand{\SUBSTCLO}[1]{}
 \newcommand{\LONGVERSION}[1]{}
 \newcommand{\LONGVERSIONCHECKED}[1]{}
 \newcommand{\SHORTVERSION}[1]{#1}
\fi

\usepackage{proof}
\usepackage{amssymb}
\usepackage{amsthm}
\usepackage{mathtools}
\usepackage{longtable}
\usepackage{enumitem}
\usepackage{xspace}
\usepackage{natbib}
\usepackage{cdsty}
\usepackage{hyperref}

\usepackage[final]{listings}
\usepackage{srcltx}

\usepackage{bbding}
\usepackage{pifont}
\usepackage{wasysym}

\input{macros.tex}

\theoremstyle{definition}
\newtheorem{example}{Example}[section]
\newtheorem{theorem}{Theorem}[section]
\newtheorem{lemma}{Lemma}[section]

\begin{document}
%
\title{A Type Theory for Defining Logics and Proofs}

\author{
\IEEEauthorblockN{Brigitte Pientka\qquad David Thibodeau}
\IEEEauthorblockA{School of Computer Science\\
McGill University
}
\and
\IEEEauthorblockN{Andreas Abel}
\IEEEauthorblockA{Dept. of Computer Science and Eng. \\
Gothenburg University
}
\and
\IEEEauthorblockN{Francisco Ferreira}
\IEEEauthorblockA{Dept. of Computing\\
Imperial College London
}
\and
\IEEEauthorblockN{Rebecca Zucchini}
\IEEEauthorblockA{ENS Paris Saclay
}
}


%



 \IEEEoverridecommandlockouts
  \IEEEpubid{\makebox[\columnwidth]{978-1-7281-3608-0/19/\$31.00~
  \copyright2019 IEEE \hfill} \hspace{\columnsep}\makebox[\columnwidth]{ }}

\maketitle

\begin{abstract}
We describe a Martin-L{\"o}f-style dependent type theory, called \cocon, that allows us to
mix the intensional function space that is used to represent
higher-order abstract syntax (HOAS) trees with the extensional
function space that describes (recursive) computations. We mediate
between HOAS representations and computations using contextual modal
types.  Our type theory also supports an infinite hierarchy of
universes and hence supports type-level computation thereby providing
metaprogramming and (small-scale) reflection.
Our main contribution is the development of a Kripke-style model for
\cocon that allows us to prove normalization. From the normalization
proof, we derive subject reduction and consistency. Our work lays the
foundation to incorporate the methodology of logical frameworks into
systems such as Agda and bridges the longstanding gap between
these two worlds.
\end{abstract}


%
\IEEEpeerreviewmaketitle

\renewcommand{\cite}[1]{\citep{#1}}

\section{Introduction}

Higher-order abstract syntax (HOAS)  is an elegant and deceptively simple idea of encoding syntax and more generally formal systems given via axioms and inference rules. The basic idea
is to map uniformly binding structures in our object language (OL) to the function space in a meta-language thereby inheriting $\alpha$-renaming and capture-avoiding substitution. In the logical framework LF \citep{Harper93jacm}, for example, we encode a simple OL consisting of functions, function application, and let-expressions using a type \lstinline!tm! as:
\begin{lstlisting}
lam : (tm -> tm) -> tm.
app : tm -> tm -> tm.
letv: tm -> (tm -> tm) -> tm.
\end{lstlisting}

The OL term $(\mathsf{lam}\;x.\mathsf{lam}\;y.\mathsf{let}\;w =
x\;y\;\mathsf{in}\;w\;y)$ is then encoded as
\begin{lstlisting}
lam \x.lam \y.letv (app x y) \w.app w y
\end{lstlisting}
using the LF abstractions to model binding.
OL substitution is modelled through LF application; for instance, the fact that
$((\mathsf{lam}~x.M)~N)$ reduces to $[N/x]M$ in our object language is expressed as
\lstinline[basicstyle=\ttfamily\footnotesize]!(app (lam M) N)!
reducing to
\lstinline[basicstyle=\ttfamily\footnotesize]!(M N)!.
%
%
This approach can offer substantial benefits: programmers do not need to build up the
basic mathematical infrastructure, they can work at a higher-level of abstraction, encodings are more compact, and hence it is easier to mechanize formal systems together with their meta-theory. 

However, this approach relies on the fact that we use an
\emph{intensional} function space that lacks recursion, case analysis,
inductive types, and universes to adequately represent syntax.  In LF,
for example, we use the
dependently-typed lambda calculus as a meta-language to represent
formal systems.
Under this view, intensional LF-style functions represent syntactic binding
structures and functions are transparent. However, we cannot write recursive programs about such syntactic structures \emph{within} LF, as we lack the power of recursion.
In contrast, (recursive) computation relies on the \emph{extensional}
type-theoretic function space.
Under
this view, functions are opaque and programmers cannot compare two
functions for equality nor can they use pattern matching on functions to inspect
their bodies. Functions are treated as a black box.

To understand the fundamental difference between defining HOAS trees in LF vs.
defining HOAS-style trees using inductive types, let us consider an inductive type
\lstinline!D! with one constructor
\lstinline!lam: (D -> D) -> D!. What is the problem with such a
definition in type theory? -- In functional ML-like languages, this is, of course,
possible, and types like \lstinline|D| can be explained using domain theory
\cite{scott:dataTypesAsLattices}.
However, the function argument to the constructor \lstinline!lam! is
opaque and we would not be able to pattern match deeper on the
argument to inspect the shape and structure of the syntax
tree that is described by it. We can only observe it by applying it to
some argument. The resulting encoding also would not be adequate,
i.e.~there are terms of type \lstinline!D! that are in normal form but
do not uniquely correspond to a term in the object language we try to
model. As a consequence, we may need to rule out such ``exotic''
representations \cite{Despeyroux:TLCA95}. But there is a more fundamental problem. In proof
assistants based on type theory such as Coq or Agda, we cannot afford
to work within an inconsistent system and we demand that all programs we write are
terminating. The definition of a constructor \lstinline!lam! as given
previously would be forbidden, as it violates what is known as the
positivity restriction. \LONGVERSION{Were we to allow it, we can easily write
non-terminating programs by pattern matching -- even without making a recursive call.

\vspace{0.2cm}
\hspace{-0.5cm}
\begin{footnotesize}
\begin{tabular}{ll}
\lstinline!apply:D -> (D -> D)!  & \lstinline!omega:D!\\
\lstinline!apply (lam f) = f!    & \lstinline!omega = lam (\x -> apply x x)!\\[0.5em]
\lstinline!Omega : D! & \\
\multicolumn{2}{l}{\lstinline!Omega = apply omega omega!}
\end{tabular}
\end{footnotesize}


Here we simply write two functions: the function \lstinline!apply! unpacks an
object of type \lstinline!D! using pattern matching and the function
\lstinline!omega! creates an object of type \lstinline!D!. Using
\lstinline!apply! and \lstinline!omega! we can now write a
non-terminating program that will continue to reproduce itself.
}

It is worth stressing that although we have extensional
type-theoretic functions, we may still have an intensional type theory
 keeping the definitional equality (and hence type checking)
 decidable. This notion of intensional equality should not be confused
 with the intensional LF-style function space that we attributed to LF and
 contrasted to the extensional function space that exists in type
 theories. 

The above example begs two questions: How can we reason inductively about
LF definitions, if they are seemingly not inductive? Do we have to
simply give up on HOAS definitions to model syntactic structures
within type theory to remain consistent?


Over the past two decades, we have made substantial progress in
bringing the intensional and extensional views closer
together. \citet{Despeyroux97} made the key observation that we can
mediate between the weak LF and the strong computation-level
function space using a box modality. The authors describe a
simply-typed lambda calculus with iteration and case constructs which
preserves the adequacy of HOAS encodings. The
well-known paradoxes are avoided through the use
of a modal box operator which obeys the laws of S4. In addition to
being simply typed, all computation had to be on closed HOAS trees. \citet{Despeyroux99}
sketch an extension to dependent type theory --
however it lacks a normalization proof. %

\beluga \cite{Pientka:IJCAR10,Pientka:CADE15} took another important
step towards writing inductive proofs
about HOAS trees by generalizing the box-modality to a contextual
modal type \cite{Nanevski:ICML05,Pientka:POPL08}. For
example, we characterize the OL term $\mathsf{let}\;w = x\;y\;\mathsf{in}\;w\;y$
as a contextual LF object
$\cbox{x, y \vdash \cletv (\capp x~y)~\lambda w.\capp w~y}$ pairing the LF term
together with its LF context. Its contextual type is
$\cbox{x{:}\tm,y{:}\tm \vdash \tm}$.  Here, $\cbox{\;\;}$ is a generalization of the box modality
described in \citet{Despeyroux97}. In particular, elements of type
$\cbox{x{:}\tm,y{:}\tm \vdash \tm}$ can be described as a set of terms
of type $\tm$ that may contain variables $x$ and $y$.
%
This allows us to adapt standard case distinction and recursion
principles to analyze \emph{contextual HOAS trees}.
This is in contrast to recursion principles on open LF
terms (see for example \citet{Hofmann:LICS99}) that are non-standard.

However, the gap to dependent type theories with recursion
and universes such as Martin-L{\"of} type theory still remains. In particular, \beluga cleanly
separates \emph{representing} syntax from \emph{reasoning} about syntax. The
resulting language is an indexed type system in the tradition of
\citet{Zenger:TCS97} and \citet{Xi99popl} where the index language is
contextual LF. 
This has the
advantage that meta-theoretic proofs
are modular and only hinge on the fact that equality in the index domain is
decidable.  However, this approach also is limited in its expressiveness: there is no  support for type-level computation or higher-ranked polymorphism, and we lack the power to express properties of computations.
This prevents us from fully
exploiting the power of metaprogramming and reflection.


In this paper, we present the Martin-L{\"o}f style dependent type
theory \cocon where we mediate between intensional LF objects and extensional
type-theoretic computations using contextual types.
As in \beluga, we can write recursive programs about contextual LF
objects. However, in contrast to \beluga, we also allow computations
to be embedded  within LF objects. For example, if a program $t$
promises to compute a value of type $\cbox{x{:}\tm,y{:}\tm \vdash \tm}$, then
we can embed $t$ directly into an LF object writing
$\clam \lambda x. \clam \lambda y. \capp \unbox{t}{}~x$, unboxing $t$.
If helpful, one might think of boxing ($\cbox{~}$) as
quoting syntax and unboxing ($\unbox{~}{}$) as unquoting computation and embedding its
value within the syntax tree.

Allowing computations within LF objects might seem like a small
change syntactically, but it has far reaching consequences. To
establish consistency of the type theory, we cannot consider
normalization of LF separately from normalization of computations
anymore, as it is done in \citet{Pientka:TLCA15} and
\citet{JacobRao:stratified2018}.
Moreover, \cocon is a
predicative type theory and supports an infinite hierarchy of
universes. This allows us to write type-level computation, i.e. we can
compute types whose shape depends on a given value. Such recursively
defined types are sometimes called large eliminations
\cite{Werner:1992}. Due to the presence of type-level computations,
dependencies cannot be erased from the model. As a consequence, the
simpler proof technique of \citet{Harper03tocl}, which considers
approximate shape of types and has been used to prove completeness of
the equivalence algorithm for LF's type theory, cannot be used in our
setting. Instead, we follow recent work by \citet{Abel:LMCS12} and
\citet{Abel:POPL18} in defining a Kripke-style semantic model for
computations that is defined recursively on its semantic type.
Our model highlights the intensional character of the LF
function space and the extensional character of computations.
Our main contribution is the design of the Kripke-style model for the
dependent type theory \cocon that allows us to establish
normalization. From the normalization proof, we derive type
uniqueness, subject reduction, and consistency. 

We believe \cocon lays the foundation to incorporate the methodology of logical frameworks into systems such as Agda \cite{Norell:phd07} or Coq \cite{bertot/casteran:2004}.
 This finally allows us to combine the world of type theory and logical frameworks inheriting the best of both worlds.


\section{Motivation}
To motivate why we want to
combine the power of LF with a full dependent type theory, we sketch
here the translation of the simply typed lambda calculus (STLC) into
cartesian closed categories (CCC) using our framework.
%
To begin, we encode simple types in LF using the type family
\lstinline!obj!.
\begin{lstlisting}
obj   : type.
one   : obj.
cross : obj -> obj -> obj.
arrow : obj -> obj -> obj.
\end{lstlisting}

We then encode STLC using the indexed type family \lstinline!tm! to
only capture well-typed terms. As before we use the intrinsic LF function space to encode STLC using HOAS.
\begin{lstlisting}
tm    : obj -> type.
tUnit : tm one.
tPair : tm A -> tm B -> tm (cross A B).
tFst  : tm (cross A B) -> tm A.
tSnd  : tm (cross A B) -> tm B.
tLam  : (tm A -> tm B) -> tm (arrow A B).
tApp  : tm (arrow A B) -> tm A -> tm B.
\end{lstlisting}

As is common practice in implementations of LF, we treat free
variables \lstinline!A! and \lstinline!B! as implicitly $\Pi$-quantified at
the outside; they can typically be reconstructed.
Our goal is to translate between well-typed terms in STLC and
morphisms and also state some of the equivalence theorems. Morphisms
are relations between objects. The standard morphisms in CCC can be
encoded directly where we define the composition of morphisms
using \lstinline!@! as an infix operation for better readability.
\begin{lstlisting}
mor  : obj -> obj -> type.
id   : mor A A.
@    : mor B C -> mor A B -> mor A C.
drop : mor A one.
fst  : mor (cross A B) A.
snd  : mor (cross A B) B.
pair : mor A B -> mor A C -> mor A (cross B C).
app  : mor (cross (arrow B C) B) C.
cur  : mor (cross A B) C -> mor A (arrow B C).
\end{lstlisting}

To translate well-typed terms in STLC, we need to traverse terms under
binders. Following \beluga, we introduce a context schema,
\lstinline!ctx!, that classifies contexts containing declarations of
type \lstinline!tm A! for some object \lstinline!A! (see page
\pageref{p:lfcontext} and also \citet{Pientka:PPDP08}).
%
%
 Before we can 
interpret STLC into CCC, we must describe how to interpret a
context as an object in CCC. This is what the function
\lstinline!ictx! does. It has type
\lstinline!($\gamma$ : ctx) => [ |- obj]!. Here we write
\lstinline!=>! for the extensional function space in contrast to
\lstinline!->! which we use for the intensional LF function space.
 For better readability,
we write our function using pattern matching, although the core type
theory we present subsequently uses recursors.
\begin{lstlisting}
rec ictx : ($\gamma$ : ctx) => [ |- obj] =
fn $\;$$\cdot$$\,\,\,$                  = [ |- one]
 | $\gamma$, x:tm (u<A>u with $\cdot$) = [ |- cross u<ictx \gamma>u u<A>u];
\end{lstlisting}

The function \lstinline!ictx! takes as input a context \lstinline!\gamma!
which we analyze via pattern matching. Intuitively, $\gamma$ is
built like lists and we can pattern match on $\gamma$
considering the empty context, written as \lstinline!$\cdot$!, and the
context that contains at least one declaration
\lstinline!x:tm (u<A>u with $\cdot$)!.
Both $\gamma$ and $A$ are pattern variables; they are bound on the
computation level. This is in contrast to LF variables that occur
inside a box and are bound by LF lambda-abstraction or by the LF
context associated with an LF object.
As \lstinline!A! denotes a
closed object and does not depend on \lstinline!\gamma!, we unbox it
together with the weakening substitution (written as $\cdot$) which
moves \lstinline!A! from the
empty LF context to the LF context $\gamma$. In general, we write
\lstinline!u<t>u with $\sigma$! for the unboxing of a computation-level term $t$ together
with an LF substitution $\sigma$ (see also page \pageref{p:unbox}).
We omit the \lstinline!with! keyword and the LF
substitution associated with unboxing, if it is the identity.
%


The function \lstinline!ictx! returns a closed object which we indicate
by \lstinline![ |- obj]!. Note that we do not simply return an LF
object of type \lstinline!obj!, as we mediate between LF objects
and computations using box and unbox.

The ideas so far follow
closely \beluga, a programming environment
that supports writing recursive programs about LF
specifications.
(However, in contrast to \beluga, we inline the recursive call
using unbox, written as \lstinline!u<ictx \gamma>u!,
as opposed to require a let-style binding.)

The real power of having a Martin-L{\"o}f style type theory, where we can embed
computations within contextual types, becomes apparent when we define the
interpretation of STLC into CCC.
The type of the interpretation function \lstinline!itm! concisely
specifies that it translates a well-typed lambda term \lstinline!m!,
that has type \lstinline!A! in the context \lstinline!\gamma!, to a
morphism from \lstinline!ictx \gamma! to \lstinline!A!. Here we rely
on the function \lstinline!ictx! that translates a context
\lstinline!\gamma! to an object. Adopting Agda's approach, we use curly
braces to indicate implicit arguments and round braces for explicit
arguments.
In \beluga{}
we would not
be able to refer to the function   \lstinline!ictx! inside the type
declaration of \lstinline!itm!, as \beluga makes a clear distinction
between contextual LF types (and LF objects) and functions about
them.
\begin{lstlisting}
rec itm : {$\gamma\;\,$: ctx} => {A : [ |- obj]} =>
         (m : [$\gamma$ |- tm (u<A>u with $\cdot$)]) =>
         $\;$[ |- mor u<ictx \gamma>u u<A>u] =
fn (p : [$\gamma$ $\vdash_\#$ tm (u<A>u with $\cdot$)]) = ivar $\gamma$ p
 | [$\gamma$ |- tUnit]               $\;\,$= [ |- drop]
 | [$\gamma$ |- tFst u<e>u]             $\,$= [ |- fst @ u<itm e>u]
 | [$\gamma$ |- tSnd u<e>u]             $\,$= [ |- snd @ u<itm e>u]
 | [$\gamma$ |- tPair u<e1>u u<e2>u]      $\;\,$=
     [ |- pair u<itm e1>u u<itm e2>u]
 | [$\gamma$ |- tLam \x.u<e>u]          $\,$= [ |- cur u<itm e>u]
 | [$\gamma$ |- tApp u<e1>u u<e2>u]       $\;\,$=
     [ |- app @ pair u<itm e1>u u<itm e2>u];
\end{lstlisting}

We implement the interpretation of STLC as morphisms by
pattern matching on \lstinline!m! considering all the constructors to
build lambda terms plus the variable case, i.e.
when we have a variable from \lstinline!\gamma!.
In the latter case, we use the pattern variable
\lstinline!p! with contextual type
\lstinline! [$\gamma$ $\vdash_\#$ tm (u<A>u with $\cdot$)]!
 which can only be instantiated with a variable from
\lstinline!\gamma!.
 We omit here the implementation of
 \lstinline!ivar! for lack of space. It simply looks up a variable
 in the LF context $\gamma$ and builds the corresponding projection.
The most interesting case is
\lstinline![$\gamma$ |- tLam \x.u<e>u]!, where \lstinline!e! has type
\lstinline![$\gamma$, x:tm u<B>u |- tm u<C>u]!. The recursive call
\lstinline!itm e! returns a morphism from
\lstinline!u<ictx ($\gamma$, x:tm u<B>u)>u! to \lstinline!u<C>u! which
matches
what is expected by \lstinline!cur!, since
\lstinline!u<ictx ($\gamma$, x:tm u<B>u)>u! evaluates to
\lstinline!(cross u<ictx $\;\gamma$>u u<B>u)!.

%
%
%
%
%
%

Next we translate a morphism to a STL term. Given a morphism between \lstinline!A! and
\lstinline!B!, we build a term of type \lstinline!B! with one variable
of type \lstinline!A!. As our types are closed, we again employ the
weakening substitution whenever we refer to \lstinline!B! in a
non-empty context.
\begin{lstlisting}
rec imorph : {A : [ |- obj]} => {B : [ |- obj]} =>
            (m : [ |- mor u<A>u u<B>u]) =>
            $\;$[x:tm u<A>u |- tm (u<B>u with $\cdot$)] =
fn $\;$[ |- id]          = [x:tm _ |- x]
 | [ |- drop]        = [x:tm _ |- tUnit]
 | [ |- fst]         = [x:tm _ |- tFst x]
 | [ |- snd]         = [x:tm _ |- tSnd x]
 | [ |- pair u<f>u u<g>u]$\;\;$=
     [x:tm _ |- tPair u<imorph f>u u<imorph g>u]
 | [ |- cur u<f>u]     $\;$=
     [x:tm _ |- tLam \y.(u<imorph f>u$\;$with tPair x y)]
 | [ |- u<f>u @ u<g>u]    =
     [x:tm _ |- u<imorph f>u with u<imorph g>u]
 | [ |- app]        $\;\;$=
     [x:tm _ |- tApp (tFst x) (tSnd x)];
\end{lstlisting}

The translation is mostly straightforward. The most interesting cases
are the case for currying and composition.
In the former, given \lstinline!cur u<f>u!
of type \lstinline![ |- mor u<A>u (arrow u<B>u u<C>u)]!,
we recursively translate \lstinline!f:[ |- mor (cross  u<A>u  u<B>u) C]!.
It yields a STL term of
type \lstinline![x:tm (cross u<A>u u<B>u)|- tm _]!.
We now need to replace the LF variable \lstinline!x! that occurs in the result of the recursive call \lstinline!imorph f! with
\lstinline!tPair x y! to build a STL term
\lstinline![x:tm u<A>u |- tm (arrow u<B>u u<C>u)]!. We hence unbox
the result of the recursive call with the
substitution \lstinline!tPair x y!. This is written as
\lstinline!u<imorph f>u with tPair x y!.
Here we do not write the domain of
the substitution explicitly, however the type of \lstinline!imorph f!
tells us that the result of translating \lstinline!f! contains one LF
variable. In general, we write LF substitutions as lists
whose domain is determined by the contextual object we unbox.

To translate a morphism \lstinline![ |- u<f>u @ u<g>u]!,
we recursively translate \lstinline!f! and \lstinline!g! where
\lstinline!imorph f! returns
a STL term of type  \mbox{\lstinline![x:tm u<B>u |- tm u<C>u]!} and
\lstinline!imorph g! returns
a STL term of type  \lstinline![x:tm u<A>u |- tm u<B>u]!. To produce the desired
STL term of type  \lstinline![x:tm u<A>u |- tm u<C>u]!, we replace the
LF variable \lstinline!x! in the translation of \lstinline!f! with the
result of the translation of \lstinline!g!. This is simply done by
\mbox{\lstinline![x:tm u<A>u |- u<imorph f>u with u<imorph g>u]!}. We note that
\lstinline!imorph g! is unboxed with the identity substitution
and hence the LF variable that occurs in the result of
\lstinline!u<imorph g>u! is implicitly renamed and bound by \lstinline!x!.


Finally, we sketch the equivalence between STLC and CCC to illustrate what new possibilities \cocon opens up. We do not
show the concrete implementation, since this would go beyond this
paper.

Assuming that we have defined convertibility of lambda-terms
(\lstinline!conv!) and equality (\lstinline!~!) between
morphism, we can now state the equivalence between STLC and CCC
succinctly.
%
%
%
%
%
%
\begin{lstlisting}
rec stlc2ccc : {\gamma : ctx} => {A : [ |- obj]} =>
      {M : [\gamma |- tm u<A>u with $\cdot$]} =>
      {N : [\gamma |- tm u<A>u with $\cdot$]} =>
      (e : [\gamma |- conv u<M>u u<N>u]) =>
      [ |- u<itm M>u ~ u<itm N>u ]

rec ccc2tm : {A : [ |- obj]} => {B : [ |- obj]} =>
      {f : [ |- mor u<A>u u<B>u]} =>
      {g : [ |- mor u<A>u u<B>u]} =>
      (m : [ |- u<f>u ~ u<g>u]) =>
      [x:tm A |- conv u<imorph f>u u<imorph g>u]

\end{lstlisting}

We hope this example provides a glimpse of what \cocon has to
offer. In the rest of the paper, we develop the dependent type
theory for \cocon that supports both defining HOAS trees using the
intensional function space of LF and defining (type-level)
computations using the extensional function space.

\section{A Type Theory for Defining Logics and Proofs }\label{sec:cocon}
\cocon combines the logical framework LF with a full dependent type
theory that supports recursion over HOAS objects and universes. We
split \cocon's grammar into different syntactic categories (see
Fig. \ref{fig:grammar}).
\LONGVERSION{The LF layer describes LF
  objects, LF types, LF contexts; the computation layer consists of
  terms and types that describe recursive computation and
  universes. We mediate the interaction between LF objects and
  computations via a (contextual) box modality: we embed contextual LF
  objects into computations, by pairing an LF object with its LF
  contexts and we embed computations within LF objects by unboxing the
  result of a computation. }
%
\begin{figure}[htb]
\begin{center}
  \begin{small}
\[
\begin{array}{p{3cm}@{~}l@{~}r@{~}l}
LF kinds           & K           & \bnfas & \lftype \bnfalt \Pityp x A K  \\
LF types           & A, B        & \bnfas & \const{a}~M_1 \ldots M_n \bnfalt \Pityp x A B\\
LF terms           & M, N       & \bnfas & \lambda x.M \mid M\,N \mid
 x \mid \const{c} \mid \unbox t \sigma \\
LF contexts        & \Psi, \Phi & \bnfas & \cdot \bnfalt \psi \bnfalt \Psi, x{:}A\\
LF context (erased) & \hatctx{\Psi},\hatctx{\Phi} & \bnfas & \cdot \bnfalt \psi \bnfalt\hatctx{\Psi}, x\\
LF substitutions   & \sigma   & \bnfas & \cdot \bnfalt \wk{\hatctx\Psi} \bnfalt \sigma, M \SUBSTCLO{\bnfalt \sclo {\hatctx\Psi} {\unbox t \sigma}} \\
LF signature       & \Sigma   & \bnfas & \cdot \mid \Sigma,
                                         \const{a}{:}K \mid \Sigma, \const{c}{:}A
\\[0.25em]
\hline
\\[-0.75em]
Contextual types & T & \bnfas &
                     \Psi \vdash A \bnfalt      
                     \Psi \vdash_\# A           
   \SUBSTCLO{\bnfalt \Psi \vdash \Phi \mid \Psi \vdash_\# \Phi}
\\
Contextual objects & C & \bnfas &
                      \hatctx{\Psi} \vdash M    
    \SUBSTCLO{\bnfalt \hatctx{\Psi} \vdash \sigma }
\\[0.25em]
\hline
\\[-0.75em]

Sorts     & u          & \bnfas & \univ k \\
Domain of discourse & \ann\tau & \bnfas & \tau \bnfalt \ctx \\
Types and & \tau, \IH, & \bnfas &  u \bnfalt \cbox T
                             \bnfalt (y :\ann{\tau}_1) \arrow \tau_2 
\\
Terms & t, s &  \bnfalt &  y \bnfalt  \cbox C \bnfalt \titer{\vec\R}{}{\IH} \rappto\Psi~\vec{t} \\
      &      &  \bnfalt & \tmfn y t \bnfalt  t_1~t_2

\\
Branches  & {\R} & \bnfas & \Gamma \mto t
\\
Contexts & \Gamma & \bnfas & \cdot \bnfalt \Gamma, y:\ann\tau 
\end{array}
\]
  \end{small}
\end{center}
\vspace{-0.15cm}
  \caption{Syntax of \cocon}
  \label{fig:grammar}
\end{figure}

\subsection{Syntax}
\paragraph{Logical framework LF with embedded computations}
As in LF, we allow dependent kinds and types; LF terms can be defined by
LF variables, constants, LF applications, and LF
lambda-abstractions. In addition, we allow a computation $t$ to be
embedded into LF terms using a closure $\unbox t {\sigma}$.
Once computation of $t$ produces a contextual object $M$ in an LF
context $\Psi$, we can embed the result by applying the substitution
$\sigma$ to $M$, moving $M$ from the LF context $\Psi$ to the current
context $\Phi$. In the source level syntax that we previously used in
the code examples, this was written as
 \lstinline!u<t>u with $\sigma$!.
\label{p:unbox}

We distinguish between computations that characterize a general LF
\emph{term} $M$ of type $A$ in a context $\Psi$, using the contextual type
$\Psi \vdash A$, and computations that are guaranteed to return a
\emph{variable} in a context $\Psi$ of type $A$, using the contextual type
$\Psi \vdash_\# A$. This distinction is exploited in the definition of
a recursor for contextual objects of type $\cbox{\Psi \vdash A}$
to characterize the base case where we consider an LF variable of LF
type $A$. \SHORTVERSION{For simplicity and lack of space, we focus on $\Psi \vdash A$ in the subsequent development. Intuitively, $\Psi \vdash_\# A$ is a special case restricted to variables from $\Psi$ inhabiting $A$.}

\LONGVERSION{For simplicity, we fix here the LF signature to include the type $\tm$ and the LF constants $\tlam$ and $\tapp$. This allows us to for example define recursors on $\tm$-objects directly.}


\paragraph{LF contexts}\label{p:lfcontext}
LF contexts are either empty or are built by extending a context with a declaration $x{:}A$.
We may also use a (context) variable $\psi$ that stands for a context prefix and must be declared on the computation-level. In particular, we can write functions where we abstract over (context) variables. Consequently, we can pass LF contexts as arguments to functions. We classify LF contexts via schemata -- for this paper, we pre-define the schema $\tmctx$.
Such context schemata are similar to adding base types to computation-level types.
We often do not need to carry the full LF context with the type annotations, but it suffices to simply consider the erased LF context. Erased LF contexts are simply lists of variables possibly with a context variable at the head.
We sometimes abuse notation and write $\hatctx{\Psi}$ for the result of erasing type information from an LF context $\Psi$.

\LONGVERSION{We can easily support recursion on LF contexts, but we do not allow computations that return a new LF context. This simplifies the design. Recall that the head of a context denotes a possibly empty sequence of declarations. This prefix should be abstract and opaque to any LF term or LF type that is considered within this context. In other words, an LF term $M$ (or LF type $A$) should be meaningful without requiring any specific knowledge about the prefix of declarations. Second, it would be difficult to enforce well-scoping and $\alpha$-renaming. To illustrate, consider the following LF term $\tapp~x~y$ in the LF context $x{:}\tm, y{:}\tm$. If we were to allow type checking to exploit equivalence relations on LF contexts that take into account computations on LF contexts, we can argue that since $x{:}\tm, y{:}\tm$ is equivalent to $\unboxc{\tcopy~\cbox{x{:}\tm, y{:}\tm}}$, $\tapp~x~y$ should also be meaningful in the latter LF context. However, now the LF variables $x$ and $y$ are free in $\tapp~x~y$.}




\paragraph{LF substitutions}

LF substitutions allow us to move between LF contexts.
The \emph{compound substitution} $\sigma,M$ extends
substitution $\sigma$ with domain $\hat\Psi$
to a substitution with domain $\hat\Psi,x$, where $M$ replaces $x$.
However, following \citet{Nanevski:ICML05},
we do not store the domain (like $\hat\Psi$) in the substitution,
it will be supplied when applying the substitution to a term
(see Section~\ref{sec:lfsubst}).
The \emph{empty substitution} $\cdot$ provides a mapping
from an empty LF context to \emph{any} LF context $\Psi$,
including a context variable $\psi$,
hence, has weakening built in.
The \emph{weakening substitution},
written as $\wk{\hatctx\Psi}$, describes the
weakening of the domain $\Psi$ to $\Psi, \wvec{x{:}A}$.
We simply write $\id$ when $|\wvec{x{:}A}| = 0$.
Unless $\hatctx\Psi$ is a context variable $\psi$, weakening
$\wk{\hat\Psi}$ is a redex
where $\wkempty$ reduces to the empty substitution
and $\wk{\hat\Psi,x}$ reduces to the compound substitution
$\wk{\hat\Psi},x$ (see also figures \ref{fig:lfeq} and \ref{fig:lfwhnfred}).
Note, however, that $\wkempty$ only describes weakening of the empty context
to a \emph{concrete} context ${\cdot}, \wvec{x{:}A}$
and, thus, does not subsume the empty substitution.




%
%

From a de Bruijn perspective, the weakening substitution $\wkempty$
which maps the empty context to
${\cdot, x_n{:}A_n, \ldots, x_1{:}A_1}$ can be viewed as a shift by
$n$. Further, like in the de Bruijn world,
$\wk {(\cdot,x_n{:}A_n, \ldots, x_1{:}A_1)}$ can be expanded and is
equivalent to the substitution $\cdot, x_n, \ldots, x_1$.
While our theory lends itself to an implementation with de Bruijn indices, we formulate our type theory using a named representation of variables. This not only simplifies our subsequent definitions of substitutions, but also leaves open how variables are realized in an implementation.


\SUBSTCLO{
Finally, we allow computations to be embedded into substitutions; such computations are useful, if want to state abstractly properties about substitutions, for example if we want to relate a context $\psi:{\tmctx}$ to another context $\phi:{\tmctx}$. Such properties naturally arise meta-theoretic proofs.
 We embed computations into substitutions using LF closures, written as $\unbox t \sigma$ following \cite{Cave:LFMTP13}. Intuitively, these closures describe when applying an LF substitution gets stuck. When we compose LF substitutions, we may get stuck when we compose $\sigma$ with a computation $t$ that promised to produce eventually an LF substitution. Hence we form an LF closure $\unbox{t}\sigma$. Further, we incorporate weakening. If we compose a weakening substitution with a computation $t$, we may be again stuck until we can proceed and know the result of $t$. These two observations lead to defining suspended substitutions as $\sclo {\hatctx\Psi} {\unbox t \sigma}$. The development follows closely ideas by \citet{Cave:LFMTP13}.

In addition to general contextual types $\Psi \vdash \Phi$ that describe a (contextual) LF substitution, we can also support the restricted form of contextual substitution types $\Psi \vdash_\# \Phi$ that are only inhabited by weakening substitutions. First-class support for weakening substitutions has proven useful in encoding formal systems using contextual LF (see for example \cite{Cave:LFMTP15}). Such renamings are crucial in many proofs such as for example normalization proofs using logical predicates \cite{Cave:LFMTP13,Cave:LFMTP15,Thibodeau:Howe16,POPLMarkReloaded}. For the subsequent development, we will concentrate on LF substitutions of contextual types $\Psi \vdash \Phi$.
}

\paragraph{Contextual objects and types}
We mediate between LF and computations using contextual types. Here,
we concentrate on contextual LF terms that have type $\Psi \vdash
A$. However, others  may be added \cite{Cave:LFMTP13}. \SUBSTCLO{Similarly, we can also add  general contextual LF substitutions that have type $\Psi \vdash \Phi$ and contextual weakening substitutions that have type $\Psi \vdash_\# \Phi$. }


\paragraph{Computations and their types}
Computations are formed by extensional functions, written as $\tmfn y t$,
applications, written as $t_1~t_2$,
boxed contextual objects, written as $\cbox C$, and
the recursor, written
as $\titer{\vec\R}{}{\IH} \rappto\Psi~\vec{t}$, where
$\vec t = t_n\ldots t_0$. We annotate the
recursor with the typing invariant $\IH$. We may either recurse over
$\Psi$ directly or we recurse over the values computed by the term
$t_0$. The LF context $\Psi$ describes the local LF world in which the
value computed by $t_0$ makes sense. The arguments $t_n \dots t_1$
describe in general the implicit arguments $t_0$ might depend on.
Finally, $\vec\R$ describes the different branches that we can take
depending on the value computed by $t_0$. A covering set of branches can be
generated generically following \citet{Pientka:TLCA15}.
In this paper, we will subsequently work with a recursor for the LF type $\tm$ 
which we encountered in the introduction, together with two LF
constants 
$\tlam : \Pityp y {(\Pityp x \tm \tm)}\tm$ and 
$\tapp : \Pityp x \tm {\Pityp y \tm \tm}$ to keep the development compact.


\LONGVERSION{For illustration, we also include other branches to construct also recursors over contextual objects of type $\cbox{\Psi \vdash_\# \tm}$, i.e. variables of type $\tm$ in the LF context $\Psi$. In this case, we only consider two branches where $\Psi = \Psi', x{:}\tm$: in the first branch $({\psi \mto t_x})$ we pattern match against $x$ and $\psi$ will be instantiated with $\Psi'$; in the second branch, $({\psi, y, f_y \mto t_y})$, the LF variable we are looking for is not $x$, but is somewhere in $\Psi'$. In this case, $y$ denotes intuitively an LF variable that is not $x$ and has type $\cbox{\psi \vdash_\# \tm}$ and we will instantiate $\psi$ with $\Psi'$; $f_y$ is the recursive call on the smaller LF context $\Psi'$ and $t_y$ is the body of the branch. }
\LONGVERSION{We also include branches for recursing over LF substitution which have either the contextual type $\cbox{\Psi \vdash \Phi}$ or $\cbox{\Psi \vdash_\# \Phi}$. Here we consider two cases: either the LF substitution is empty then we choose the first branch, or it is of the shape $\sigma, m$ and $f_\sigma$ denotes the recursive call on the smaller LF substitution $\sigma$. }

Computation-level types consist of boxed contextual types, written as
$\cbox T$, and dependent types, written as $ (y:\ann{\tau}_1) \arrow
\tau_2$. We overload the dependent function space and allow as domain
of discourse both computation-level types and the schema $\tmctx$ of
LF context. We use $\tmfn y t$ to introduce functions of both kinds. We also
overload function application $t\;s$ to eliminate dependent types $(y :
\tau_1) \arrow \tau_2$ and $(y : \tmctx) \arrow \tau_2$, although in
the latter case $s$ stands for an LF context. We separate LF contexts
from contextual objects, as we do not allow functions that return an
LF context. 

\cocon has an infinite hierarchy of predicative universes, written as
$\univ k$ where $k \in \NN$.  The universes are not
cumulative. Adopting PTS-style notation, we can define \cocon and its
universes using sorts $u \in \So = \{ \univ i \mid i \in \NN \}$,
axioms $\Ax = \{(\univ i,\, \univ {i+1}) \mid i \in \NN \}$, and rules
$\Ru = \{ (\univ i,\, \univ j,\, \univ {\mathsf{max}(i,j)}) \mid i,j \in \NN \}$.

\LONGVERSION{
\begin{example}
We return the position of an LF variable in an LF context by writing a function $\mathsf{pos}$ that has type $\IH = (\psi : \tmctx) \arrow (x : \cbox{\psi \vdash_\# \tm}) \arrow \mathsf{nat}$.

\[
  \begin{array}{r@{~}lcl}
\mathsf{fn}~\psi \Rightarrow \mathsf{fn}~x \Rightarrow \mathsf{rec}^\IH (& \psi & \mto & 0 \\
\mid & \psi, y, f_y & \mto & 1 + f_y)~\psi~x
  \end{array}
\]

\end{example}

}

\subsection{LF Substitution Operation}
\label{sec:lfsubst}

Our type theory distinguishes between LF variables and computation
variables and we define substitution for both.
We define LF substitutions uniformly using a simultaneous substitution
operation written as $\lfs \sigma \Psi M$\LONGVERSION{ (and similarly
  $\lfs \sigma \Psi A$ and $\lfs \sigma \Psi K$) }.  As an LF
substitution $\sigma$ is simply a list of terms, we need to supply its
domain $\hat\Psi$ to look up the instantiation for an LF variable $x$ in
$\sigma$.

\LONGVERSION{Let's consider first a few examples to get a better intuition. Let's look at a few examples to get a better intuition.
%

\paragraph{Examples 1} Consider the LF term $\tapp~\unbox{\cbox{x,y \vdash \tapp~x~y}}{\wk{x, y}}~w$. This LF term is obviously well-typed in the (normal) LF context $x{:}\tm, y{:}\tm, w{:}\tm$ and applying the substitution $\wk{x,y}$ to $\tapp~x~y$ is meaningful as $\wk{x,y}$ expands to $\cdot, x, y$. When we apply $\cdot, x, y$ to unbox $\cbox{x,y \vdash \tapp~x~y}$, we resurrect the domain and apply $[\cdot, x, y ~/~ \cdot, x, y](\tapp~x~y)$.

\paragraph{Examples 2} What about considering the $\alpha$-equivalent term $\unbox{\cbox{x',y' \vdash \tapp~x'~y'}}{\wk{x,y}}$? --
Again we observe that $\wk{x,y}$ expands to $\cdot, x, y$; when we apply $\cdot, x, y$ to unbox $\cbox{x',y' \vdash \tapp~x'~y'}$, we resurrect the domain and apply $[\cdot,x,y ~/~ \cdot, x', y'](\tapp~x'~y')$ effectively renaming $x'$ and $y'$ to $x$ and $y$ respectively.

}
%
%

%
\begin{small}
  \[
  \begin{array}{l@{~}c@{~}l@{~~~}l}
\lfs \sigma {\Psi} ( \lambda x. M  ) & = &   \lambda x. M' & {\mbox{where~}}\lfs {\sigma,x}{\Psi, x}(M) = M'\quad\\
& & & {\mbox{provided\ that\ }x \notin \FV(\sigma) ~\mbox{and}~ x \not\in \hatctx{\Psi}}
    \\[0.25em]
 \lfs \sigma \Psi (M~N) & = & M'~N' & {\mbox{where~}}{\lfs \sigma \Psi (M) = M'}\\
& & & \mbox{~~~and}~\lfs \sigma \Psi (N)= N'
    \\[0.25em]
\lfs \sigma \Psi (\unbox{ t} {\sigma'})& = & \unbox t {\sigma''} & {\mbox{where~}}
\lfs{\sigma}{\Psi}(\sigma') = \sigma''
    \\[0.25em]
\lfs \sigma \Psi  (x)  & = & M & {\mbox{where~}}\pos x~\lfs \sigma \Psi = M
    \\[0.25em]
 \lfs \sigma \Psi c     & = & \multicolumn{2}{l}{c}  \\[1em]
\lfs \sigma \Psi (\cdot)  & = & \multicolumn{2}{l}{\cdot}  \\[0.25em]
\lfs \sigma \Psi (\wk{\hatctx\Phi})  & = & \sigma' & {\mbox{where~}}\trunc_\Phi ~(\sigma / \hatctx\Psi) = \sigma'\\[0.25em]
\SUBSTCLO{\lfs \sigma \Psi (\sclo {\hatctx\Phi} {\unbox{t}{\sigma'}}) & = & \sclo {\hatctx\Phi} {\unbox{t}{\sigma''}} & {\mbox{where~}}\lfs \sigma \Psi (\sigma') = \sigma''  \\[0.25em]}
\lfs \sigma \Psi (\sigma', M) & = & \sigma'', M' & ~\mbox{where~}
\lfs \sigma \Psi (\sigma') = \sigma''\\
& & & \mbox{~~~~and}~\lfs \sigma \Psi (M)=M'
 \end{array}
  \]
\end{small}

 \LONGVERSION{
\[
\begin{array}{l@{~}c@{~}l@{~}l}
\pos x~\lfs {\sigma, M}{\Psi, x} & = & M & \\[0.25em]
\pos x~\lfs {\sigma, M}{\Psi, y} & = & \pos x~ \lfs \sigma \Psi & \\[0.25em]
\pos x~\lfs {\wk{\hatctx\Psi}}{\Psi} & = & x &\mathrm{where}~x\in\hatctx\Psi \\
\SUBSTCLO{
  \pos x~\lfs {\sclo{\hatctx\Psi} {\unbox t\sigma}} \Psi & = &
    \lfs{\sigma}{\Phi}(M) &
    \mathrm{where}~t = \cbox{\hatctx\Phi \vdash \sigma'}
 \mathrm{and}~
    \trunc_\Psi ~ (\sigma' / \hatctx\Psi) = \sigma''
\\
 & & & ~\mathrm{and}~
    \pos x~\lfs{\sigma''}{\Psi} = M
\\[0.25em]
\pos x~\lfs {\sclo{\hatctx\Psi} {\unbox t \sigma}} \Psi & = & \unbox{\cbox{\hatctx\Psi',x \vdash x}}{\sclo {\hatctx\Psi',x} {\unbox t \sigma}} &
    \mathrm{where}~t \neq {\cbox{\hatctx\Phi \vdash \sigma'}}
  ~\mathrm{and}~\hatctx\Psi = \hatctx\Psi',x, \vec y
\\
}
\pos x~\lfs \sigma \Psi & = & \mbox{fails otherwise}
 \end{array}
  \]
}
Let us comment on a few cases.
\LONGVERSION{
When pushing the substitution through an application $M~N$, we simply
apply it to $M$ and $N$ respectively. When pushing the LF substitution
through a $\lambda$-abstraction, we extend it.
When applying $\sigma$ to an LF variable $x$, we retrieve the
corresponding instantiation from $\sigma$ using the auxiliary function
$\posno$ which works as expected.
}
When applying the LF substitution $\sigma$ to the LF closure
$\unbox{t}{\sigma'}$ we leave $t$ untouched, since $t$ cannot contain
any free LF variables and compose $\sigma$ and $\sigma'$. Composition
of LF substitutions is straightforward.
%
%
%
\SUBSTCLO{
However, we need to be careful in handling substitution closures
$\sclo{\hatctx\Psi} \unbox{t}\sigma$, as we need to guarantee that we
make progress, if we can. For example, if $t =  \cbox{\hatctx\Phi
  \vdash \sigma'}$, then we want continue to look up the variable $x$
in $\sigma'$. To ensure we work with the same domain $\hatctx\Psi$, we
first truncate $\sigma'$ dropping instantiations that do not play a
role. Finally, we apply $\sigma$ to the term that the variable $x$ is
mapped to. If $t$ cannot be further unfolded, then we create a
closure. If $r \neq  \unbox {\cbox{\hatctx\Phi \vdash \sigma'}}
{\sigma}$, we create an LF closure on the LF term level. Intuitively,
we proceed, once the LF closure is known. \SHORTVERSION{The full
  definition can be found in the accompanying long version and is
  omitted due to lack of space (see \citet{cocon:arxiv19}).}
}
%
%
%
When we apply $\sigma$ to $\wk{\hatctx\Phi}$, we truncate $\sigma$ and
only keep those entries corresponding to the LF context $\Phi$. Recall
that $\wk{\hatctx\Phi}$ provides a weakening substitution from a
context $\Phi$ to another context $\Psi = (\Phi, \wvec{x{:}A})$.
\SHORTVERSION{%
Intuitively, truncation throws away the entries of $\sigma$ corresponding to the $\vec x$;
for the formal definition, please consult the long version \citep{cocon:arxiv19}.
}%
\LONGVERSION{%
The simultaneous substitution $\sigma$ provides
mappings for all the variables in $\hatctx{\Phi}, \vec{x}$. The result
of $[\sigma / \hatctx \Phi, \vec{x}]\wk{\hatctx\Phi}$ then should only
provide mappings for all the variables in $\Phi$. We use the operation
$\trunc$ to remove irrelevant instantiations.
The definition of truncation is straightforward.
\[
\begin{array}{l@{~}c@{~}l@{~}l}
\trunc_\Psi ~ (\sigma / \hatctx \Psi) & = & \sigma\\
\trunc_\Psi ~ (\sigma, M / \hatctx \Phi, x) & = & \trunc_\Psi ~(\sigma / \hatctx \Phi)\\
\trunc_\Psi ~ (\wk{(\hatctx\Psi, \vec{x})} / \hatctx \Psi, \vec{x}) & = & \wk{\hatctx\Psi} \\
\SUBSTCLO{\trunc_\Psi ~(\sclo {\hatctx\Psi, \vec{x}} {\unbox t \sigma}  / \hatctx \Psi, \vec{x}) & = & \sclo {\hatctx\Psi} {\unbox t \sigma}\\}
\trunc_\Psi ~(\dot / \cdot) & = & \mbox{fails}~\Psi\neq \cdot
\end{array}
\]
}

\subsection{Computation-level Substitution Operation}

The computation-level substitution operation $\{t/x\}t'$ traverses the computation $t'$ and replaces any free occurrence of the computation-level variable $x$ in $t'$ with $t$\LONGVERSION{\/(see Fig.~\ref{fig:csub})}. The interesting case is $\{t/x\}\cbox{C}$. Here we push the substitution into $C$ and we will further apply it to objects in the LF layer. When we encounter a closure such as $\unbox{t''}{\sigma}$, we continue to push it inside $\sigma$ and also into $t''$.
%
When substituting an LF context $\Psi$ for the variable $\psi$ in a
context $\Phi$, we rename the declarations present in $\Phi$. This is
a convention; it would equally work to rename the variable
declarations in $\Psi$. For example, in
$\{(x{:}\tm, y{:}\tm) / \psi\}(\hatctx \psi, x \vdash \tlam~\lambda y. \tapp~x~y~)$,
we rename the variable $x$ in $(\hatctx\psi, x)$ and replace $\psi$ with
$(x{:}\tm, y{:}\tm)$ in
$(\hatctx\psi, w \vdash \tlam~\lambda y. \tapp~w~y)$. This results in
$x, y, w \vdash \tlam~\lambda y. \tapp~w~y$. When type checking this
term we will eventually also $\alpha$-rename the $\lambda$-bound LF
variable $y$.

Last, we define simultaneous computation-level substitution using the
judgment \fbox{$\Gamma'\vdash \theta : \Gamma$}. For simplicity, we
overload the typing judgment, just writing
$\Gamma \vdash t : \ann\tau$, although when $\ann\tau = \tmctx$, then
$t$ stands for an LF context.

  \[
  \begin{array}{c}
    \infer{\Gamma' \vdash \cdot : \cdot}{\vdash \Gamma'}
    \quad
    \infer{\Gamma'\vdash\theta, t / x :\Gamma, x: \ann\tau}{\Gamma'\vdash \theta : \Gamma & \Gamma' \vdash t : \{\theta \}\ann\tau }
  \end{array}
  \]

We distinguish between a substitution $\theta$ that provides
instantiations for variables declared in the
computation context $\Gamma$, and a renaming substitution $\rho$ which maps variables  in the computation context $\Gamma$ to the same variables in the context $\Gamma'$ where $\Gamma'
= \Gamma, \wvec{x{:}\ann\tau}$ and $\Gamma' \vdash \rho : \Gamma$. We
write $\Gamma' \leq_\rho \Gamma$ for the latter.  We note that the weakening and substitution properties for simultaneous substitutions also hold for renamings.

\LONGVERSION{
\begin{figure}
  \centering
\[
  \begin{array}{l@{~=~}l@{~\qquad}l}
\multicolumn{3}{l}{\mbox{Computation-level substitution for terms}}\\[0.5em]
   \{t/y\}(\tmfn x t') & \tmfn x \{t/y\}t'    & \mathrm{provided~} x \not\in \FV(t) \\[0.1em]
   \{t/y\}(t_1~t_2)    & \{t/y\}t_1 ~ \{t/y\}t_2 & \\[0.1em]
   \{t/y\}(\titer{\R}{}{\IH} \rappto \Psi~t) & \titer {\{t/y\}b_1 \mid \{t/y\}b_k}{}{\IH} \rappto \{t/y\}\Psi~\{t/y\}t & \mathrm{where~} \R = b_1 \mid \ldots \mid b_n \\[0.1em]
   \{t/y\}\cbox C & \cbox {\{t/y\} (C)} & \\[0.1em]
   \{t/y\} (y) & t & \\[0.1em]
   \{t/y\} (x) & x & \mathrm{where~} x \not\eq y\\[0.1em]
   \{t/y\} (\cdot) & \cdot & \\[0.1em]
   \{t/y\} (\Psi, x:A) & \{t/y\}\Psi, x:\{t/y\}A & \mathrm{provided~} (\hatctx{\Psi},x) \not\in \FV(t)\\[0.5em]
\multicolumn{3}{l}{\mbox{Computation-level substitution for branches}}\\[0.5em]
   \{t/y\} (\vec x \mto t') & (\vec x \mto \{t/y\}t') &
   \mathrm{provided~} \vec x \not \in \FV(t)
\\[0.5em]
\multicolumn{3}{l}{\mbox{Computation-level substitution for contextual objects}}\\[0.5em]
   \{t/y\}(\hatctx{\Psi} \vdash M) & \{t/y\}\hatctx\Psi \vdash \{t/y\}M & \mathrm{provided~} \hatctx\Psi\not\in \FV(t) \\[0.1em]
   \{t/y\}(\hatctx{\Psi} \vdash \sigma) & \{t/y\}\hatctx\Psi \vdash \{t/y\}\sigma & \mathrm{provided~} \hatctx\Psi\not\in \FV(t) \\[0.1em]
\multicolumn{3}{l}{\mbox{Computation-level substitution for LF objects}}\\[0.5em]
   \{t/y\}(\lambda x.M) & \lambda x. \{t/y\}M & \\[0.1em]
   \{t/y\}(M~N) & \{t/y\}M ~ \{t/y\}N & \\[0.1em]
   \{t/y\}(\unbox {t'}{\sigma})   & \unbox {\{t/y\}t'}{\{t/y\}\sigma} \\[0.1em]
   \{t/y\}(\const c) & \const c & \\[0.1em]
   \{t/y\}(x)        & x & \\[0.1em]
\multicolumn{3}{l}{\mbox{Computation-level substitution for LF substitutions}}\\[0.5em]
   \{t/y\}(\cdot) & \cdot & \\[0.1em]
   \{t/y\}(\sigma, M) & \{t/y\}\sigma,~\{t/y\}M & \\[0.1em]
   \{t/y\}(\wk{\hatctx\Psi})   & \wk{\{t/y\}({\hatctx\Psi})} & \\[0.1em]
\SUBSTCLO{   \{t/y\}(\sclo {\hatctx\Psi} r) & \sclo {\{t/y\}({\hatctx\Psi})} {\csubclo t y {~}r}}
  \end{array}
\]
  \caption{Computation-level Substitution}\label{fig:csub}
\end{figure}
}

\LONGVERSION{
\begin{figure}[htb]
  \centering
\[
\begin{array}{c}
\multicolumn{1}{p{13.5cm}}{
\fbox{$\Gamma ; \Phi \vdash A : \lftype$}~\mbox{and}~
\fbox{$\Gamma ; \Phi \vdash K : \lfkind$}~~LF type $A$ is well-kinded
  and LF kind $K$ is well-formed}
\\[1em]
\infer{\Gamma; \Psi \vdash \const a: K}{\Gamma \vdash \Psi : \ctx & \const a{:}K \in \Sigma}
\quad
\infer{\Gamma ; \Psi \vdash P~M : [M/x]K}
{\Gamma ; \Psi \vdash P : \Pityp x A K & \Gamma; \Psi \vdash M : A}
\quad
\infer{\Gamma ; \Psi \vdash \Pityp x A B : \lftype}
{\Gamma ; \Psi \vdash A : \lftype & \Gamma ; \Psi, x{:}A \vdash B : \lftype}
\\[1em]
\infer{\Gamma ; \Psi \vdash A : K}
{\Gamma ; \Psi \vdash A : K'  & \Gamma ; \Psi \vdash K' \equiv K : \lfkind}
\quad
\infer{\Gamma ; \Psi \vdash \lftype : \lfkind}{\Gamma \vdash \Psi : \ctx}
\quad
\infer{\Gamma ; \Psi \vdash \Pityp x A K: \lfkind}
{\Gamma ; \Psi \vdash A : \lftype & \Gamma ; \Psi, x{:}A \vdash K : \lfkind}
\end{array}
\]
\caption{Kinding Rules for LF Types}\label{fig:lfkinding}
\end{figure}
}

\begin{figure}[htb]
  \centering\small
\[
\begin{array}{c}
\LONGVERSION{ \multicolumn{1}{p{13cm}}{
 \fbox{$\Gamma ; \Psi \vdash_\# M : A$}~
 ~~LF term $M$ of LF type $A$ in the LF context $\Psi$ and context $\Gamma$ describes a variable}
 \\[0.3em]
\infer{\Gamma ; \Psi\vdash_\# M :  A}
  {\Gamma ; \Psi \vdash M \equiv x : A & \Psi(x) = A
}
\qquad
\infer{\Gamma ; \Psi \vdash_\# M : A}
      {\Gamma ; \Psi \vdash_\# M : B &
       \Gamma ; \Psi \vdash B \equiv A : \lftype }
\\[0.5em]
\infer{\Gamma ; \Psi \vdash_\# M : A}
{\Gamma ; \Psi \vdash M \equiv \unbox{t}{\sigma} : \lfs {\sigma}\Phi (A) &
 \Gamma \vdash t : [\Phi \vdash_\# A] &
 \Gamma; \Psi \vdash_\# \sigma : \Phi }
\\[0.5em]
}
\multicolumn{1}{p{8cm}}{
\fbox{$\Gamma ; \Psi \vdash M : A$}~
~~LF term $M$ has LF type $A$ \newline\mbox{\hspace{2.15cm}in} the LF context $\Psi$ and context $\Gamma$\newline}
\\
\infer{\Gamma ; \Psi \vdash x : A}{
\Gamma \vdash \Psi : \ctx&
x{:}A \in \Psi}
\quad
\infer{\Gamma ; \Psi \vdash \const{c} : A}
      {\Gamma\vdash \Psi : \ctx & \const{c}{:}A \in \Sigma}
\\[0.5em]
\infer{\Gamma ; \Psi \vdash M~N : [N/x]B}
      {\Gamma ; \Psi \vdash M : \Pityp x A B &
       \Gamma ; \Psi \vdash N : A}
~~
 \infer{\Gamma ; \Psi \vdash \lambda x.M : \Pityp x A B}
         {\Gamma ; \Psi, x{:}A \vdash M : B}
\\[0.5em]
\infer{\Gamma ; \Psi \vdash \unbox t {\sigma} : \lfs \sigma \Phi A}
         {\Gamma \vdash t : [\Phi \vdash A]  ~\mbox{or}~\Gamma \vdash t : [\Phi \vdash_\# A] &
          \Gamma; \Psi \vdash \sigma : \Phi}
\\[0.5em]
\qquad
\infer
      {\Gamma ; \Psi \vdash M : A}
      {\Gamma ; \Psi \vdash M : B &
       \Gamma ; \Psi \vdash B \equiv A : \lftype }
%
\\[1em]
\multicolumn{1}{p{8cm}}{
\fbox{$\Gamma ; \Phi \vdash \sigma : \Psi$}~~LF substitution $\sigma$
  provides a mapping \newline\mbox{\hspace{2.15cm}from} the LF context $\Psi$ to $\Phi$\newline}
\\[0.3em]
\infer{\Gamma ;\Psi, \wvec{x{:}A} \vdash \wk{\hatctx\Psi} : \Psi}
{\Gamma\vdash \Psi, \wvec{x{:}A} : \ctx }
\quad
 \infer{\Gamma ; \Phi \vdash \cdot : \cdot}{\Gamma \vdash \Phi :\ctx }
\\[0.5em]
\infer{\Gamma ; \Phi \vdash \sigma, M : \Psi, x{:}A}
      {\Gamma ; \Phi \vdash \sigma : \Psi &
       \Gamma ; \Phi \vdash M : \lfs \sigma \Psi A}
\SUBSTCLO{
\\[0.3em]
\infer{\Gamma ; \Phi \vdash \sclo {\hatctx\Psi} {\unbox{t}{\sigma}} : \Psi}
{\Gamma \vdash t : \cbox{\Psi' \vdash \Psi, \wvec{x{:}A}}
 & \Gamma ; \Phi \vdash \sigma : \Psi'
&  \Gamma\vdash \Psi, \wvec{x{:}A} : \ctx }
\\[0.3em]}
\LONGVERSION{\\[1em]
 \multicolumn{1}{p{13cm}}{
 \fbox{$\Gamma ; \Psi \vdash_\# \sigma : \Phi$}~
 ~~LF substitution $\sigma$ from LF context $\Phi$ to the LF context $\Psi$ is a weakening subst.}
 \\[0.5em]
\infer{\Gamma ; \Psi, \wvec{x{:}A} \vdash_\# \sigma : \Psi}
  {\Gamma ; \Psi, \wvec{x{:}A} \vdash \sigma \equiv \wk{\hatctx\Psi} : \Psi
}
\\[0.3em]
\SUBSTCLO{\infer{\Gamma ; \Psi  \vdash_\# \sigma : \Phi}
{\Gamma ; \Psi \vdash \sigma \equiv \sclo {\hatctx\Phi} {\unbox{t}{\wk{\Phi,\vec{x}}}} : \Phi &
 \Gamma \vdash t : \cbox{\Psi \vdash_\# \Phi, \wvec{x{:}A}} &
 \Gamma ; \Psi \vdash \wk{\hatctx\Psi} : \Phi,\wvec{x{:}A}}
}}
\end{array}
\]
\caption{Typing Rules for LF Terms and LF Substitutions}\label{fig:lftyping}
\end{figure}

\subsection{LF Typing}

We concentrate here on the typing rules for LF terms, LF substitutions and
LF contexts (see Fig.~\ref{fig:lftyping})%
\SHORTVERSION{ and skip the rules for LF types and kinds.}
\LONGVERSION{. The rules for LF types and kinds are straightforward (see Fig.~\ref{fig:lfkinding}).}
All of the typing rules have access to an LF signature $\Sigma$
which we omit to keep the presentation compact.
Typing of variables $x$, constants $\const c$, application $M\,N$ and abstraction $\lambda x.M$ is as usual.
\LONGVERSION{
In typing rules for LF abstractions $\lambda x.M$ we simply extend the LF context and check the body $M$. When we encounter an LF variable, we look up its type in the LF context.}
The conversion rule is important and subtle. We only allow conversion of types -- conversion of the LF context is not necessary, as we do not allow computations to return an LF context.
%
Importantly, given a computation $t$ that has type $\cbox {\Psi \vdash A}$ or $\cbox{\Psi \vdash_\# A}$, we can embed it into the current LF context $\Phi$ by forming the closure $\unbox t {\sigma}$ where $\sigma$ provides a mapping for the variables in $\Psi$. This formulation generalizes previous work which only allowed \emph{variables} declared in $\Gamma$ to be embedded in LF terms.
Previous work enforced a strict separation between computations and LF terms.
%
%
%

The typing rules for LF substitutions are as expected.
\SUBSTCLO{The rule for suspended composition of substitution $\sclo {\hatctx\Phi}{\unbox{t}{\sigma}}$ with domain $\Psi$  requires some care.  Here we ensure that $t$ computes an LF substitution that maps LF variables from $\Psi, \wvec{x{:}A}$ to the LF context $\Psi'$. Hence, we allow $t$ to provide a ``bigger'' substitution than required; these additional instantiations will be truncated when we know what LF substitution the term $t$ computes. The ``stuck'' LF substitution $\sigma$ must map LF variables from $\Psi'$ to the final target LF context $\Phi$. As for LF terms, we distinguish between general LF substitutions and LF renamings (or weakenings) that guarantee that we only map LF variables to LF variables.
}

\SHORTVERSION{The typing rules for LF contexts simply analyze the
  structure of an LF context. When we reach the head, we either
  encounter an empty LF context or a context variable $y$ which must
  be declared in the computation-level context $\Gamma$. The rules can
  be found in the long version. }
\LONGVERSION{Last, we consider the typing rules for LF contexts (see Fig.~\ref{fig:lfctxtyping}). They simply analyze the structure of an LF context. When we reach the head we either encounter a empty LF context or an context variable $y$ which must be declared in the computation-level context $\Gamma$.

\begin{figure}[htb]
  \centering\small
\[
\begin{array}{c}
\multicolumn{1}{l}{
\fbox{$\Gamma \vdash \Psi : \ctx$}~\mbox{LF context $\Psi$ is a well-formed}} \\[0.3em]
\infer{\Gamma \vdash \cdot : \ctx}{\vdash \Gamma}
\quad
\infer{\Gamma \vdash \unboxc{y} : \ctx}{\Gamma(y) = \tmctx & \vdash \Gamma}
\quad
\infer{\Gamma \vdash \Psi, x{:}A : \ctx}
{\Gamma \vdash \Psi : \ctx & \Gamma ; \Psi \vdash A : \lftype}
\end{array}
\]
\caption{Typing Rules for LF Contexts}\label{fig:lfctxtyping}
\end{figure}
}





\subsection{Definitional LF Equality}
For LF terms, equality is $\beta\eta$. In addition, we can reduce $\unbox{\Psi \vdash M}{\sigma}$ by simply applying $\sigma$ to $M$. We omit the transitive closure rules as well as congruence rules, as they are straightforward.

For LF substitutions, we take into account that weakening substitutions are not unique.
 For example, the substitution $\wk\cdot$ may stand for a mapping from the empty context to another LF context; so does the empty  substitution $\cdot$.
Similarly, $\wk{x_1, \ldots x_n}$ is equivalent to the substitution $\wkempty, x_1, \ldots, x_n$.


\begin{figure}[h]
  \centering\small
  \[
  \begin{array}{c}
    \multicolumn{1}{p{8.5cm}}{\fbox{$\Gamma ; \Psi \vdash M \equiv N : A$}\quad LF term $M$ is definitionally equal\newline\mbox{\hspace{3.1cm}to} LF term $N$ at LF type $A$\newline}\\
    \infer{\Gamma ; \Psi \vdash M \equiv \lambda x.M~x : \Pityp x A B}{\Gamma ; \Psi \vdash M : \Pityp x A B}
\\[0.75em]
    \infer{\Gamma ; \Psi \vdash (\lambda x.M_1)~M_2 \equiv [M_2/x]M_1 : [M_2/x]B}
          {\Gamma ; \Psi, x{:}A \vdash M_1 : B & \Gamma ; \Psi \vdash M_2 : A}
    \\[0.75em]
\infer
  {\Gamma ; \Psi \vdash \unbox{\cbox{\hatctx{\Phi} \vdash N}}{\sigma} \equiv \lfs \sigma\Phi N :\lfs \sigma\Phi A}
          {\Gamma ; \Phi \vdash N : A & \Gamma ; \Psi \vdash \sigma : \Phi}
\\[1em]
\multicolumn{1}{p{8.5cm}}{\fbox{$\Gamma ; \Psi \vdash \sigma \equiv \sigma' : \Phi$}~\quad LF substitution $\sigma$ is definitionally equal\newline\mbox{\hspace{3.1cm}to} LF substitution $\sigma'$\newline}\\[-0.75em]
 \infer{\Gamma ; \Psi \vdash \wkempty \equiv \cdot : \cdot}{
 \Gamma \vdash \Psi : \ctx }
 \quad
 \infer{\Gamma ; \Phi, x{:}A, \wvec{y{:}B} \vdash \wk{\hatctx\Phi,x} \equiv (\wk{\Phi}, x) : (\Phi, x{:}A)}
 {
 \Gamma \vdash \Phi, x{:}A, \wvec{y{:}B} : \ctx
 }
\LONGVERSION{\\[1em]
\infer{\Gamma ; \Psi \vdash \sigma, M \equiv \sigma', N : \Phi, x{:}A}
{\Gamma ; \Psi \vdash \sigma \equiv \sigma' : \Phi &
 \Gamma ; \Psi \vdash M \equiv N : \lfs \sigma \Phi A }
}
\end{array}
\]
\caption{Reduction and Expansion for LF Terms and LF Substitutions}
\label{fig:lfeq}
\end{figure}

\subsection{Contextual LF Typing and Definitional Equivalence}

\LONGVERSION{Typing and equivalence of contextual objects (Fig.~\ref{fig:ctxtyping})is standard.}
%
We lift typing and definitional equality on LF terms to
contextual objects. For example, two contextual objects $\hatctx{\Psi} \vdash M$ and $\hatctx{\Psi}\vdash N$ are equivalent at LF type $\cbox{\Psi \vdash A}$, if $M$ and $N$ are equivalent in $\Psi$.

\LONGVERSION{
  \begin{center}
    \begin{small}
\[
  \begin{array}{c}
    \multicolumn{1}{p{8.5cm}}{\fbox{$\Gamma \vdash C : T$}~ Contextual
  Objects $C$ has Contextual Type $T$ in $\Gamma$}
\\[0.5em]
\infer{\Gamma \vdash (\hatctx{\Psi} \vdash M) : (\Psi \vdash A)}
{
  \Gamma ; \Psi \vdash M : A}
  \end{array}
\]
    \end{small}
  \end{center}
}

\subsection{Computation Typing}
We describe well-typed computations in Fig.~\ref{fig:comptyping} using the typing judgment $\Gamma \vdash t : \tau$.
Computations only have access to computation-level variables declared in the context $\Gamma$. We use the judgment $\vdash \Gamma$ to describe well-formed contexts where every declaration $x{:}\ann\tau$ in $\Gamma$ is well-formed.
\LONGVERSION{
\begin{small}
\[
  \begin{array}{l@{~~~}c}
{\mbox{Well-formed context:}~\fbox{$\vdash \Gamma$}}\\[-1.5em]
& \infer{\vdash \cdot}{} \qquad
\infer{\vdash \Gamma, x{:}\ann\tau}{\vdash \Gamma & \Gamma \vdash \ann\tau : u}
  \end{array}
\]
\end{small}
}

\begin{figure}[htb]
  \centering\small
\[
\begin{array}{c}
\multicolumn{1}{l}{\fbox{$\Gamma \vdash t : \tau$}
  ~\mbox{and}~\fbox{$\Gamma \vdash \tau : u$} \mbox{ Typing and kinding judg. for comp.}~~}
\\[0.75em]
\infer{\Gamma \vdash y : \ann\tau}{y:\ann\tau \in \Gamma&\vdash\Gamma}
\quad
\infer[(u_1, u_2) \in \Ax]{\Gamma \vdash u_1: u_2}{\vdash \Gamma}
\\[0.5em]
\infer[(u_1,~u_2,~u_3) \in \Ru]{\Gamma \vdash (y:\ann\tau_1) \arrow \tau_2 : u_3}
      {\Gamma \vdash \ann\tau_1 : u_1 &
       \Gamma, y{:}\ann\tau_1 \vdash \tau_2 : u_2}
\\[0.75em]
\infer{\Gamma \vdash \cbox{T} : u}{\Gamma \vdash T}
\quad
\infer{\Gamma \vdash t~s : \{s/y\}\tau_2}
{\Gamma \vdash t : (y:\ann\tau_1) \arrow \tau_2 &
 \Gamma \vdash s : \ann\tau_1}
\\[0.75em]
\infer{\Gamma \vdash \tmfn y t : (y:\ann\tau_1) \arrow \tau_2}
        {\Gamma, y:\ann\tau_1 \vdash t : \tau_2 & \Gamma \vdash (y:\ann\tau_1) \arrow \tau_2 : u}
\\[0.75em]
\infer{\Gamma \vdash \cbox C : \cbox T}{\Gamma \vdash C : T}
\quad
\infer{\Gamma \vdash t : \tau}
{\Gamma \vdash t : \tau' & \Gamma \vdash \tau' \equiv \tau : u}
\LONGVERSION{\\[1em]
\multicolumn{1}{l}{\mbox{Schema checking of LF context}~~\fbox{$\Gamma \vdash \Psi : \tmctx$}}
\\
\multicolumn{1}{l}{\mbox{;Well-formedness of schema}~\fbox{$\Gamma \vdash \tmctx : u$}}\\[0.75em]
\infer{\Gamma\vdash \tmctx : u}{\vdash \Gamma}
\quad
\infer{\Gamma \vdash \cdot : \tmctx}{\vdash\Gamma}
\\[0.75em]
\infer{\Gamma \vdash \Psi, x{:}A : \tmctx}
{\Gamma \vdash \Psi : \tmctx & \Gamma ; \Psi \vdash A : \lftype & \Gamma ; \Psi \vdash A \equiv \tm : \lftype}
}
\end{array}
\]
  \caption{Typing Rules for Computations (Without Recursor)}
  \label{fig:comptyping}
\end{figure}

To avoid duplication of typing rules, we overload the typing judgment and write $\ann\tau$ instead of $\tau$, if the same judgment is used to check that a given LF context is of schema $\tmctx$. For example, to ensure that $(y : \ann\tau_1) \arrow \tau_2$ has kind $u_3$, we check that $\ann\tau_1$ is well-kinded. For compactness, we abuse notation writing $\Gamma \vdash \tmctx : u$ although the schema $\tmctx$ is not a proper type whose elements can be computed.
In the typing rules for computation-level (extensional) functions, the
input to the function which we also call domain of discourse may
either be of type $\tau_1$ or $\tmctx$. To eliminate a term $t$ of
type $(y : \tau_1) \arrow \tau_2$, we check that $s$ is of type
$\tau_1$ and then return $\{s/y\}\tau_2$ as the type of $t~s$.  To
eliminate a term of type $(y : \tmctx) \arrow \tau$, we overload
application simply writing $t~s$, although $s$ stands for an LF
context and check that $s$ is of schema $\tmctx$. This distinction
between the domains of discourse is important, as we only allow LF
contexts to be built either by a context variable or an LF type
declaration, but do not compute an LF context recursively.
We can embed contextual object $C$ into computations by boxing it and transitioning to the typing rules for LF. We eliminate contextual types using a recursor, see Fig.~\ref{fig:comptypingrec}.
Here, we define an iterator over $t$ of type
$\cbox{\Psi \vdash \tm}$ to keep the exposition compact.
For a deeper discussion on how to generate recursors for contextual objects of type $\Psi \vdash A$ and LF contexts, we refer the reader to \citet{Pientka:TLCA15}.

In general, the output type of the recursor may depend on the argument
we are recursing over. We hence annotate the recursor itself with an
invariant $\IH$. Here, the recursor over $\tm$ is annotated with $\IH = (\psi : \tmctx) \arrow (y:\cbox{\psi \vdash \tm}) \arrow \tau$. To check that the recursor $\titer{\R}{}\IH~\Psi~t$ has type $\{\Psi/\psi, t/y\}\tau$, we check that each of the three branches has the specified type $\IH$. In the base case, we may assume in addition to $\psi:{\tmctx}$ that we have a variable $p:\cbox{\unboxc{\psi} \vdash_\#\tm}$ and check that the body has the appropriate type.
%
%
If we encounter a contextual LF object built with the LF constant $\tapp$, then we choose the branch $b_\tapp$. We assume $\psi{:}\tmctx$, $m{:}\cbox{\unboxc{\psi}\vdash \tm}$, $n{:}\cbox{\unboxc{\psi}\vdash \tm}$, as well as $f_n$ and $f_m$ which stand for the recursive calls on $m$ and $n$ respectively. We then check that the body $t_\tapp$ is well-typed.
If we encounter an LF object built with the LF constant $\tlam$, then we choose the branch $b_\tlam$. We assume $\psi{:}\tmctx$ and $m{:}\cbox{\psi, x{:}\tm \vdash \tm}$ together with the recursive call $f_m$ on $m$ in the extended LF context $\psi, x{:}\tm$. We then check that the body  $t_\tlam$ is well-typed.

 \LONGVERSION{
To check that a given LF context $\Psi$ is of a particular defined schema, we check that every declaration $x{:}A$ in is an instance of the declared schema. We omit here the discussion and refer the reader to \cite{Pientka:PPDP08}.
}

\begin{figure}[htb]
  \centering\small
\[
\begin{array}{c}
\LONGVERSION{
\multicolumn{1}{l}{\mbox{Recursor over LF parameters}~\IH = (\psi : \tmctx) \arrow (q:\cbox{\psi \vdash_\# \tm}) \arrow \tau }\\[0.5em]
\infer[]
{\Gamma \vdash \tmrecctx {\IH} {\psi \mto b_e}{\psi, q, f_q \mto b_c} \rappto \Psi~t: \{\Psi/\psi,~t/y\}\tau}
{
  \begin{array}{lll}
\Gamma \vdash t :  \cbox{\Psi \vdash_\# \tm} & \Gamma \vdash \IH : u & \\
\Gamma \vdash (\psi \mto b_e) : \IH &
\Gamma \vdash (\psi, q, f_q \mto b_c) : \IH
  \end{array} }
\\[1em]
\multicolumn{1}{p{8.5cm}}{\mbox{Branches where} \quad$\IH = (\psi : \tmctx) \arrow (y:\cbox{\psi \vdash_\# \tm}) \arrow \tau$ }
\\[1em]
\infer{\Gamma \vdash (\psi \mto b_e) : \IH}
{\Gamma, \psi:\tmctx \vdash  b_e : \{{(\psi,x{:}\tm)}/\psi,~\cbox{\psi, x \vdash x}/p \}\tau}
\\[0.5em]
\infer{\Gamma \vdash (\psi, q, f_q \mto b_c) : \IH}
{ \Gamma, \psi:\tmctx, q:\cbox{\psi \vdash_\# \tm},  f_q: \{q/p\}\tau  \vdash  b_c :
  \{(\psi,x{:}\tm)/\psi,~\{\cbox{\psi, x \vdash  \unbox{q}{\wk{\psi}}}/p \}\tau }
\\[1em]
}
\multicolumn{1}{l}{\mbox{Recursor over LF terms}~\IH = (\psi : \tmctx) \arrow (y:\cbox{\psi \vdash \tm}) \arrow \tau }\\[0.5em]
\infer
{\Gamma \vdash \tmrec {\IH} {b_v} {b_{\mathsf{app}}} {b_{\tlam}} \rappto \Psi~t : \{{\Psi}/\psi,~t/y\}\tau}
{
  \begin{array}{lll}
\Gamma \vdash t :  \cbox{\Psi \vdash \tm}  & \Gamma \vdash \IH : u  & \\
\Gamma \vdash b_v : \IH & \Gamma \vdash b_{\mathsf{app}} : \IH & \Gamma \vdash b_{\mathsf{lam}} : \IH
  \end{array}
}
\\[1em]
\multicolumn{1}{p{8.5cm}}{\mbox{Branches where} \quad$\IH = (\psi : \tmctx) \arrow (y:\cbox{\psi \vdash \tm}) \arrow \tau$ }
\\[1em]
\inferrule*
{ \Gamma, \psi:\tmctx, p:\cbox{~\unboxc{\psi} \vdash_\# \tm}  \vdash  t_v : \{p / y\}\tau}
{\Gamma \vdash ({\psi,p \mto t_v}) : \IH }
\\[0.5em]
\infer
{\Gamma \vdash (\psi, m, n, f_n, f_m \mto t_{\mathsf{app}}) : \IH}
{\begin{array}{l@{~}c@{~}l}
 \Gamma, \psi:\tmctx, & & \\
         m{:}\cbox{\unboxc{\psi} \vdash \tm}, n{:}\cbox{\unboxc{\psi} \vdash \tm}& & \\
          f_m{:} \{m/y\}\tau, f_n{:} \{n/y\}\tau  & \vdash &  t_{\mathsf{app}} : \{\cbox{\unboxc{\psi} \vdash \tapp\unbox{m}{}~\unbox{n}{}}/y\}\tau
  \end{array}
}
\\[0.5em]
\infer{\Gamma \vdash \psi, m, f_m \mto t_{\tlam} : \IH}
{ \begin{array}{l@{~}c@{~}l}
\Gamma, \phi:\tmctx, & & \\   m{:}\cbox{\phi, x{:}\tm \vdash \tm}, & & \\
         f_m{:}\{(\phi, x{:}\tm)/\psi, m /y \} \tau          & \vdash & t_{\tlam} : \{\phi/\psi, \cbox{\unboxc{\phi} \vdash \tlam~\lambda x.\unbox{m}{}} / y\}\tau
 \end{array}
 }
\end{array}
\]
  \caption{Typing Rules for Recursors}
  \label{fig:comptypingrec}
\end{figure}

\subsection{Definitional Equality for Computations}

Concerning definitional equality for computations
(Fig.~\ref{fig:etype}), we concentrate on the reduction rules.
We omit the transitive closure and congruence rules, as they are as expected.


\begin{figure*}[h]
  \centering
\small
  \[
  \begin{array}{c}
          \infer{\Gamma  \vdash (\tmfn y t)~s \equiv \{s/y\}t : \{s/y\}\tau_2}
                {\Gamma \vdash \tmfn y t : ( y{:}\ann\tau_1) \arrow \tau_2 & \Gamma \vdash s : \ann\tau_1
                }
\qquad
\infer{\Gamma \vdash t \equiv \cbox{\hatctx \Psi \vdash \unbox{t}{\wk{\hatctx\Psi}}} : \cbox{\Psi \vdash A}}{
        \Gamma \vdash t : \cbox{\Psi \vdash A}}
                \\[1em]
\multicolumn{1}{l}{
\mbox{let}~{\R} =  ({\psi,p \mto t_p} \mid {\psi,m,n,f_m, f_n \mto t_{\mathsf{app}}} \mid {\psi, m, f_m \mto t_{\tlam}})
~\mbox{and}~\IH = (\psi : \tmctx) \arrow (y : \cbox{\psi \vdash \tm}) \arrow \tau}
\\[0.75em]
\infer{\Gamma \vdash \titer{\R}{\tm}{\IH}\rappto~\Psi~\cbox{\hatctx{\Psi} \vdash \tlam~\lambda x.M} \equiv \{\theta\}t_{\tlam}
                  :  \{\Psi/\psi, \cbox{\hatctx{\Psi} \vdash \tlam~\lambda x.M}/y\}\tau }
      {\Gamma \vdash \Psi : \tmctx \qquad \Gamma; \Psi, x{:}\tm \vdash M : \tm
       {\qquad\Gamma \vdash \IH : u}
      }
 \\[0.5em]\mbox{where}~
 \theta   =  \Psi/\psi,~
             \cbox{\hatctx{\Psi}, x \vdash M}/m,~
             \titer{\R}{\tm}{\IH}~\rappto {(\Psi,x{:}\tm)}~{\cbox{\hatctx{\Psi}, x \vdash M}}/f
\\[1em]
\infer{\Gamma \vdash \titer{\R}{\tm}{\IH}~\rappto {\Psi} \cbox{\hatctx{\Psi} \vdash \tapp~M~N} \equiv \{\theta\} t_{\tapp}
:  \{\Psi/\psi, \cbox{\hatctx{\Psi} \vdash \tapp~M~N}/y\}\tau }
{
\Gamma  \vdash \Psi : \tmctx  \qquad
\Gamma; \Psi   \vdash  M : \tm  \qquad
\Gamma; \Psi   \vdash  N : \tm \qquad  {\Gamma \vdash \IH : u}
}
\\[0.5em]
\mbox{where}~
\theta  =  \Psi/\psi,~\cbox{\hatctx{\Psi} \vdash M}/m,~
           \cbox{\hatctx{\Psi} \vdash N}/n,~
           \titer{\R}{\tm}{\IH}~\rappto {\Psi}~\cbox{\hatctx{\Psi} \vdash M}/f_m,~
           \titer{\R}{\tm}{\IH}~\rappto {\Psi}~\cbox{\hatctx{\Psi} \vdash N}/f_n

\\[1em]
\infer{\Gamma \vdash \trec{\R}{\tm}{\IH}~\rappto {\Psi}~{\cbox{\hatctx \Psi \vdash x} }
          \equiv \{\Psi/\psi,~\cbox{\hatctx{\Psi} \vdash x}/p\} t_p :  \{\Psi/\psi, \cbox{\Psi \vdash x}/y\}\tau }
      {x{:}\tm \in \Psi  & \Gamma \vdash \Psi : \tmctx
             {\qquad\Gamma \vdash \IH : u}}
      \\[0.75em]
  \end{array}
  \]
  \caption{Definitional Equality for Computations}
  \label{fig:etype}
\end{figure*}

We consider two computations to be equal if they evaluate to the same result. We propagate values through computations and types relying on the computation-level substitution operation. When we apply a term $s$ to a computation $\tmfn y t$, we $\beta$-reduce and replace $y$ in the body $t$ with $s$. We unfold the recursor depending on the value passed. If it is $\cbox{\hatctx{\Psi} \vdash \tlam~ \lambda x.M}$, then we choose the branch $t_{\tlam}$. If the value is $\cbox{\hatctx{\Psi} \vdash \tapp~ M~N}$, we continue with the branch $t_{\tapp}$. If it is ${\cbox{\hatctx{\Psi} \vdash x} }$, i.e. the variable case, we continue with $t_v$. Note that if $\Psi$ is empty, then the case for variables is unreachable, since there is no LF variable of type $\tm$ in the empty LF context and hence the contextual type $\cbox{\cdot \vdash_\# \tm}$ is empty.

We also include the expansion of a computation $t$ at type $\cbox{\Psi \vdash A}$; it is equivalent to unboxing $t$ with the identity substitution and subsequently boxing it, i.e. $t$ is equivalent to $\cbox{\hatctx \Psi \vdash \unbox{t}{\wk{\hatctx{\Psi}}}}$ .


\section{Elementary Properties}
For the LF level, we can establish well-formedness of LF context, LF substitution and weakening properties. In addition, we have LF context conversion and equality conversion for LF types. As usual, we can also prove directly functionality and injectivity of Pi-types for the LF level.
%
\LONGVERSION{
\begin{lemma}[LF Context Conversion]\label{lm:lfctxconv}
Assume $\Gamma \vdash \Psi, x{:}A : \ctx$ and $\Gamma ; \Psi \vdash B : \lftype$.\\\mbox{\quad}
If $\Gamma ; \Psi, x{:}A \vdash \JLF$ and $\Gamma ; \Psi \vdash A \equiv B : \lftype$
then $\Gamma ; \Psi, x{:}B \vdash \JLF$.
\end{lemma}
\begin{proof}
Proof using LF Substitution Lemma. 
\LONGVERSIONCHECKED{
\\[0.5em]
\prf{$\Gamma \vdash \Psi, x{:}A : \ctx$ \hfill by assumption}
\prf{$\Gamma \vdash \Psi : \ctx$ \hfill by inversion}
\prf{$\Gamma \vdash \Psi, x{:}B : \ctx$ \hfill by rule}
\prf{$\Gamma ; \Psi, x{:}B \vdash \wk{\hatctx{\Psi}}: \Psi$ \hfill by rule}
\prf{$\Gamma ; \Psi, x{:}B \vdash x : B$ \hfill by rule}
\prf{$\Gamma ; \Psi \vdash B \equiv A : \lftype$ \hfill by symmetry}
\prf{$\Gamma ; \Psi, x{:}B \vdash x : A$ \hfill by rule}
\prf{$\Gamma ; \Psi, x{:}B \vdash x :\lfs{\wk{\hatctx\Psi}}{\Psi} A$ \hfill as $\lfs{\wk{\hatctx\Psi}}{\Psi} A = A$}
\prf{$\Gamma ; \Psi, x{:}B \vdash \wk{\hatctx{\Psi}}, x : \Psi, x{:}A$ \hfill by rule}
\prf{$\Gamma ; \Psi, x{:}B \vdash \J$ \hfill by Lemma \ref{lm:lfsubst} and $\lfs{\wk{\hatctx\Psi}, x}{\Psi, x} \J = \J$}

}
\end{proof}
}

\begin{lemma}[Functionality of LF Typing]\label{lm:func-lftyping}
Let $\Gamma ; \Psi \vdash \sigma_1 : \Phi$ and $\Gamma ; \Psi \vdash \sigma_2 : \Phi$, and
 $\Gamma ; \Psi \vdash \sigma_1 \equiv \sigma_2 : \Phi$.
\begin{enumerate}
\LONGVERSION{ \item If $\Phi = \Phi_i, x_i{:}A, \wvec{y{:}A}$ and $\Gamma ; \Phi \vdash x_i : A$
      then $\Gamma ; \Psi \vdash \lfs{\sigma_1}\Phi (x_i) \equiv \lfs{\sigma_2}\Phi  (x_i) :
     \lfs{\sigma_1}\Phi (A)$.}
 \item If $\Gamma ; \Phi \vdash \sigma : \Phi'$
 then $\Gamma ; \Psi \vdash \lfs{\sigma_1}{\Phi}\sigma \equiv \lfs{\sigma_2}{\Phi}\sigma : \Phi'$.

\item If $\Gamma ; \Phi \vdash M\,{:}\,A$ then\\
      \mbox{$\Gamma ; \Psi \vdash \lfs{\sigma_1}{\Phi}M \equiv \lfs{\sigma_2}{\Phi}M : \lfs{\sigma_1}{\Phi}A$}.
\LONGVERSION{\item If $\Gamma ; \Phi \vdash A : \lftype$ then
      $\Gamma ; \Psi \vdash \lfs{\sigma_1}{\Phi}A \equiv \lfs{\sigma_2}{\Phi}A : \lftype$.}
\end{enumerate}
\end{lemma}
\begin{proof}
By induction on $\Gamma ; \Phi \vdash M : A$ (resp. $\Gamma ; \Phi \vdash \sigma : \Phi'$\LONGVERSION{\/and  $\Gamma ; \Phi \vdash A : \lftype$}) followed by another inner induction on $\Gamma ; \Psi \vdash \sigma_1 \equiv \sigma_2 : \Phi$ to prove (1).
\LONGVERSIONCHECKED{
\\[1em]
We concentrate first on the variable case (1).

\pcase{$\ianc
 {\Gamma \vdash \Phi_0, x_0{:}A_0, \wvec{y{:}B} : \ctx}
 {\Gamma ; \Phi_0, x_0{:}A_0, \wvec{y{:}B} \vdash \wk{\hatctx\Phi,x} \equiv \wk{\Phi_0}, x_0: \Phi_0, x_0{:}A_0}{}
$}
\prf{Let $x_i \in \hatctx\Phi_0$ and $\Phi_0 = \bullet, x_n, \ldots, x_1$ where $\bullet$ stands for either the empty context or a variable.
Then $\pos {x_i} \lfs{\wk{\bullet, x_n \ldots, x_1}}{x_n \ldots, x_1} = x_i$}
\\
 \prf{\emph{Subcase.} $x_i = x_0$}
 \prf{$\pos {x_i}~\lfs {\wk{\hatctx\Phi,x_0}}{\Phi_0, x_0} = x_0$ \hfill since $x_i{:}A_i \in (\Phi_0, x_0:A_0)$}
\prf{$\pos {x_i}~\lfs { \wk{\Phi_0}, x_0}{\Phi_0, x_0} = x_0$ \hfill by $\pos{}{}$}
\\
\prf{\emph{Subcase.} $x \not= x_0$ and $x_i \in x_n, \ldots, x_1$}
\prf{$\pos {x_i}~\lfs {\wk{\hatctx\Phi,x}}{\Phi_0, x_0} = x_i$ \hfill since $x_i{:}A_i \in (\Phi_0, x_0:A_0)$}
\prf{$\pos {x_i}~\lfs { \wk{\Phi_0}, x_0}{\Phi_0, x_0} = \pos
  {x_i}~\lfs{\wk{\Phi_0}}{\Phi_0} = x_i$ \\ \indent \hfill since  $\pos {x_i}  [\wk{\bullet, x_n \ldots, x_1}/{\bullet, x_n \ldots, x_1}] = x_i$}
\prf{$\Gamma ; \Phi_0,  x_0{:}A_0, \wvec{y{:}B} \vdash x_i \equiv x_i  : A_i$ \hfill
using $A_i = [\wk{\bullet, x_n, \ldots, x_{i-1}} / \bullet, x_n, \ldots, x_{i-1}]A_i$}
\\[1em]
\SUBSTCLO{
\pcase{$\ianc
{\Gamma ; \Phi' \vdash \sigma' : \Phi, \wvec{y{:}A} \qquad \Gamma ; \Psi \vdash \sigma : \Phi' }
{\Gamma ; \Psi \vdash \sclo {\hatctx\Phi} {\unbox{\cbox{\hatctx {\Phi'} \vdash \sigma'}}{\sigma}}  \equiv   \lfs\sigma {\Phi'}~(\lfss {\sigma'}{\hatctx{\Phi},\vec x}  ~\wk\Phi) : \Phi}{}
$
}
\prf{$\pos x~(\sclo {\hatctx\Phi} {\unbox{\cbox{\hatctx {\Phi'}  \vdash \sigma'}}{\sigma}})  =  \lfs{\sigma}{\Phi'}(M) $ where $\trunc_\Psi ~ (\sigma' / \hatctx\Psi) = \sigma''~\mathrm{and}~
    \pos x~\lfs{\sigma''}{\Psi} = M$ }
\prf{$\Gamma ; \Phi' \vdash M : \lfs{\sigma''}{\Psi}(A)$ \hfill by LF Variable Lookup (Lemma \ref{lm:lflookup})}
\prf{$\Gamma ; \Psi \vdash \lfs{\sigma}{\Phi'}(M) : \lfs{\sigma}{\Phi'}(\lfs{\sigma''}{\Psi}(A))$ \hfill by LF subst. lemma}
\prf{$\Gamma ; \Psi \vdash \lfs{\sigma}{\Phi'}(M) : \lfs{\sclo {\hatctx\Phi} {\unbox{\cbox{\hatctx {\Phi'} \vdash \sigma'}}{\sigma}}  }{\Phi} A$ \hfill by LF subst. prop.}
\prf{$\pos~x~ (\lfss {\sigma'}{\hatctx{\Phi},\vec x}  ~\wk\Phi / \hatctx\Phi) = \pos~x~(\sigma''/\hatctx\Phi) = M$ where $\sigma'' =\trunc_\Psi ~ (\sigma' / \hatctx\Phi)$}
\prf{$\Gamma ; \Psi \vdash \lfs{\sclo {\hatctx\Phi} {\unbox{\cbox{\hatctx {\Phi'} \vdash \sigma'}}{\sigma}}}{\Phi}(x)  \equiv   \lfs\sigma {\Phi'}~(\lfs {\lfss {\sigma'}{\hatctx{\Phi},\vec x}  ~\wk\Phi}{\Phi}(x)) : \lfs{\sclo {\hatctx\Phi} {\unbox{\cbox{\hatctx {\Phi'} \vdash \sigma'}}{\sigma}}  }{\Phi} A$}
\\[1em]
\pcase{$\ianc
{r = \unbox{t}{\sigma} \qquad t \neq \cbox{\Phi' \vdash \sigma'}\qquad \Gamma ; \Phi \vdash \sigma : \Phi' \quad \Gamma \vdash t : \cbox{\Phi' \vdash \Psi, x_0{:}A_0,\wvec{x{:}A}}}
{\Gamma ; \Phi \vdash \sclo {\hatctx\Psi,x_0} r \equiv   \sclo{\hatctx\Psi} r, \unbox{\cbox{\hatctx\Psi,x_0 \vdash x_0}} {\sclo {\hatctx\Psi,x_0} r} : \Psi,x_0{:}A_0}{}
$}
\prf{$\pos~x_i~\lfs{\sclo {\hatctx\Psi,x_0} r }{\Psi} = \unbox{\cbox{\hatctx\Psi_i,x_i \vdash x_i}}{\sclo {\hatctx\Psi_i,x_i} r}$ where $\hatctx\Psi = \hatctx\Psi_i,x_i,\ldots,x_1$}
\prf{$\pos~x_i~ \lfs{\sclo{\hatctx\Psi} r, \unbox{\cbox{\hatctx\Psi,x_0 \vdash x_0}} {\sclo {\hatctx\Psi,x_0} r}}{\Psi,x_0}$}
\\
\prf{\emph{Subcase.}$x_i = x_0$}
\prf{$\pos~x_i~\lfs{\sclo {\hatctx\Psi,x_0} r }{\Psi} = \unbox{\cbox{\hatctx\Psi,x_0 \vdash x_0}}{\sclo {\hatctx\Psi,x_0} r}$}
\prf{$\pos~x_i~ \lfs{\sclo{\hatctx\Psi} r, \unbox{\cbox{\hatctx\Psi,x_0 \vdash x_0}} {\sclo {\hatctx\Psi,x_0} r}}{\Psi,x_0} = \unbox{\cbox{\hatctx\Psi,x_0 \vdash x_0}} {\sclo {\hatctx\Psi,x_0} r}$}
\prf{$\Gamma ; \Phi \vdash \lfs{ \sclo {\hatctx\Psi,x_0} r}{\Psi_0,x_0}(x)
     \equiv \lfs{\sclo{\hatctx\Psi} r, \unbox{\cbox{\hatctx\Psi,x_0 \vdash x_0}} {\sclo {\hatctx\Psi,x_0} r}}{\Psi_0, x_0} (x) : \lfs{ \sclo {\hatctx\Psi,x_0} r}{\Psi_0, x_0}(A)$}
\\
\prf{\emph{Subcase.}$x_i \neq x_0$}
\prf{$\pos~x_i~\lfs{\sclo {\hatctx\Psi,x_0} r }{\Psi} = \unbox{\cbox{\hatctx\Psi_i,x_i \vdash x_i}}{\sclo {\hatctx\Psi_i,x_i} r}$ where $\hatctx\Psi = \hatctx\Psi_i,x_i,\ldots,x_1$}
\prf{$\pos~x_i~ \lfs{\sclo{\hatctx\Psi} r, \unbox{\cbox{\hatctx\Psi,x_0 \vdash x_0}} {\sclo {\hatctx\Psi,x_0} r}}{\Psi,x_0} = \pos x_i~\lfs{\sclo{\hatctx\Psi} r}{\Psi} = \unbox{\cbox{\hatctx\Psi_i,x_i \vdash x_i}}{\sclo {\hatctx\Psi_i,x_i} r}
$}
}
\\[1em]
\pcase{$\ianc
{\Gamma ; \Psi \vdash \sigma = \sigma' : \Phi \quad
 \Gamma ; \Psi \vdash M \equiv N : \lfs \sigma \Phi A }
{\Gamma ; \Psi \vdash \sigma, M \equiv \sigma', N : \Phi, y{:}A}{}$
}
\prf{\emph{Subcase.} $x = y$}
\prf{$\pos x~\lfs{\sigma, M}{\Phi, x{:}A} = M$ \hfill by def. of $\pos{}{}$}
\prf{$\pos x~\lfs{\sigma', N}{\Phi, x{:}A} = N$ \hfill by def. of $\pos{}{}$}
\prf{$\Gamma ; \Psi \vdash M \equiv N : \lfs \sigma \Phi A$ \hfill by premise}
\prf{$\Gamma ; \Psi \vdash M \equiv N : \lfs {\sigma, M}{\Phi,x} A$ \hfill since $ \lfs {\sigma, M}{\Phi,x} A = \lfs \sigma \Phi A$ }
\\
\pcase{$\ianc{\Gamma ; \Psi \vdash \sigma' \equiv \sigma : \Phi}
             {\Gamma ; \Psi \vdash \sigma \equiv \sigma' : \Phi}{}$}
\prf{$\Gamma ; \Psi \vdash \lfs{\sigma'}{\Phi}(x) \equiv \lfs{\sigma}{\Phi}(x) : \lfs{\sigma'}{\Phi}A$ \hfill by IH}
\prf{$\Gamma ; \Psi \vdash \lfs{\sigma'}{\Phi}A \equiv \lfs{\sigma}{\Phi} A$ \hfill by IH}
\prf{$\Gamma ; \Psi \vdash \lfs{\sigma}\Phi (x_i) \equiv \lfs{\sigma'}\Phi  (x_i) :
     \lfs{\sigma'}\Phi (A)$ \hfill by type conversion}
\\[1em]
%
}
\end{proof}

\LONGVERSION{
\begin{lemma}[Equality Inversion]\label{lm:lfeqinv}
If $\Gamma ; \Psi \vdash A \equiv \Pi x{:}B_1.B_2 : \lftype$
or $\Gamma ; \Psi \vdash \Pi x{:}B_1.B_2 \equiv A : \lftype$
then $A = \Pi x{:}A_1.A_2$ for some $A_1$ and $A_2$
and $\Gamma ; \Psi \vdash A_1 \equiv B_1 : \lftype$
and $\Gamma ; \Psi, x{:}A_1 \vdash A_2 \equiv B_2 : \lftype$.
\end{lemma}
\begin{proof}
By induction on the definitional equality derivation.
\end{proof}
}
\begin{lemma}[Injectivity of LF Pi-Types]\label{lm:lfpi-inj} $\;$\\
If \mbox{$\Gamma ; \Psi \vdash \Pi x{:}A.B \equiv \Pi x{:}A'.B' : \lftype$} \\ then
\mbox{$\Gamma ; \Psi \vdash A \equiv A' : \lftype$} and
$\Gamma ; \Psi, x{:}A \vdash B \equiv B' : \lftype$.
\end{lemma}
\begin{proof}
By equality inversion\LONGVERSION{~(Lemma \ref{lm:lfeqinv})}.
\end{proof}


For the computation level, we also know that computation context
$\Gamma$ is well-formed; in addition, weakening and substitution
properties hold. However, proving functionality of typing and
injectivity of Pi-types on the computation-level must be postponed.

%


\section{Weak Head Reduction}\label{sec:whred}
The operational semantics of \cocon uses weak head reduction and mirrors declarative equality. It proceeds lazily.
We characterize weak head normal forms (whnf) for both, (contextual) LF and
computations\LONGVERSION{~(\LONGVERSION{Fig.~\ref{fig:lfwhnf} and} Fig.~\ref{fig:whnf})}. They are mutually defined.

\begin{definition}[Whnf of LF]\quad
  \begin{itemize}
  \item An LF term $M$ is in whnf, $\norm M$, iff
$M = \lambda x.N$,  or $M$ is neutral, i.e. $\neut M$, or
$M =  \unbox t \sigma$ and $t$ is neutral (i.e. $\neut t$).

\item An LF term $M$ is neutral, $\neut M$, iff $M$ is of the form $h~M_1\ldots M_n$ where $h$ is either an LF variable or a constant $\const{c}$.
  \end{itemize}
\end{definition}

LF substitutions of the form $\sigma,M$, $\wk{\psi}$ or $\cdot$ are in whnf.
LF types are also always considered to be in whnf, as computation may only produce a contextual LF term, but not a contextual LF type.
Last, (erased) LF contexts are in whnf, as we do not allow computations to return an LF context.

Computation-level expressions are in whnf, if they do not trigger any further computation-level reductions.

\begin{definition}[Whnf of Computations]
\quad
  \begin{itemize}
  \item A term $t$ is in whnf, $\norm t$, if
$t$ is a $(\tmfn y s)$ or $(y : \tau_1) \arrow \tau_2$ or $u$,
$t$ is $\cbox{C}$ or $\cbox{T}$, or $t$ is neutral.

\item A term $t$ is neutral,  $\neut t$, if
$t$ is a variable, $t = s_1~s_2$ where $\neut s_1$, $t = (\trec{\R}{\tm}{\IH}~\rappto \Psi~s)$ where either $\neut s$ or $s = \cbox{\hatctx{\Psi}\vdash \unbox t \sigma}$ and $\neut t$.
  \end{itemize}
\end{definition}

We consider boxed objects $\cbox{C}$ and boxed types $\cbox{T}$ in
whnf, as the contextual object $C$ will be further reduced when we use
them and have to unbox them. The remaining definition of whnf
characterizes terms that do not trigger any further reductions.
We note that weakening preserves whnfs.

We now define weak head reductions
(Fig.~\ref{fig:lfwhnfred} and  Fig.~\ref{fig:whnfred}). If an LF term
is not already in $\mathsf{whnf}$, we have two cases: either we
encounter an LF application $M~N$ and we may need to beta-reduce or
we find an embedded computation $\unbox t \sigma$. If $t$ is neutral, then we are done; otherwise $t$ reduces to a contextual object $\cbox{\hatctx{\Psi} \vdash M}$, and we continue to reduce $[\sigma / \hatctx \Psi]M$.

\begin{figure}[htb]
  \centering\small
  \[
    \begin{array}{c}
\multicolumn{1}{l}{\fbox{$M \lfwhnf N$}:~~\mbox{LF term $M$ weak head
      reduces to $N$ s.t. $\norm N$}}
\\[1em]
\infer{M\;N \lfwhnf R}
{M \lfwhnf \lambda x.M' & [N/x]M' \lfwhnf R}
\qquad
\infer{M~N \lfwhnf R \;N}{
M \lfwhnf R & \neut R}
 \\[0.5em]
\infer{M \lfwhnf M}
{\norm M}
~~
\infer{\unbox t{\sigma} \lfwhnf N}
      {t \whnf \cbox{\hatctx{\Psi} \vdash M}~~
       \lfs {\sigma}{\Psi} M \lfwhnf N}
~~
\infer{\unbox t{\sigma} \lfwhnf \unbox{n}{\sigma}}
      {t \whnf n ~~ \neut {n}}
%
\\[0.5em]
 \multicolumn{1}{l}{\fbox{$\sigma \lfwhnf \sigma'$}:~\mbox{LF subst. $\sigma$ weak head reduces to $\sigma'$ s.t. $\norm \sigma'$}}
 \\[1em]
  \infer{\sigma \lfwhnf \sigma}{\norm \sigma}
\quad
\infer{\wk\cdot \lfwhnf \cdot}
      { }
\SUBSTCLO{\quad
\infer{\sclo {\cdot} {\unbox{t}{\sigma} } \lfwhnf \cdot}{ }
}
\qquad
\infer{\wk{(\hatctx{\Psi},x)} \lfwhnf \wk{\hatctx\Psi}, x}
      {}
\SUBSTCLO{\\[0.5em]
\infer{(\sclo {\psi} {\unbox{t}{\sigma}}) \lfwhnf \sigma''}{
t \whnf \cbox{\hatctx \Psi' \vdash \sigma_1} &
[\sigma / \hatctx {\Psi'}]\sigma_1 \lfwhnf \sigma' &
\trunc_\psi ~ \sigma' = \sigma''
}
\\[0.5em]
\infer{\sclo {\hatctx\Psi,x}{\unbox t \sigma} \lfwhnf \sclo{\hatctx\Psi} \unbox{t}\sigma , \unbox{\cbox{\hatctx\Psi,x \vdash x}}{\sclo {\hatctx\Psi,x} {\unbox t \sigma}}}{}
\qquad
\infer{\sclo {\psi} {\unbox{t}{\sigma} } \lfwhnf \sclo {\psi} {\unbox{n}{\sigma} }}{
t \whnf n & \neut n }
}
    \end{array}
  \]
  \caption{Weak Head Reductions for LF Terms and LF Substitutions}
  \label{fig:lfwhnfred}
\end{figure}

If a computation-level term $t$ is not already in $\norm$, we
have either an application $t_1~t_2$ or a recursor. For an application $t_1~t_2$, we reduce $t_1$. If it reduces to a function, we continue to
beta-reduce, otherwise, we build a neutral application.
For the recursor $\titer{\vec\R}{\tm}{\IH}~\Psi~t$, either
$t$ reduces to a neutral term, then we cannot proceed; or, $t$
reduces to $\cbox{\hatctx{\Psi} \vdash M}$, and then we proceed to
further reduce $M$. If the result is $\unbox {t'}{\sigma}$, where $t'$
is neutral, then we cannot proceed; if the result is $N$ where $N$ is
neutral, then we choose the appropriate branch in $\R$
using the judgment ${\R} \cappto (\Psi)~(\hatctx{\Psi} \vdash N) \whnf v$.
 We note that weak head reduction for LF and computation is deterministic and stable under weakening and LF substitutions.

\begin{figure}[htb]
  \centering\small
  \[
    \begin{array}{c}
\multicolumn{1}{l}{\fbox{$t \whnf r$}:~~\mbox{Term $t$ weak head reduces to $r$ s.t. $\norm r$}}\\[1em]
\infer{t \whnf t}{\norm t}
\quad
 \infer{t_1\;t_2 \whnf v}
          {t_1 \whnf \tmfn y t & \{t_2/y\}t \whnf v}
\quad
 \infer{t_1\;t_2 \whnf w~t_2}
      {t_1 \whnf w & \neut w  }
\\[0.5em]
\infer{\trec{\R}{\tm}{\IH}~\rappto \Psi~t \whnf \trec{\R}{\tm}{\IH}~\rappto \Psi~s}
{
 t \whnf s & \neut{s} }
\\[0.5em]
\infer{\trec{\R}{\tm}{\IH}~\rappto \Psi~t \whnf \titer{\R}{}{\IH}
             \rappto \Psi~(\hatctx{\Psi} \vdash \unbox {t'} \sigma)}
{
 t \whnf \cbox{\hatctx \Psi \vdash M} &
 M   \lfwhnf \unbox {t'} \sigma \quad \neut t'
}
\\[0.5em]
\infer{\titer{\R}{}{\IH}~\rappto \Psi~t \whnf v}
{
 t \whnf \cbox{\hatctx \Psi \vdash M} &
 M   \lfwhnf N \quad \neut N \quad
 {\R} \cappto (\Psi)~(\hatctx{\Psi} \vdash N) \whnf v
}
\\[0.5em]
\multicolumn{1}{l}{\mbox{let}~
\R =  ({\psi,p \mto t_v} \mid {\psi,m,n,f_m, f_n \mto t_{\mathsf{app}}}
      \mid {\psi, m, f_m \mto t_{\tlam}})}
\\[0.5em]
\infer{{\R} \cappto  (\Psi)~(\hatctx{\Psi} \vdash \tapp~M~N)  \whnf  v}
 {
  \begin{array}{@{}l@{}l@{}l@{}l@{}}
  \{ & \Psi/\psi,~
      \cbox{\hatctx{\Psi} \vdash M}/m,~\cbox{\hatctx{\Psi} \vdash N}/n,~~&& \\
     & \trec{\R}{\tm}\IH~\rappto {\Psi}~{\cbox{\hatctx{\Psi} \vdash M}}/f_m,~
      \trec{\R}{\tm}\IH~\rappto {\Psi}~\cbox{\hatctx{\Psi} \vdash N}/f_n~\} & t_{\tapp} &\whnf v
      \end{array}
}
\\[0.5em]
\infer[]
{{\R} \cappto (\Psi)~(\hatctx{\Psi} \vdash \tlam~\lambda x.M)  \whnf  v}
 {
      \begin{array}{@{}l@{}l@{}l@{}l@{}}
  \{ & \Psi/\psi,~~\cbox{\hatctx{\Psi}, x \vdash M}/m,~~&&\\
     &   ~\trec{\R}{\tm}{\IH}\rappto {(\Psi,x{:}\tm)}~{\cbox{\hatctx{\Psi}, x \vdash M}}/f_m  \}& t_{\tlam} \whnf v
      \end{array}
}
\\[1em]
 \infer{{\R} \cappto  (\Psi)~(\hatctx{\Psi} \vdash x) \whnf  v}
       { \{\Psi/\psi,~\cbox{\hatctx{\Psi} \vdash x}/p\} t_v \whnf v}
  \end{array}
  \]
  \caption{Weak Head Reductions for Computations}
  \label{fig:whnfred}
\end{figure}

\LONGVERSION{
\begin{lemma}[Determinacy of Whnf Reduction]\label{lem:detwhnf}\quad
  \begin{enumerate}
  \item If $M \lfwhnf N_1$ and $M \lfwhnf N_2$ then $N_1 = N_2$.
  \item If $\sigma \lfwhnf \sigma_1$ and $\sigma \lfwhnf \sigma_2$ then $\sigma_1 = \sigma_2$.
  \item \label{it:comp-detwhnf} If $t \whnf t_1$ and $t \whnf t_2$ then $t_1 = t_2$.
  \end{enumerate}
\end{lemma}
\begin{proof}
By inspection of the rules.
\end{proof}
}
To ease the technical development,  we introduce notational abbreviations for well-typed whnfs in Fig. \ref{fig:wtwhnf}.

\begin{figure}
  \centering
\[
  \begin{array}{c}
\infer{\Gamma ; \Psi \vdash M \lfwhnf N : A}
{
\begin{array}{@{}ll@{}}
\Gamma ; \Psi \vdash M : A       & \\
\Gamma ; \Psi \vdash N : A       & \Gamma ; \Psi \vdash M \equiv N : A \quad M \lfwhnf N
 \end{array}
}
\\[0.5em]
\infer{\Gamma ; \Psi \vdash \sigma_1 \lfwhnf \sigma_2 : \Phi}
{
\begin{array}{@{}ll@{}}
\Gamma ; \Psi \vdash \sigma_1 : \Phi       & \\
\Gamma ; \Psi \vdash \sigma_2 : \Phi       & \Gamma ; \Psi \vdash \sigma_1 \equiv \sigma_2 : \Phi \quad \sigma_1 \lfwhnf \sigma_2
 \end{array}
}
\\[0.5em]
\infer{\Gamma \vdash t \whnf t' : \tau}
{
\Gamma \vdash t : \tau       \quad
\Gamma \vdash t' : \tau      \quad
\Gamma \vdash t \equiv t' : \tau
\quad t \whnf t'
}
  \end{array}
\]
  \caption{Well-Typed Whnf}
  \label{fig:wtwhnf}
\end{figure}

\LONGVERSION{
\begin{lemma}[Weak Head Reductions Preserved Under Weakening]\label{lem:weakwhnf}\quad
  \begin{enumerate}
  \item If\/ $\Gamma ; \Psi \vdash M \lfwhnf N: A$ and $\Gamma' \leq_\rho \Gamma$
    then $\Gamma' ; \{\rho\}\Psi \vdash \{\rho\}M \lfwhnf \{\rho\}N : \{\rho\}A$.
  \item If\/ $\Gamma ; \Psi \vdash \sigma \lfwhnf \sigma': \Phi$ and $\Gamma' \leq_\rho \Gamma$
    then $\Gamma' ; \{\rho\}\Psi \vdash \{\rho\}\sigma \lfwhnf \{\rho\}\sigma' : \{\rho\}\Phi$.
  \item \label{it:sweakcomp} If\/ $\Gamma \vdash t \whnf t': \tau$ and $\Gamma' \leq_\rho \Gamma$
    then $\Gamma' \vdash \{\rho\}t \whnf \{\rho\}t' : \{\rho\}\tau$.

  \end{enumerate}

\end{lemma}
\begin{proof}
By mutual induction on the first derivation using the computation-level substitution lemma \ref{lm:compsubst}, as renaming $\Gamma' \leq_\rho \Gamma$ are a special case of computation-level substitutions.
\end{proof}

 \begin{lemma}[LF Weak Head Reduction Is Stable Under LF Substitutions]\label{lm:lfwhnfsub}
 Let $\Gamma ; \Psi \vdash \sigma : \Phi$.
 \begin{enumerate}
 \item If $\Gamma ; \Phi \vdash M \lfwhnf \unbox {t_1} {\sigma_1} : A$
       then $\Gamma ; \Psi \vdash \lfs \sigma \Phi M \lfwhnf \unbox{t_1}{\lfs \sigma \Phi {\sigma_1}} : \lfs \sigma \Phi A$.
 \item If $\Gamma ; \Phi \vdash M \lfwhnf \lambda x.N : \Pi x{:}A.B$
       then $\Gamma ; \Psi \vdash \lfs \sigma \Phi M \lfwhnf \lfs \sigma \Phi (\lambda x.N) : \lfs \sigma \Phi (\Pi x{:}A.B)$.
 \item If $\Gamma ; \Phi \vdash M \lfwhnf x : A$ and $\Gamma ; \Psi \vdash \sigma(x) \lfwhnf N : \lfs\sigma\Phi A$
       then $\Gamma ; \Psi \vdash \lfs \sigma \Phi M \lfwhnf N : \lfs \sigma \Phi A$.
 \item If $\Gamma ; \Phi \vdash M \lfwhnf \tapp~M_1~M_2 : \tm$
       then $\Gamma ; \Psi \vdash \lfs\sigma\Phi M \lfwhnf \lfs\sigma\Phi (\tapp~M_1~M_2) : \tm$.
 \item If $\Gamma ; \Phi \vdash M \lfwhnf \tlam~M_1 : \tm$
       then $\Gamma ; \Psi \vdash \lfs\sigma\Phi M \lfwhnf \lfs\sigma\Phi (\tlam~M_1) : \tm$.
 \item If $\Gamma ; \Phi \vdash \sigma_1 \lfwhnf \sigma_2 : \Phi'$
       then $\Gamma ; \Psi \vdash \lfs\sigma\Phi\sigma_1 \lfwhnf \lfs\sigma\Phi\sigma_2 : \Phi'$.
 \end{enumerate}
\end{lemma}
\begin{proof}
By induction on $M \lfwhnf M'$ relation that is part of the well-typed weak head reduction using Lemma \ref{lm:lfctxwf} and \ref{lm:lfsubst}.
\LONGVERSIONCHECKED{\\[1em]
For (1): \fbox{If $\Gamma ; \Phi \vdash M \lfwhnf \unbox {t_1} {\sigma_1} : A$
       then $\Gamma ; \Psi \vdash \lfs \sigma \Phi M \lfwhnf \unbox{t_1}{\lfs \sigma \Phi {\sigma_1}} : \lfs \sigma \Phi A$.}
\\[1em]
\prf{\emph{Case} $M = \unbox{t_0}{\sigma_1}$ and $t_0 \whnf t_1$ and $\neut t_1$}
\prf{$\Gamma ; \Phi \vdash M : A$ \hfill by assumption}
\prf{$\Gamma ; \Phi \vdash \unbox{t_1}{\sigma_1} : A$ \hfill by assumption}
\prf{$\Gamma ; \Psi \vdash \lfs\sigma\Phi M : \lfs\sigma\Phi A$ \hfill by LF subst. lemma \ref{lm:lfsubst}}
\prf{$\Gamma ; \Psi \vdash  \lfs\sigma\Phi (\unbox{t_1}{\sigma_1}) :  \lfs\sigma\Phi A$ \hfill by LF subst. lemma \ref{lm:lfsubst}}
\prf{$ \lfs\sigma\Phi M \lfwhnf \unbox{t_1}{\lfs\sigma\Phi {\sigma_1}}$ \hfill since $\neut t_1$ and $t_0 \whnf t_1$ and LF subst. prop.}
\prf{$\Gamma ; \Phi \vdash \lfs\sigma\Phi M \lfwhnf \unbox{t_1}{\lfs\sigma\Phi {\sigma_1}} : \lfs\sigma\Phi A$ \hfill by well-typed whnf (Def \ref{def:typedwhnf})}
\\[0.25em]
\prf{\emph{Case} $M = \unbox{t_0}{\sigma_0}$ and $t_0 \whnf \cbox{\hatctx\Phi' \vdash M'}$, $\lfs{\sigma_0}{\Phi'}M' \lfwhnf \unbox{t_1}{\sigma_1}$}
\prf{$\Gamma ; \Phi \vdash M : A$ \hfill by assumption}
\prf{$\Gamma ; \Phi \vdash \unbox{t_1}{\sigma_1} : A$ \hfill by assumption}
\prf{$\Gamma ; \Psi \vdash \lfs\sigma\Phi M : \lfs\sigma\Phi A$ \hfill by LF subst. lemma \ref{lm:lfsubst}}
\prf{$\Gamma ; \Psi \vdash  \lfs\sigma\Phi (\unbox{t_1}{\sigma_1}) :  \lfs\sigma\Phi A$ \hfill by LF subst. lemma \ref{lm:lfsubst}}
\prf{$\lfs{\lfs\sigma\Phi{\sigma_0}}{\Phi'}M' \lfwhnf \unbox{t_1}{\lfs\sigma\Phi {\sigma_1}}$ \hfill by IH (and subst. prop)}
\prf{$\unbox{t_0}{\lfs\sigma\Phi \sigma_0} \lfwhnf \unbox{t_1}{\lfs\sigma\Phi{\sigma_1}}$ \hfill by whnf}
\prf{$\Gamma ; \Psi \vdash \lfs\sigma\Phi (\unbox{t_0}{\sigma_0}) \lfwhnf \unbox{t_1}{\lfs\sigma\Phi{\sigma_1}} : \lfs\sigma\Phi A$ \hfill by subst. prop. and  well-typed whnf (Def \ref{def:typedwhnf})}
\\[0.25em]
\prf{\emph{Case} $M = \unbox{t_1}{\sigma_1}$ and $\neut t_1$}
\prf{$\Gamma ; \Phi \vdash M : A$ \hfill by assumption}
\prf{$\Gamma ; \Phi \vdash \unbox{t_1}{\sigma_1} : A$ \hfill by assumption}
\prf{$\Gamma ; \Psi \vdash \lfs\sigma\Phi M : \lfs\sigma\Phi A$ \hfill by LF subst. lemma \ref{lm:lfsubst}}
\prf{$\Gamma ; \Psi \vdash  \lfs\sigma\Phi (\unbox{t_1}{\sigma_1}) :  \lfs\sigma\Phi A$ \hfill by LF subst. lemma \ref{lm:lfsubst}}
\prf{$\lfs\sigma\Phi(\unbox{t_1}{\sigma_1}) = \unbox{t_1}{\lfs\sigma\Phi{\sigma_1}}$ \hfill by subst. def.}
\prf{$\neut \unbox{t_1}{\lfs\sigma\Phi{\sigma_1}}$ \hfill since $\neut t_1$}
\prf{$\lfs\sigma\Phi M \lfwhnf \unbox{t_1}{\lfs\sigma\Phi {\sigma_1}}$ \hfill by whnf}
\prf{$\Gamma ; \Psi \vdash \lfs\sigma\Phi M \lfwhnf \unbox{t_1}{\lfs\sigma\Phi {\sigma_1}} : \lfs\sigma\Phi A$ \hfill by well-typed whnf (Def \ref{def:typedwhnf})}
\\[-0.75em]
\prf{\emph{Case} $M = M_1~M_2$ and $M \lfwhnf \unbox{t_1}{\sigma_1}$}
\prf{$\Gamma ; \Phi \vdash M : A$ \hfill by assumption}
\prf{$\Gamma ; \Phi \vdash \unbox{t_1}{\sigma_1} : A$ \hfill by assumption}
\prf{$\Gamma ; \Psi \vdash \lfs\sigma\Phi M : \lfs\sigma\Phi A$ \hfill by LF subst. lemma \ref{lm:lfsubst}}
\prf{$\Gamma ; \Psi \vdash  \lfs\sigma\Phi (\unbox{t_1}{\sigma_1}) :  \lfs\sigma\Phi A$ \hfill by LF subst. lemma \ref{lm:lfsubst}}
\prf{$M_1 \lfwhnf \lambda x.M'$ and $[M_2/x]M' \lfwhnf \unbox{t_1}{\sigma_1}$ \hfill by inversion}
\prf{$\lfs\sigma\Phi {[M_2/x]M'} \lfwhnf \lfs\sigma\Phi (\unbox{t_1}{\sigma_1})$ \hfill by IH}
\prf{$\lfs{\sigma, \lfs\sigma\Phi M_2}{\Phi, x}M' \lfwhnf \lfs\sigma\Phi(\unbox{t_1}{\sigma_1})$ \hfill by subst. def}
\prf{$\lfs\sigma\Phi M_1 \lfwhnf \lfs\sigma\Phi (\lambda x.M')$ \hfill by IH}
\prf{$\lfs\sigma\Phi{M_1} \lfwhnf \lambda x. \lfs{\sigma, x}{\Phi,x}M'$ \hfill by subst. def.}
\prf{$[\lfs\sigma\Phi M_2/x](\sigma, x) = \sigma, \lfs\sigma\Phi M_2$ \hfill by subst. def.}
\prf{$\lfs\sigma\Phi M_1~\lfs\sigma\Phi M_2 \lfwhnf \lfs\sigma\Phi (\unbox{t_1}{\sigma_1})$ \hfill by whnf rules}
\prf{$\lfs\sigma\Phi(M_1~M_2) \lfwhnf \unbox{t_1}{\lfs\sigma\Phi \sigma_1}$ \hfill by subst. def.}
}
\LONGVERSIONCHECKED{\\[1em]
For (3): \fbox{If $\Gamma ; \Phi \vdash M \lfwhnf x : A$ and $\Gamma ; \Psi \vdash \sigma(x) \lfwhnf N : \lfs\sigma\Phi A$
       then $\Gamma ; \Psi \vdash \lfs \sigma \Phi M \lfwhnf N : \lfs \sigma \Phi A$.}
\\[1em]
\\[0.5em]
\prf{\emph{Case} $M = x$ where $x \in \Phi$ and $\norm M$}
\prf{$\Gamma ; \Phi \vdash x : A$ and $x:A \in \Phi$ \hfill since $\Gamma ; \Phi \vdash M \lfwhnf x : A$}
\prf{$\Gamma ; \Phi \vdash M : A$ \hfill since $\Gamma ; \Phi \vdash M \lfwhnf x : A$}
\prf{$\Gamma ; \Psi \vdash \lfs\sigma\Phi M : \lfs\sigma\Phi A$ \hfill LF subst. lemma \ref{lm:lfsubst}}
\prf{$\Gamma ; \Psi \vdash N : \lfs\sigma\Phi A$ \hfill since $\Gamma ; \Psi \vdash \sigma : \Phi$ and $x:A \in \Phi$}
\prf{$\lfs\sigma\Phi M = \lfs\sigma\Phi x = \sigma(x)$ \hfill by subst. def.}
\prf{$\lfs\sigma\Phi M \lfwhnf N$ \hfill since $\sigma(x) \lfwhnf N$}
\\[-0.75em]
\prf{\emph{Case} $M = M_1~M_2$ and $M \lfwhnf x$}
\prf{$\Gamma ; \Phi \vdash x : A$ and $x:A \in \Phi$ \hfill since $\Gamma ; \Phi \vdash M \lfwhnf x : A$}
\prf{$\Gamma ; \Phi \vdash M : A$ \hfill since $\Gamma ; \Phi \vdash M \lfwhnf x : A$}
\prf{$\Gamma ; \Psi \vdash \lfs\sigma\Phi M : \lfs\sigma\Phi A$ \hfill LF subst. lemma \ref{lm:lfsubst}}
\prf{$M_1 \lfwhnf \lambda x.M'$ and $[M_2/x]M' \lfwhnf x$ \hfill by inversion}
\prf{$\lfs\sigma\Phi [M_2/x]M' \lfwhnf N$ \hfill by IH using $\sigma(x) \lfwhnf N$}
\prf{$\lfs{\sigma,\lfs{\sigma}{\Phi}M_2}{\Phi, x}M' \lfwhnf N$ \hfill by subst. def.}
\prf{$\lfs\sigma\Phi M_1 \lfwhnf \lambda x.[\sigma, x]M'$ \hfill by IH and LF subst. prop.}
\prf{$\lfs\sigma\Phi M \lfwhnf N$ \hfill by whnf rules}
\prf{$\Gamma ; \Psi \vdash \lfs\sigma\Phi M \lfwhnf N : \lfs\sigma \Phi A$ \hfill by  well-typed whnf (Def \ref{def:typedwhnf})}
\\[0.25em]
\prf{\emph{Case} $M = \unbox{t_1}{\sigma_1}$ and $M \lfwhnf x$}
\prf{$\Gamma ; \Phi \vdash x : A$ and $x:A \in \Phi$ \hfill since $\Gamma ; \Phi \vdash M \lfwhnf x : A$}
\prf{$\Gamma ; \Phi \vdash M : A$ \hfill since $\Gamma ; \Phi \vdash M \lfwhnf x : A$}
\prf{$\Gamma ; \Psi \vdash \lfs\sigma\Phi M : \lfs\sigma\Phi A$ \hfill LF subst. lemma \ref{lm:lfsubst}}
\prf{$t_1 \whnf \cbox{\hatctx\Phi' \vdash M'}$ \hfill since $\Gamma ; \Phi \vdash M \lfwhnf x : \tm$}
\prf{$\sigma_1 \lfwhnf \sigma_2$ and $\lfs{\sigma_2}{\Phi'}M' \lfwhnf x$ \hfill since $\Gamma ; \Phi \vdash M \lfwhnf x : \tm$}\\[-0.85em]
\prf{$\lfs\sigma\Phi \sigma_1 \lfwhnf \lfs\sigma\Phi \sigma_2$ \hfill by IH}\\[-0.85em]
\prf{$\lfs\sigma\Phi (\lfs{\sigma_2}{\Phi'}M' ) \lfwhnf N$ \hfill by IH}\\[-0.85em]
\prf{$\lfs{\lfs\sigma \Phi \sigma_2}{\Phi'} M'\lfwhnf N$ \hfill by subst. prop.}\\[-0.85em]
\prf{$\unbox{t_1}{\lfs\sigma \Phi \sigma_1} \lfwhnf N$ \hfill by whnf rules}\\[-0.85em]
\prf{$\Gamma ; \Psi \vdash \unbox{t_1}{\lfs\sigma \Phi \sigma_1} \lfwhnf N : \tm$ \hfill by  well-typed whnf (Def \ref{def:typedwhnf})}
}
\LONGVERSIONCHECKED{
\\[1em]
For (6): \fbox{If $\Gamma ; \Phi \vdash \sigma_1 \lfwhnf \sigma_2 : \Phi'$
       then $\Gamma ; \Psi \vdash \lfs\sigma\Phi\sigma_1 \lfwhnf \lfs\sigma\Phi\sigma_2 : \Phi'$.
}
\\[1em]
\pcase{$\norm \sigma_1$ and $\sigma_2 = \sigma_1$}
\prf{$\norm \lfs\sigma\Phi \sigma_1$ \hfill by def. of $\norm$}
\prf{$\lfs\sigma\Phi \sigma_1 \lfwhnf \lfs\sigma\Phi \sigma_1$ \hfill by whnf}
\prf{$\Gamma ; \Phi \vdash \sigma_1 : \Phi'$ and $\Gamma ; \Phi \vdash \sigma_2 : \Phi'$ \hfill by assumption}
\prf{$\Gamma ; \Psi \vdash \lfs\sigma\Phi\sigma_1 : \Phi$ and $\Gamma ; \Psi \vdash \lfs\sigma\Phi\sigma_2 : \Phi$     \hfill by LF subst. lemma}
\prf{$\Gamma ; \Psi \vdash \lfs\sigma\Phi\sigma_1 \lfwhnf \lfs\sigma\Phi\sigma_1 : \Phi'$  \hfill by well-typed whnf (Def \ref{def:typedwhnf})}
\\
\pcase{$\sigma_1 = \wk\cdot$}
\prf{$\Gamma ; \Phi \vdash \wk\cdot : \cdot$ \hfill by assumption}
\prf{$\Gamma \vdash \Phi : \ctx$ and $\Gamma \vdash \Psi : \ctx$ \hfill by well-formedness of LF context (Lemma \ref{lm:lfctxwf})}
\prf{$\Gamma ; \Phi \vdash \cdot : \cdot$ \hfill by typing rule}
\prf{$\cdot = \lfs\sigma\Phi{\wk\cdot} = \trunc_\cdot (\sigma / \hatctx\Phi)$ \hfill by subst. def.}
\prf{$\Gamma ; \Psi \vdash \cdot : \cdot$ \hfill by typing rule}
\prf{$\norm \cdot$ \hfill by whnf}
\prf{$\cdot \lfwhnf \cdot$ \hfill by whnf}
\prf{$\lfs\sigma\Phi{\wk\cdot} \lfwhnf \lfs\sigma\Phi\cdot$ \hfill since $\cdot = \lfs\sigma\Phi{\wk\cdot}$ and $\lfs\sigma\Phi\cdot = \cdot$}
\prf{$\Gamma ; \Psi \vdash \lfs\sigma\Phi \wk\cdot \lfwhnf \lfs\sigma\Phi\cdot : \cdot$ \hfill  by well-typed whnf (Def \ref{def:typedwhnf})}
\\[1em]
\pcase{$\sigma_1 = \wk{\hatctx\Phi',x}$ where $\Phi = \Phi', x{:}A, \wvec{x{:}A}$ }
\prf{$\sigma_2 = \wk{\hatctx\Phi'}, x$ \hfill by assumption}
\prf{$\Gamma ; \Phi \vdash \wk{\hatctx\Phi',x} : \Phi', x{:}A$ and $\Gamma ; \Phi \vdash \wk{\hatctx\Phi'} : \Phi'$ \hfill by assumption and typing}
\prf{$\Gamma ; \Psi \vdash \sigma : \Phi', x{:}A, \wvec{x{:}A}$ \hfill by assumption\\}
\prf{\emph{Sub-case} : $\sigma = (\sigma', M, \vec{M})$}
\prf{$\Gamma ; \Psi \vdash (\sigma', M, \vec{M}) :  \Phi', x{:}A, \wvec{x{:}A}$ \hfill where $\sigma = (\sigma', M, \vec{M})$}
\prf{$\Gamma ; \Psi \vdash \sigma' : \Phi'$ and $\Gamma ; \Psi \vdash M : \lfs{\sigma'}{\Phi'}A$ \hfill by inversion}
\prf{$ \lfs\sigma\Phi (\wk{\hatctx\Phi',x}) = \trunc_{\hatctx\Phi',x} (\sigma / \Phi) = \sigma', M$ \hfill by definition}
\prf{$\lfs\sigma\Phi (\wk{\hatctx\Phi'}) = \trunc_{\hatctx\Phi'}(\sigma / \Phi) = \sigma'$ \hfill by definition}
\prf{$\Gamma ; \Psi \vdash \sigma', M \lfwhnf \sigma', M : \Phi', x{:}A$ \hfill since $\neut (\sigma', M)$}
\prf{$\Gamma ; \Psi \vdash \lfs\sigma\Phi (\wk{\hatctx\Phi',x}) \lfwhnf \lfs\sigma\Phi (\wk{\hatctx{\Phi'}}, x) : \Phi', x{:}A$ \hfill  by well-typed whnf (Def \ref{def:typedwhnf})\\}
\prf{\emph{Sub-case} : $\sigma = \wk{\hatctx\Phi', x, \vec x}$}
\prf{$ \lfs\sigma\Phi (\wk{\hatctx\Phi',x}) = \trunc_{\hatctx\Phi',x} (\sigma / \Phi) = \trunc_{\hatctx\Phi', x}(\wk{\hatctx\Phi', x, \vec x} / \hatctx\Phi',x,\vec x) = \wk{\hatctx\Phi',x}$ \hfill by definition}
\prf{$\lfs\sigma\Phi (\wk{\hatctx\Phi'},x) = \trunc_{\hatctx\Phi'} (\sigma / \Phi) = \wk{\hatctx\Phi'}$ \hfill by definition}
\prf{$\lfs\sigma\Phi (x) = x$ \hfill since $\sigma = \wk{\hatctx\Phi', x, \vec x}$}
\prf{$\Gamma ; \Phi', x{:}A, \wvec{x:A} \vdash \wk{\hatctx\Phi',x} \lfwhnf \wk{\hatctx\Phi'}, x : \Phi', x{:}A$ \hfill by $\lfwhnf$ rule and  by well-typed whnf (Def \ref{def:typedwhnf})}
\SUBSTCLO{\\[1em]
\prf{\emph{Sub-case} : $\sigma = \sclo {\hatctx\Phi',x,\vec x} \unbox{t}{\sigma'}$  }
\prf{$\Gamma ; \Psi \vdash \sclo {\hatctx\Phi',x,\vec x} \unbox{t}{\sigma'} : \Phi', x{:}A, \wvec{x{:}A}$ \hfill by assumption}
\prf{$ \lfs\sigma\Phi \wk{\hatctx\Phi',x} = \trunc_{\hatctx\Phi',x} (\sigma / \Phi) = \sclo {\hatctx\Phi',x}{\unbox{t}{\sigma'}}$ \hfill by definition}
\prf{$\lfs\sigma\Phi \wk{\hatctx\Phi'} = \trunc_{\hatctx\Phi'} (\sigma / \Phi) = \sclo {\hatctx\Phi'}{\unbox{t}{\sigma'}}$ \hfill by definition}
\prf{$\lfs\sigma\Phi (x) = \unbox{\cbox{\hatctx\Phi \vdash x}}{\sigma}$ \hfill by def.}
\prf{$\Gamma ; \Phi', x{:}A, \wvec{x{:}A} \vdash \sclo {\hatctx\Phi',x}{\unbox{t}{\sigma'}} \lfwhnf \sclo {\hatctx\Phi'}{\unbox{t}{\sigma'}},~\unbox{\cbox{\hatctx\Phi \vdash x}}{\sigma} : \Phi', x{:}A$ \hfill by rules $\lfwhnf$}
\prf{$\Gamma ; \Phi', x{:}A, \wvec{x:A} \vdash \lfs\sigma\Phi (\wk{\hatctx\Phi',x}) \lfwhnf \lfs\sigma\Phi (\wk{\hatctx{\Phi'}}, x) : \Phi', x{:}A$ \hfill by well-typed whnf (Def \ref{def:typedwhnf}) }
}
\SUBSTCLO{
\\[1em]
\pcase{$\sigma_1 = (\sclo {\psi} {\unbox{t}{\xi}})$ and $t \whnf \cbox{\psi,\vec{w} \vdash \xi_1}$ and\\
$\lfs {\xi}{\psi, \vec{w}} \xi_1 \lfwhnf \xi'$ and
$\lfs{\xi'}{\psi,\vec{w}}(\wk\psi) = \trunc_\psi ~ \xi' = \sigma_2$}
\prf{$\Gamma ; \lfs\sigma\Phi \lfs {\xi}{\psi, \vec{w}} \xi_1 \lfwhnf \lfs\sigma\Phi\xi'$ \hfill by IH}
\prf{$\Gamma ; \Phi', x{:}A \wvec{x:A} \vdash \lfs{\lfs\sigma\Phi \xi}{\psi,\vec{w}} \xi_1 \lfwhnf \lfs\sigma\Phi\xi'$ \hfill by subst. prop.}
\prf{$\Gamma ; \Phi', x{:}A, \wvec{x:A} \vdash (\sclo {\psi} {\unbox{t}{\lfs\sigma\Phi \xi}})\lfwhnf \lfs{\lfs\sigma\Phi\xi'}{\psi,\vec{w}}(\wk\psi) $ \hfill by $\lfwhnf$ and \\
\mbox{$\quad$}\hfill well-typed whnf (Def \ref{def:typedwhnf})}
\prf{$\Gamma ; \Phi', x{:}A, \wvec{x:A} \vdash \lfs\sigma\Phi (\sclo {\psi} {\unbox{t}{\xi}}) \lfwhnf \lfs\sigma\Phi (\sigma_2) : \Phi', x{:}A$ \hfill by subst. prop.
}
\\[1em]
\pcase{$\sigma_1 = \sclo {\phi} {\unbox t \xi} $ and $t \whnf t'$ and $\neut t'$ and $\sigma_2 = \sclo{\phi}\unbox{t'}{\xi}$}
\prf{$\Gamma ; \Psi \vdash \sclo {\phi}{\unbox t {\lfs\sigma\Phi \xi}} \lfwhnf \sclo{\phi}{\unbox{t'}{\lfs\sigma\Phi \xi}}$\hfill since $\neut t'$\\
\mbox{$\quad$}\hfill and  by well-typed whnf (Def \ref{def:typedwhnf})}
\prf{$\Gamma ; \Psi \vdash \lfs \sigma\Phi ( \sclo {\phi} {\unbox{t}{\xi} }) \lfwhnf \lfs\sigma\Phi (\sclo{\phi}{\unbox{t'}{\xi}})$ \hfill by subst. prop.}
}
\\[1em]
}
\LONGVERSIONCHECKED{
\fbox{For (4): If $\Gamma ; \Phi \vdash M \lfwhnf \tapp~M_1~M_2 : \tm$
       then $\Gamma ; \Psi \vdash \lfs\sigma\Phi M \lfwhnf \lfs\sigma\Phi (\tapp~M_1~M_2) : \tm$.}
\\[1em]
\pcase{$\ianc
{\norm (\tapp~M_1~M_2)}
{\tapp~M_1~M_2 \lfwhnf \tapp~M_1~M_2}{}
$
}
\prf{$\Gamma ; \Phi \vdash \tapp~M_1~M_2 : \tm$ \hfill by assumption}
\prf{$\Gamma ; \Psi \vdash \lfs\sigma\Phi (\tapp~M_1~M_2) : \tm$ \hfill by LF subst. lemma}
\prf{$\Gamma ; \Psi \vdash \tapp~\lfs\sigma\Phi(M_1)~\lfs\sigma\Phi(M_2) : \tm$ \hfill subst. prop}
\prf{$\norm (\tapp~\lfs\sigma\Phi(M_1)~\lfs\sigma\Phi(M_2))$ \hfill by whnf def.}
\prf{$\Gamma ; \Psi \vdash \lfs\sigma\Phi (\tapp~M_1~M_2) \lfwhnf \lfs\sigma\Phi (\tapp~M_1~M_2) : \tm$ \hfill by subst. prop., $\lfwhnf$ rule, and Def. \ref{def:typedwhnf}}
\\[1em]
\pcase{$\ianc
{M \lfwhnf \lambda x.M' \quad [N/x]M' \lfwhnf \tapp~M_1~M_2 }
{M\;N \lfwhnf \tapp~M_1~M_2}{}
$
}
\prf{$\Gamma ; \Phi \vdash \tapp~M_1~M_2 : \tm$ \hfill by assumption}
\prf{$\Gamma ; \Psi \vdash  \lfs\sigma\Phi (\tapp~M_1~M_2) : \tm$ \hfill by LF subst. lemma}
\prf{$\Gamma ; \Phi \vdash M~N : \tm$ \hfill by assumption}
\prf{$\Gamma ; \Psi \vdash \lfs\sigma\Phi (M~N) : \tm$ \hfill by LF subst. lemma}
\prf{$\Gamma ; \Psi \vdash \lfs\sigma\Phi M \lfwhnf \lfs\sigma\Phi (\lambda x.M') : \tm$ \hfill by IH}
\prf{$\Gamma ; \Psi \vdash \lfs\sigma\Phi([N/x]M') \lfwhnf \lfs\sigma\Phi (\tapp~M_1~M_2) : \tm$ \hfill by IH}
\prf{$\Gamma ; \Psi \vdash \lfs\sigma\Phi (M~N) \lfwhnf \lfs\sigma\Phi (\tapp~M_1~M_2) : \tm$ \hfill
with $\lfs\sigma\Phi([N/x]M') = \lfs{\sigma, \lfs\sigma\Phi N}{\Phi, x{:}\tm}$ }
\\[1em]
\pcase{$\ianc
      {t \whnf \cbox{\hatctx{\Phi'} \vdash M}\quad
       \lfs {\sigma'}{\Phi'} M \lfwhnf \tapp~M_1~M_2 }
{\unbox t{\sigma'} \lfwhnf \tapp~M_1~M_2 }{}
$
}
\prf{$\Gamma ; \Psi \vdash \lfs{\sigma}{\Phi}\lfs {\sigma'}{\Phi'} M \lfwhnf  \lfs\sigma\Phi (\tapp~M_1~M_2) : \tm$ \hfill by IH}
\prf{$\Gamma ; \Psi \vdash \unbox t{\lfs{\sigma}{\Phi}\sigma'} \lfwhnf \lfs{\sigma}{\Phi}(\tapp~M_1~M_2) : \tm$ \hfill by  subst. prop., $\lfwhnf$ rule, and Def. \ref{def:typedwhnf}}

The remaining cases for (2) and (5) are similar.
}
 \end{proof}
}





\section{Kripke-style Logical Relation}\label{sec:logrel}
We construct a Kripke-logical relation that is defined on well-typed terms to prove weak head
normalization. 
Our semantic definitions for computations follow closely
\citet{Abel:LMCS12} to accommodate type-level computation.

\begin{figure}[h]
  \centering\small
  \[
    \begin{array}{c}
\infer{\Gamma ; \Psi \Vdash M = N : \tm}
{
\begin{array}{@{}lll@{}}
\Gamma ; \Psi \vdash M \lfwhnf \unbox {t_1} {\sigma_1} : \tm   & \typeof (\Gamma  \vdash t_1) = \cbox{\Phi_1 \vdash \tm} &
\\
\Gamma ; \Psi \vdash N \lfwhnf \unbox {t_2} {\sigma_2} : \tm & \typeof (\Gamma  \vdash t_2) = \cbox{\Phi_2 \vdash \tm} & \\
\multicolumn{3}{@{}c@{}}{
\Gamma  \vdash t_1 \equiv t_2 : \cbox{\Phi_1 \vdash \tm} \quad \Gamma ; \Psi \Vdash \sigma_1 = \sigma_2 : \Phi_1 \quad \Gamma \vdash \Phi_1 \equiv \Phi_2 : \ctx }
\end{array}
}
\\[0.75em]
\infer {\Gamma ; \Psi \Vdash M = N : \tm}
{
      \begin{array}{@{}ll@{}}
\Gamma ; \Psi \vdash M \lfwhnf \tlam~M' : \tm & \\
\Gamma ; \Psi \vdash N \lfwhnf \tlam~N' : \tm & \Gamma ; \Psi, x{:}\tm \Vdash M'~x = N'~x: \tm
      \end{array}
}
\\[0.75em]
\infer {\Gamma ; \Psi \Vdash M = N : \tm}
{
      \begin{array}{@{}ll@{}}
\Gamma ; \Psi \vdash M \lfwhnf \tapp~{M_1}~{M_2} : \tm & \Gamma ; \Psi \Vdash M_1 = N_1 : \tm \\
\Gamma ; \Psi \vdash N \lfwhnf \tapp~{N_1}~{N_2} : \tm & \Gamma ; \Psi \Vdash M_2 = N_2 : \tm
      \end{array}
}
\\[0.75em]
\infer {\Gamma ; \Psi \Vdash M = N : \tm}
{ \Gamma ; \Psi \vdash M \lfwhnf x : \tm  &
  \Gamma ; \Psi \vdash N \lfwhnf x : \tm }
  \end{array}
\]
\hrulefill
  \caption{Semantic Equality for LF Terms: \fbox{$\Gamma ; \Psi \Vdash M = N : A$}}
  \label{fig:LFsem}
\end{figure}

We start by defining semantic equality for LF terms of type $\tm$ (Fig.~\ref{fig:LFsem}), as these are the terms the recursor eliminates and it illustrates the fact that we are working with syntax trees.
To define semantic equality for LF terms $M$ and $N$, we consider different cases depending on their whnf: 1) if they reduce to $\tapp~M_1~M_2$ and $\tapp~N_1~N_2$ respectively, then $M_i$ must be semantically equal to $N_i$; 2) if they reduce to $\tlam~M'$ and $\tlam~N'$ respectively, then the bodies of $M'$ and $N'$ must be equal. To compare their bodies, we apply both $M'$ and $N'$ to an LF variable $x$ and consider $M'~x$ and $N'~x$ in the extended LF context $\Psi, x{:}\tm$. This has the effect of opening up the body and replacing the bound LF variable with a fresh one. This highlights the difference between the intensional LF function space and the extensional nature of the computation-level functions. In the former, we can concentrate on LF variables and continue to analyze the LF function body; in the latter, we consider all possible inputs, not just variables; 3) if the LF terms $M$ and $N$ may reduce to the same LF variable in $\Psi$, then they are obviously also semantically equal; 4) last, if $M$ and $N$ reduce to $\unbox {t_i} {\sigma_i}$ respectively. In this case $t_i$ is neutral and we only need to semantically compare the LF substitutions $\sigma_i$ and check whether the terms $t_i$ are definitional equal. However, what type should we choose? -- As the computation $t_i$ is neutral, we can infer a unique type $\cbox{\Phi \vdash \tm}$ which we can use. 
 This is defined as follows:

\begin{small}
\[
  \begin{array}{c}
 \multicolumn{1}{l}{\mbox{Type inference for neutral computations $t$}: \typeof (\Gamma \vdash t) = \tau}\\[1em]
 \infer{\typeof (\Gamma \vdash t~s) = \{s/y\}\tau_2}
 {\typeof (\Gamma \vdash t ) = \tau & \tau \whnf (y{:}\tau_1) \arrow \tau_2 &
  \Gamma \vdash s : \tau_1}
\\[1em]
 \infer{\typeof (\Gamma \vdash x) = \tau}{x{:}\tau \in \Gamma}
\quad
\infer{\typeof (\Gamma \vdash \titer{\R}{}{\IH}~\Psi~t) = \{\Psi/\psi, t/y\}\tau}
{\IH = (\psi : \tmctx) \arrow (y : \cbox{\psi \vdash \tm}) \arrow \tau}
  \end{array}
\]
\end{small}

\LONGVERSION{
  \begin{lemma}\label{lm:typeof}
If $\Gamma \vdash t : \tau$ and $\neut t$ then $\typeof (\Gamma \vdash t) = \tau'$ and $\Gamma \vdash \tau \equiv \tau' : u$.
  \end{lemma}
  \begin{proof}
By induction on $\neut t$.
  \end{proof}
}

\LONGVERSION{
\begin{lemma}[Weakening of Type Inference for Neutral Computations]\label{lem:wktypinf}
If $\typeof (\Gamma \vdash t) = \tau$ and $\Gamma' \leq_\rho \Gamma$ then $\typeof (\Gamma' \vdash \{\rho\}t) = \tau'$ s.t. $\tau' = \{\rho\}\tau$.
\end{lemma}
\begin{proof}
By induction on    $\typeof (\Gamma \vdash t) = \tau$ using Lemma \ref{lem:weakwhnf} (\ref{it:sweakcomp}).
\end{proof}
}

Semantic equality for LF substitutions is also defined by considering different whnfs (Fig. \ref{fig:LFsemctx}). As we only work with well-typed LF objects, there is only one inhabitant for an empty context. Moreover, given an LF substitution with domain $\Phi, x{:}A$, we can weak head reduce the LF substitutions $\sigma$ and $\sigma'$ and continue to recursively compare them. An LF substitution with domain $\psi$, a context variable, reduces to  $\wk\psi$.

\begin{figure}[h]
\small
  \centering
  \[
    \begin{array}{c}
\SUBSTCLO{
\infer{\Gamma ; \Psi \Vdash \sigma = \sigma' : \phi}
{
      \begin{array}{@{}lll@{}}
\Gamma ; \Psi \vdash \sigma \lfwhnf (\sclo \phi {\unbox{t_1}{\sigma_1}})  : \phi
   & \typeof (\Gamma \vdash t_1) = \cbox{\Phi_1 \vdash \Phi'_1}  & \Gamma \vdash \Phi'_1 \equiv (\phi, \wvec{x{:}A}) : \ctx
\\
\Gamma ; \Psi \vdash \sigma' \lfwhnf (\sclo \phi {\unbox{t_2}{\sigma_2}}) : \phi
    & \typeof (\Gamma \vdash t_2) = \cbox{\Phi_2 \vdash \Phi'_2} & \Gamma \vdash \Phi'_2 \equiv (\phi, \wvec{x{:}A}) : \ctx
\\
\Gamma \vdash t_1 \equiv t_2 : \cbox{\Phi_1 \vdash \Phi'_1}
    & \Gamma ; \Psi \Vdash \sigma_1 = \sigma_2 : \Phi_1
    & \Gamma \vdash \Phi_1 \equiv \Phi_2 : \ctx
\end{array}
}
\\[0.5em]
}
\raisebox{1ex}{
\ianc{
\Gamma ; \Psi \vdash \sigma \lfwhnf \cdot : \cdot \quad
\Gamma ; \Psi \vdash \sigma' \lfwhnf \cdot : \cdot }
{\Gamma ; \Psi \Vdash \sigma = \sigma' : \cdot}{}}
 \\[0.5em]
 \ianc{
\Gamma ; \psi, \wvec{x{:}A} \vdash \sigma \lfwhnf \wk{\psi} : \psi \quad
\Gamma ; \psi, \wvec{x{:}A} \vdash \sigma' \lfwhnf \wk{\psi} : \psi
}{\Gamma ; \psi, \wvec{x{:}A} \Vdash \sigma = \sigma' : \psi}{}
\\[1em]
 \infer{\Gamma ; \Psi \Vdash \sigma = \sigma' : \Phi, x{:}A}
  {\begin{array}{@{}ll@{}}
\Gamma ; \Psi \vdash \sigma \lfwhnf \sigma_1, M : \Phi, x{:}A  & \Gamma ; \Psi \Vdash \sigma_1 = \sigma_2 : \Phi \\
\Gamma ; \Psi \vdash \sigma' \lfwhnf \sigma_2, N : \Phi, x{:}A & \Gamma ; \Psi  \Vdash M = N : \lfs{\sigma_1}{\Phi} A
\end{array}
}
  \end{array}
\]
  \caption{Semantic Equality for LF Substitutions: \fbox{$\Gamma ; \Psi \Vdash \sigma = \sigma' : \Phi$}}
  \label{fig:LFsemctx}
\end{figure}

\LONGVERSION{
We can then lift the semantic equality definition to contextual terms: \fbox{$\Gamma \semlf C = C' : T$}.

\[
\infer{\Gamma \semlf (\hatctx{\Psi} \vdash M) = (\hatctx{\Psi} \vdash N): (\Psi \vdash A)}
{\Gamma ; \Psi \Vdash M = N : A}
\]
}

\LONGVERSION{
\begin{figure}
\hrulefill
  \centering\small
  \[
    \begin{array}{c}
\infer{\Gamma \semlf (\hatctx{\Psi} \vdash M) = (\hatctx{\Psi} \vdash N): (\Psi \vdash A)}
{\Gamma ; \Psi \Vdash M = N : A}
\LONGVERSION{\quad
\infer{\Gamma \semlf (\hatctx{\Psi} \vdash M) = (\hatctx{\Psi} \vdash N) : (\Psi \vdash_\# A)}
{ \Gamma ; \Psi \Vdash_\# M = N : A}}
\SUBSTCLO{\\[1em]
\infer{\Gamma \semlf (\hatctx{\Psi} \vdash \sigma) = (\hatctx{\Psi} \vdash \sigma'): (\Psi \vdash \Phi)}
{\Gamma ; \Psi \Vdash \sigma = \sigma' : \Phi}
\qquad
\infer{\Gamma \semlf (\hatctx{\Psi} \vdash \sigma) = (\hatctx{\Psi} \vdash \sigma) : (\Psi \vdash_\# \Phi)}
{ \Gamma ; \Psi \Vdash_\# \sigma = \sigma' : \Phi}}
  \end{array}
\]
\hrulefill
  \caption{Semantic Typing for Contextual LF Terms}
  \label{fig:LFsem2}
\end{figure}
}

Defining semantic kinding and semantic equality is intricate, as they
depend on each other and we need to ensure our definitions are
well-founded. Following \citet{Abel:POPL18}, we first define semantic
kinding, i.e. $\Gamma \Vdash \ann\tau : u$ (Fig. \ref{fig:semkind})
which technically falls into two parts: $\Gamma \Vdash \tau : u$ and
$\Gamma \Vdash \tmctx : u$ where the latter is simply notation, as
$\tmctx$ is not a computation-level type. Function types
$(y:\ann\tau_1) \arrow \tau_2$ are semantically well-kinded if
$\ann\tau_1$ is semantically well-kinded in any extension
$\Gamma'$ of $\Gamma$ and $\{s/y\}\tau_2$ is well-kinded
for any term $s$ that has semantic type $\ann\tau_1$. In our definition, we
make the renaming $\rho$ that allows us to move from
$\Gamma$ to $\Gamma'$ explicit. The definition of semantic kinding is
inductively defined on $\ann\tau$.


\begin{figure}[h]
\vspace{1ex}\small
\[
  \begin{array}{c}
\infer{\Gamma \Vdash \tau : u}
  { \Gamma \der \tau \whnf \cbox{T} : u
  & \Gamma \vdash T \equiv T }
\qquad
\infer{\Gamma \Vdash \tau : u}
  {\Gamma \der \tau \whnf u' : u & u' < u}
\\[0.5em]
\infer{\Gamma \Vdash \tau : u}
{\Gamma \vdash \tau \whnf x~\vec{t} : u \qquad \neut (x~\vec{t})}
\qquad
\infer{\Gamma \Vdash \tmctx: u}
  {\vdash \Gamma}
\\[0.5em]
\infer{\Gamma \Vdash \tau : u}
  {
    \begin{array}{@{}c@{}}
     \Gamma \der \tau \whnf (y:\ann\tau_1) \arrow \tau_2 : u
     \hfill~\quad
     \forall \Gamma' \leq_\rho \Gamma.\ \Gamma' \Vdash \{\rho\}\ann\tau_1 : u_1
   \\
   \forall \Gamma' \leq_\rho \Gamma.\ \Gamma' \Vdash s = s~ : \{\rho\}\ann\tau_1
    \Longrightarrow \Gamma' \Vdash \{\rho, s/y\} \tau_2 : u
    \end{array}
}\\
\multicolumn{1}{c}{\mbox{where}~(u_1,~u_2,~u) \in \Ru}
  \end{array}
\]
\caption{Semantic Kinding for Types \fbox{$\Gamma \sem \ann\tau : u$} (inductive)}
\label{fig:semkind}
\end{figure}

Semantic kinding (Fig.~\ref{fig:semkind}) is used as a measure to define the semantic typing for computations. In particular, we define $\Gamma \Vdash \ann\tau = \ann\tau' : u$ and $\Gamma \Vdash t = t' : \ann\tau$ recursively on the semantic kinding of $\ann\tau$. i.e. $\Gamma \Vdash \ann\tau : u$.  For better readability, we simply write for example $\Gamma \Vdash t = t': \cbox{T}$ instead of
$\Gamma \Vdash t = t': \tau$ where $\tau \whnf \cbox{T}$, and $\Gamma
\vdash T \equiv T$ in proofs. The extensional character of function
types is apparent in the semantic equality for terms at function type.
Semantic equality at type $\cbox{\Psi \vdash A}$ falls back to
semantic equality on LF terms at type $A$ in the LF context $\Psi$.
\LONGVERSION{
We note that to prove reflexivity for types, we would need to strengthen our semantic kinding definition with the additional premise:
$\forall \Gamma' \leq_\rho \Gamma.\ \Gamma' \Vdash s = s' : \{\rho\}\ann\tau_1
   \Longrightarrow \Gamma' \Vdash \{\rho, s/y\} \tau_2 = \{\rho, s'/y\} \tau_2 : u_2$. This is possible, but since semantic reflexivity for types is not needed, we keep the semantic kinding definition more compact.
}

\begin{figure*}
  \centering\small
\[
  \begin{array}{c}
\multicolumn{1}{l}{\mbox{Semantic equality for types: \fbox{$\Gamma \Vdash \ann\tau = \ann\tau' :u$} defined by recursion on $\Gamma \Vdash \tau : u$}}\\[0.5em]
\multicolumn{1}{l}{\mbox{Semantic equality for terms: \fbox{$\Gamma
    \Vdash t = t' : \ann\tau$} defined by recursion on $\Gamma \Vdash \ann\tau : u$}}
\\[0.75em]
\infer{\Gamma \Vdash \tmctx = \tmctx : u}{}
\quad
\inferrule*{\Gamma \vdash \tau' \whnf u' : u}
           {\Gamma \Vdash u' = \tau' : u}
\quad
\inferrule*{\Gamma \vdash \tau' \whnf \cbox{T'} : u~ \and
            \Gamma \vdash T \equiv T'}
           {\Gamma \Vdash \cbox T = \tau' : u }
\quad
\infer{\Gamma \Vdash x~\vec{t} = \tau' : u}
{\Gamma \vdash \tau' \whnf x~\vec{s} : u~ &
\Gamma \vdash x~\vec{t} \equiv x~\vec{s} : u}
\\[0.75em]
\infer[(u_1, u_2, u) \in \Ru]
           {\Gamma \Vdash (y:\ann\tau_1) \arrow \tau_2 = \tau' : u}
{
 \begin{array}{@{}l@{}}
\Gamma \vdash \tau' \whnf (y' :\ann\tau_1' )\arrow \tau_2' : u \qquad\hfill
 \forall \Gamma' \leq_\rho \Gamma.~ \Gamma'\Vdash \{\rho\}\ann\tau_1 = \{\rho\}\ann\tau_1' : u_1\\
 \forall \Gamma'\leq_\rho \Gamma.~ \Gamma' \Vdash s = s' : \{\rho\}\ann\tau_1 \Longrightarrow \Gamma' \Vdash \{\rho, s/y\} \tau_2 = \{\rho,s'/y'\}\tau_2' : u_2
    \end{array}
}
\\[0.75em]
\infer{\Gamma \Vdash \Psi = \Psi' : \tmctx }
      {\Gamma \vdash \Psi \equiv \Psi' : \tmctx  }
\qquad
\infer{\Gamma \Vdash t = t' : \cbox {\Psi \vdash A} }
      {\Gamma \vdash t \whnf w : \cbox{\Psi \vdash A} &
\Gamma \vdash t' \whnf w' : \cbox{\Psi \vdash A} &
\Gamma ; \Psi \Vdash \unbox{w}{\id} = \unbox{w'}{\id} : A}
\\[0.75em]
\infer{\Gamma \Vdash t = t' : x~\vec{s}}
{\Gamma \vdash t \whnf n :  x~ \vec{s} &
 \Gamma \vdash t' \whnf n' :  x~\vec{s} & ~\neut n,n'~&
\Gamma \vdash n \equiv n' : x~\vec{s}
}
\\[0.75em]
\infer{\Gamma \Vdash t = t' : (y:\ann\tau_1) \arrow \tau_2 }
{\Gamma \vdash t \whnf w : (y:\ann\tau_1) \arrow \tau_2 &
\Gamma \vdash t'\whnf w': (y:\ann\tau_1) \arrow \tau_2 &
\forall \Gamma' \leq_\rho \Gamma.
~\Gamma' \Vdash s = s' : \{\rho\}\ann\tau_1 \Longrightarrow \Gamma' \Vdash \{\rho\}w~s = \{\rho\}w'~s' : \{\rho, s/y\}\tau_2
}
\end{array}
\]
  \caption{Semantic Equality for Computations}
  \label{fig:semcomp}
\end{figure*}

\section{Semantic Properties }\label{sec:semprop}
\subsection{Semantic Properties of LF}

If an LF term is semantically well-typed, then it is also syntactically well-typed. Furthermore, our definition of semantic LF typing is stable under renaming and semantic LF equality is preserved under LF substitution and allows for context conversion.




\LONGVERSION{
 \begin{lemma}[Well-Formedness of Semantic LF Typing] \quad \label{lm:semlfwf}
   \begin{enumerate}
   \item If $\Gamma ; \Psi \Vdash M = N : A$ then
         $\Gamma ; \Psi \vdash M : A$ and $\Gamma ; \Psi \vdash N : A$
         and $\Gamma \vdash M \equiv N : A$.
   \item If $\Gamma ; \Psi \Vdash \sigma_1 = \sigma_2 : \Phi$ then
         $\Gamma ; \Psi \vdash \sigma_1 : \Phi$ and $\Gamma ; \Psi \vdash \sigma_2 : \Phi$
         and $\Gamma ; \Psi \vdash \sigma_1 \equiv \sigma_2 : \Phi$.
   \end{enumerate}
 \end{lemma}
 \begin{proof}
By induction on the semantic definition. In each case, we refer the Def.~\ref{def:typedwhnf}. To illustrate, consider the case where
$\Gamma ; \Psi \vdash M \lfwhnf \lambda x.M' : \Pi x{:}A.B$ and
$\Gamma ; \Psi \vdash N \lfwhnf \lambda x.N' : \Pi x{:}A.B$, we also know that
$\Gamma ; \Psi \vdash M \equiv \lambda x.M' : \Pi x{:}A.B$ and
$\Gamma ; \Psi \vdash N \equiv \lambda x.N' : \Pi x{:}A.B$ by Def.~\ref{def:typedwhnf}.

Further, we have that $\Gamma ; \Psi, x{:}A \Vdash M' = N' : B$. By
IH, we get that $\Gamma ; \Psi, x{:}A \vdash M' \equiv N' : B$ By
dec. equivalence rules, we have $\Gamma ; \Psi \vdash \lambda x.M'
\equiv \lambda x.N' : \Pi x{:}A.B$. Therefore, by symmetry and
transitivity of $\equiv$, we have $\Gamma ; \Psi \vdash M \equiv N :
\Pi x{:}A.B$. The typing invariants are left implicit.
\LONGVERSIONCHECKED{
\\
 We show the expanded proofs below concentrating on showing
  $\equiv$ and leaving the tracking of typing invariants implicit.
  \\[1em]
\SUBSTCLO{
  \pcase {$\infer{\Gamma ; \Psi \Vdash \sigma = \sigma' : \Phi}
    {\begin{array}{@{}lll@{}}
        \Gamma ; \Psi \vdash \sigma \lfwhnf (\sclo \phi {\unbox{t_1}{\sigma_1}})  : \phi
        & \typeof (\Gamma \vdash t_1) = \cbox{\Phi_1 \vdash \Phi'_1}
        & \Gamma \vdash \Phi'_1 \equiv (\phi, \wvec{x{:}A}) : \ctx \\
        \Gamma ; \Psi \vdash \sigma' \lfwhnf (\sclo \phi {\unbox{t_2}{\sigma_2}}) : \phi
        & \typeof (\Gamma \vdash t_2) = \cbox{\Phi_2 \vdash \Phi'_2}
        & \Gamma \vdash \Phi'_2 \equiv (\phi, \wvec{x{:}A}) : \ctx \\
        \Gamma \vdash t_1 \equiv t_2 : \cbox{\Phi_1 \vdash \Phi'_1}
        & \Gamma ; \Psi \Vdash \sigma_1 = \sigma_2 : \Phi_1
        & \Gamma \vdash \Phi_1 \equiv \Phi_2 : \ctx
      \end{array}}$}
  \prf{$\Gamma ; \Psi \vdash \sigma \equiv (\sclo \phi {\unbox{t_1}{\sigma_1'}}) : \phi$
    and $\Gamma ; \Psi \vdash \sigma' \equiv (\sclo \phi {\unbox{t_2}{\sigma_2'}}) : \phi$
    \hfill by Def.~\ref{def:typedwhnf}}
  \prf{$\Gamma ; \Psi \vdash \sigma_1' \equiv \sigma_2' : \Phi_1$
    \hfill by induction hypothesis}
  \prf{$\Gamma ; \Psi \vdash (\sclo \phi {\unbox{t_1}{\sigma_1'}})
    \equiv (\sclo \phi {\unbox{t_2}{\sigma_2'}}) : \phi$
    \hfill by dec. equivalence rules}
  \prf{$\Gamma ; \Psi \vdash \sigma \equiv \sigma': \phi$
    \hfill by symmetry, transitivity, conversion of $\equiv$}
  \\
}
  \pcase {$\infer{\Gamma ; \Psi \Vdash \sigma = \sigma' : \cdot}
    {\Gamma ; \Psi \vdash \sigma \lfwhnf \cdot : \cdot  &
      \Gamma ; \Psi \vdash \sigma' \lfwhnf \cdot : \cdot }$}
  \prf{$\Gamma ; \Psi \vdash \sigma_1 \equiv \cdot : \cdot$
    and $\Gamma ; \Psi \vdash \sigma_2 \equiv \cdot : \cdot$
    \hfill by Def.~\ref{def:typedwhnf}}
  \prf{$\Gamma ; \Psi \vdash \sigma_1 \equiv \sigma_2 : \cdot$ \hfill   by symmetry and transitivity of $\equiv$}
  \\
  \pcase {$ \infer{\Gamma ; \psi, \wvec{x{:}A} \Vdash \sigma = \sigma' : \psi}{
      \Gamma ; \psi, \wvec{x{:}A} \vdash \sigma \lfwhnf \wk{\psi} : \psi &
      \Gamma ; \psi, \wvec{x{:}A} \vdash \sigma' \lfwhnf \wk{\psi} : \psi
    }$}
  \prf{$\Gamma ; \psi, \wvec{x{:}A} \vdash \sigma_1 \equiv \wk{\psi} : \psi$ and
    $\Gamma ; \psi, \wvec{x{:}A} \vdash \sigma_2 \equiv \wk{\psi} : \psi$
    \hfill by Def.~\ref{def:typedwhnf}}
  \prf{$\Gamma ; \psi \wvec{x{:}A} \vdash \sigma_1 \equiv \sigma_2 : \psi$
    \hfill by symmetry and transitivity of $\equiv$}
  \\
  \pcase{$ \infer{\Gamma ; \Psi \Vdash \sigma_1 = \sigma_2 : \Phi, x{:}A}
    {\begin{array}{@{}ll@{}}
       \Gamma ; \Psi \vdash \sigma_1 \lfwhnf \sigma'_1, M : \Phi, x{:}A  & \\
       \Gamma ; \Psi \vdash \sigma_2 \lfwhnf \sigma'_2, N : \Phi, x{:}A &
       \Gamma ; \Psi \Vdash \sigma'_1 = \sigma'_2 : \Phi
       \qquad \Gamma ; \Psi  \Vdash M = N : \lfs{\sigma'_1}{\Phi} A
      \end{array}}$}
  \prf{$\Gamma ; \Psi \vdash \sigma_1 \equiv \sigma_1', M : \Phi, x{:}A$
    and $\Gamma ; \Psi \vdash \sigma_2 \equiv \sigma_2', N : \Phi, x{:}A$
    \hfill by Def.~\ref{def:typedwhnf}}
  \prf{$\Gamma ; \Psi \vdash \sigma_1' \equiv \sigma_2' : \Phi$
    and $\Gamma ; \Psi  \vdash M \equiv N : \lfs{\sigma_1'}{\Phi} A$
    \hfill by induction hypothesis}
  \prf{$\Gamma ; \Psi \vdash \sigma_1', M \equiv \sigma_2', N : \Phi, x{:}A$
    \hfill by dec. equivalence rules}
  \prf{$\Gamma ; \Psi \vdash \sigma_1 \equiv \sigma_2 : \Phi, x{:}A$
    \hfill by symmetry and transitivity of $\equiv$}
  \\
  \pcase{$\infer{\Gamma ; \Psi \Vdash M = N : \tm}
    {\begin{array}{@{}lll@{}}
        \Gamma ; \Psi \vdash M \lfwhnf \unbox {t_1} {\sigma_1} : \tm
        & \typeof (\Gamma  \vdash t_1) = \cbox{\Phi_1 \vdash \tm} &
        \Gamma \vdash \Phi_1 \equiv \Phi_2 : \ctx \\
        \Gamma ; \Psi \vdash N \lfwhnf \unbox {t_2} {\sigma_2} : \tm
        & \typeof (\Gamma  \vdash t_2) = \cbox{\Phi_2 \vdash \tm} &
        \Gamma  \vdash t_1 \equiv t_2 : \cbox{\Phi_1 \vdash \tm}
        \quad \Gamma ; \Psi \Vdash \sigma_1 = \sigma_2 : \Phi_1
      \end{array}}$}
  \prf{$\Gamma ; \Psi \vdash M \equiv \unbox {t_1} {\sigma_1} : \tm $
    and $\Gamma ; \Psi \vdash N \equiv \unbox {t_2} {\sigma_2} : \tm$
    \hfill by Def.~\ref{def:typedwhnf}}
  \prf{$\Gamma ; \Psi \vdash \sigma_1 \equiv \sigma_2 : \Phi_1$ \hfill by induction hypothesis}
  \prf{$\Gamma ; \Psi \vdash \unbox {t_1} {\sigma_1} \equiv \unbox {t_2} {\sigma_2} : \tm$
    \hfill by dec. equivalence rules}
  \prf{$\Gamma ; \Psi \vdash M \equiv N : \tm$ \hfill by symmetry and transitivity  of $\equiv$}
  \\
  \pcase{$\infer {\Gamma ; \Psi \Vdash M = N : \tm}
    {\begin{array}{@{}ll@{}}
        \Gamma ; \Psi \vdash M \lfwhnf \tlam~M' : \tm & \\
        \Gamma ; \Psi \vdash N \lfwhnf \tlam~N' : \tm
        & \Gamma ; \Psi \Vdash M' = N': \Pi x{:}\tm.\tm
      \end{array}}$}
  \prf{$\Gamma ; \Psi \vdash M \equiv \tlam~M' : \tm$
    and $\Gamma ; \Psi \vdash N \equiv \tlam~N' : \tm$
    \hfill by Def.~\ref{def:typedwhnf}}
  \prf{$\Gamma ; \Psi \vdash M' \equiv N': \Pi x{:}\tm.\tm$
    \hfill by induction hypothesis}
  \prf{$\Gamma ; \Psi \vdash \tlam~M' \equiv \tlam~N' : \tm$
    \hfill by dec. equivalence rules}
  \prf{$\Gamma ; \Psi \vdash M \equiv N : \tm$ \hfill by symmetry and transitivity  of $\equiv$}
  \\
  \pcase{$\infer {\Gamma ; \Psi \Vdash M = N : \tm}
    {\begin{array}{@{}ll@{}}
        \Gamma ; \Psi \vdash M \lfwhnf \tapp~{M_1}~{M_2} : \tm & \\
        \Gamma ; \Psi \vdash N \lfwhnf \tapp~{N_1}~{N_2} : \tm &
        \Gamma ; \Psi \Vdash M_1 = N_1 : \tm \quad \Gamma ; \Psi \Vdash M_2 = N_2 : \tm
      \end{array}}$}
  \prf{$\Gamma ; \Psi \vdash M \equiv \tapp~{M_1}~{M_2} : \tm$
    and $\Gamma ; \Psi \vdash N \equiv \tapp~{N_1}~{N_2} : \tm$
    \hfill by Def.~\ref{def:typedwhnf}}
  \prf{$\Gamma ; \Psi \vdash M_1 \equiv N_1 : \tm$
    and $\Gamma ; \Psi \vdash M_2 \equiv N_2 : \tm$
    \hfill by induction hypothesis}
  \prf{$\Gamma ; \Psi \vdash \tapp~{M_1}~{M_2} \equiv \tapp~{N_1}~{N_2}$
    \hfill by dec. equivalence rules}
  \prf{$\Gamma ; \Psi \vdash M \equiv N : \tm$ \hfill by symmetry and transitivity  of $\equiv$}
  \\
  \pcase{$\infer {\Gamma ; \Psi \Vdash M = N : \tm}
    { \Gamma ; \Psi \vdash M \lfwhnf x : \tm  &
      \Gamma ; \Psi \vdash N \lfwhnf x : \tm }$}
  \prf{$\Gamma ; \Psi \vdash M \equiv x : \tm$
    and $\Gamma ; \Psi \vdash N \equiv x : \tm$
    \hfill by Def.~\ref{def:typedwhnf}}
  \prf{$\Gamma ; \Psi \vdash M \equiv N : \tm$ \hfill by symmetry and transitivity  of $\equiv$}}
 \end{proof}

\begin{lemma}[Semantic Weakening for LF]\label{lem:semweak}\quad
\begin{enumerate}
  \item \label{it:sweaksemlf} If\/ $\Gamma ; \Psi \Vdash M = N: A$ and $\Gamma' \leq_\rho \Gamma$
     then $\Gamma' ; \{\rho\}\Psi \Vdash \{\rho\}M = \{\rho\}N: \{\rho\}A$.
  \item \label{it:sweaklfsub} If\/ $\Gamma ; \Psi \Vdash \sigma = \sigma': \Phi$ and $\Gamma' \leq_\rho \Gamma$
     then $\Gamma' ; \{\rho\}\Psi \Vdash \{\rho\}\sigma = \{\rho\}\sigma': \{\rho\}\Phi$.
  \end{enumerate}
\end{lemma}
\begin{proof}
By induction on the first derivation.
\LONGVERSIONCHECKED{
\\[1em]
\pcase{$\ianc
{\begin{array}{@{}lll@{}}
\Gamma ; \Psi \vdash M \lfwhnf \unbox {t_1} {\sigma_1} : \tm   & \typeof (\Gamma  \vdash t_1) = \cbox{\Phi_1 \vdash \tm} &
\Gamma \vdash \Phi_1 \equiv \Phi_2 : \ctx\\
\Gamma ; \Psi \vdash N \lfwhnf \unbox {t_2} {\sigma_2} : \tm & \typeof (\Gamma  \vdash t_2) = \cbox{\Phi_2 \vdash \tm} &
\Gamma ; \Psi \Vdash \sigma_1 = \sigma_2 : \Phi_1 \quad
\Gamma  \vdash t_1 \equiv t_2 : \cbox{\Phi_1 \vdash \tm}
\end{array}}
{\Gamma \Vdash M = N : \tm}{}$}
    \prf{$\Gamma'; \{\rho\}\Psi \vdash \{\rho\}M \lfwhnf \{\rho\}(\unbox {t_1} {\sigma_1}) : \tm$ and
         $\Gamma'; \{\rho\}\Psi \vdash \{\rho\}N \lfwhnf \{\rho\}(\unbox {t_2} {\sigma_1}) : \tm$ \hfill by Lemma \ref{lem:weakwhnf}}
    \prf{$\Gamma'; \{\rho\}\Psi \vdash \{\rho\}M \lfwhnf \unbox {\{\rho\}t_1} {\{\rho\}\sigma_1} : \tm$ and
         $\Gamma'; \{\rho\}\Psi \vdash \{\rho\}N \lfwhnf \unbox {\{\rho\}t_2} {\{\rho\}\sigma_2} : \tm$ \hfill by subst. def.}
    \prf{$\Gamma'; \{\rho\}\Psi \Vdash \{\rho\}\sigma_1 = \{\rho\}\sigma_2: \{\rho\}\Phi_1$ \hfill by IH}
    \prf{$\Gamma' \vdash \{\rho\}\Phi_1 \equiv \{\rho\}\Phi_2 : \ctx$ \hfill by Lemma \ref{lem:weakcomp}}
    \prf{$\Gamma' \vdash \{\rho\}t_1 \equiv \{\rho\}t_2 : \{\rho\}\cbox{\Phi_1 \vdash \tm}$ \hfill by Lemma \ref{lem:weakcomp}}
    \prf{$\Gamma' \vdash \{\rho\}t = \{\rho\}t': \cbox{\{\rho\}\Phi_1 \vdash \tm}$ \hfill by substitution def.}
    \prf{$\typeof (\Gamma'  \vdash \{\rho\}t_1) = \cbox{\{\rho\}\Phi_1 \vdash \tm}$ \hfill by Lemma \ref{lem:wktypinf}}
    \prf{$\typeof (\Gamma'  \vdash \{\rho\}t_2) = \cbox{\{\rho\}\Phi_2 \vdash \tm}$ \hfill by Lemma \ref{lem:wktypinf} }
    \prf{$\Gamma' ; \{\rho\}\Psi \Vdash \{\rho\}M = \{\rho\}N: \{\rho\}\tm$ \hfill by rule and substitution def.}
\\
}
\end{proof}

}

\begin{lemma}[Backwards Closure for LF Terms]\label{lem:lfbclosed}\quad\\
If $\Gamma; \Psi \Vdash  Q = N: A$ (or $\Gamma; \Psi \Vdash N = Q : A$)\\
and $\Gamma ; \Psi \vdash M \lfwhnf Q : A$
then $\Gamma \sem M = N: A$

\end{lemma}
\begin{proof}
By case analysis on $\Gamma; \Psi \Vdash  Q = N: A$ and the fact that
$Q$ is in $\norm$. 
\LONGVERSIONCHECKED{
$\quad$\\[1em]
\pcase{$\ianc
{
\begin{array}{@{}lll@{}}
\Gamma ; \Psi \vdash Q \lfwhnf \unbox {t_1} {\sigma_1} : \tm   & \typeof (\Gamma  \vdash t_1) = \cbox{\Phi_1 \vdash \tm} &
\Gamma \vdash \Phi_1 \equiv \Phi_2 : \ctx \\
\Gamma ; \Psi \vdash N \lfwhnf \unbox {t_2} {\sigma_2} : \tm & \typeof (\Gamma  \vdash t_2) = \cbox{\Phi_2 \vdash \tm} &
\Gamma  \vdash t_1 \equiv t_2 : \cbox{\Phi_1 \vdash \tm} \quad \Gamma ; \Psi \Vdash \sigma_1 = \sigma_2 : \Phi_1
\end{array}
}
{\Gamma ; \Psi \Vdash Q = N : \tm}{}$}
\prf{$\Gamma ; \Psi \vdash M \lfwhnf Q : A$ \hfill by assumption}
\prf{$\norm Q$ \hfill by invariant of $\lfwhnf$}
\prf{$Q = \unbox {t_1} {\sigma_1}$ \hfill since $\norm Q$}
\prf{$\Gamma ; \Psi \Vdash M = N : \tm$ \hfill using $\M \lfwhnf \unbox{t_1}{\sigma_1}$ and sem. def.}
\\[1em]
\pcase{$\ianc
{
      \begin{array}{@{}ll@{}}
\Gamma ; \Psi \vdash Q \lfwhnf \tlam~M' : \tm & \\
\Gamma ; \Psi \vdash N \lfwhnf \tlam~N' : \tm & \Gamma ; \Psi \Vdash M' = N': \Pi x{:}\tm.\tm
      \end{array}
}{\Gamma ; \Psi \Vdash Q = N : \tm}{}
$}
\prf{$\Gamma ; \Psi \vdash M \lfwhnf Q : A$ \hfill by assumption}
\prf{$\norm Q$ \hfill by invariant of $\lfwhnf$}
\prf{$Q = \tlam~M' $ \hfill by $\Gamma ; \Psi \vdash Q \lfwhnf \tlam~M' : \tm $ using $\norm Q$}
\prf{$\Gamma ; \Psi \Vdash M = N : \tm$ \hfill using $\Gamma ; \Psi \vdash M \lfwhnf \tlam~M' : \tm$}
}
\end{proof}

\LONGVERSION{
 \begin{lemma}[Semantic LF Equality Is Preserved Under LF Substitution] \label{lem:semlfeqsub}\quad
   \begin{enumerate}
   \item If $\Gamma ; \Psi \sem \sigma = \sigma' : \Phi$
and $\Gamma ; \Phi \sem M = N : A$
then $\Gamma ; \Psi \sem \lfs\sigma\Phi M = \lfs{\sigma'}\Phi N : \lfs\sigma\Phi A$.
   \item If $\Gamma ; \Psi \sem \sigma = \sigma' : \Phi$
and $\Gamma ; \Phi \sem \sigma_1 = \sigma_2 : \Phi'$
then $\Gamma ; \Psi \sem \lfs\sigma\Phi \sigma_1 = \lfs{\sigma'}\Phi \sigma_2 : \Phi'$.
   \end{enumerate}
 \end{lemma}
 \begin{proof}
Proof by mutual induction on $\Gamma ; \Phi \sem M = N : A$ and $\Gamma ; \Phi \sem \sigma = \sigma' : \Phi$ using the fact that weak head reduction is preserved under substitution (Lemma \ref{lm:lfwhnfsub}).
\LONGVERSIONCHECKED{
\\[0.5em]
 (1) \fbox{If $\Gamma ; \Psi \sem \sigma = \sigma' : \Phi$
and $\Gamma ; \Phi \sem M = N : A$
then $\Gamma ; \Psi \sem \lfs\sigma\Phi M = \lfs{\sigma'}\Phi N :\lfs\sigma\Phi A$.}
\\[0.5em]
\pcase{
$\ianc
{
\begin{array}{@{}lll@{}}
\Gamma ; \Phi \vdash M \lfwhnf \unbox {t_1} {\sigma_1} : \tm   & \typeof (\Gamma  \vdash t_1) = \cbox{\Phi_1 \vdash \tm} &
\Gamma \vdash \Phi_1 = \Phi_2 : \ctx\\
\Gamma ; \Phi \vdash N \lfwhnf \unbox {t_2} {\sigma_2} : \tm & \typeof (\Gamma  \vdash t_2) = \cbox{\Phi_2 \vdash \tm} &
\Gamma ; \Phi \Vdash \sigma_1 = \sigma_2 : \Phi_1 \quad
\Gamma \vdash t_1 \equiv t_2 : \cbox{\Phi_1 \vdash \tm}
\end{array}
}
{\Gamma ; \Phi \Vdash M = N : \tm}{}$
}
\prf{$\Gamma ; \Psi \vdash \sigma : \Phi$ and $\Gamma ; \Psi \vdash \sigma' : \Phi$ \hfill by well-formedness of semantic equ. (Lemma \ref{lm:semlfwf})}
\prf{$\Gamma ; \Psi \vdash \lfs\sigma\Phi M \lfwhnf \lfs\sigma\Phi (\unbox{t_1}{\sigma_1}) : \tm$ \hfill by Lemma \ref{lm:lfwhnfsub} }
\prf{$\Gamma ; \Psi \vdash \lfs{\sigma'}{\Phi} N \lfwhnf \lfs{\sigma'}{\Phi}(\unbox{t_2}{\sigma_2}) : \tm$ \hfill by Lemma \ref{lm:lfwhnfsub}}
\prf{$\Gamma ; \Phi \Vdash \lfs\sigma\Phi\sigma_1 = \lfs{\sigma'}\Phi \sigma_2 : \Phi_1$ \hfill by IH}
\prf{$\Gamma ; \Phi \Vdash \lfs{\sigma}\Phi M = \lfs{\sigma'}\Phi N : \tm$ \hfill by well-typed }
\\
\pcase{
$\ianc
{ \Gamma ; \Phi \vdash M \lfwhnf x : \tm  \quad
  \Gamma ; \Phi \vdash N \lfwhnf x : \tm }
 {\Gamma ; \Phi \Vdash M = N : \tm}{}$
}
\prf{$\Gamma ; \Psi \Vdash \sigma(x) = \sigma'(x) : \tm$ \hfill by $\Gamma ; \Phi \Vdash \sigma = \sigma' : \Psi$}
\prf{$\Gamma ; \Psi \vdash \sigma(x) \lfwhnf M' : \tm$ and $\Gamma ; \Phi \vdash \sigma'(x) \lfwhnf N' : \tm$
 \hfill by $\Gamma ; \Psi \Vdash \sigma(x) = \sigma'(x) : \tm$
}
\prf{$\Gamma ; \Psi \Vdash M' = N' : \tm$ \hfill since both $\norm M'$ and $\norm N'$}
\prf{$\Gamma ; \Psi \vdash \lfs\sigma\Phi M \lfwhnf M' : \tm$    \hfill by Lemma \ref{lm:lfwhnfsub}
}
\prf{$\Gamma ; \Psi \vdash \lfs\sigma\Phi N \lfwhnf N' : \tm$    \hfill by Lemma \ref{lm:lfwhnfsub}}
\prf{$\Gamma ; \Phi \Vdash \lfs\sigma\Phi M = \lfs{\sigma'}\Phi N : \tm$ \hfill Backwards Closure (Lemma \ref{lem:lfbclosed})}
\\
Other cases are similar.
}
\LONGVERSIONCHECKED{
\\[1em]
(2) \fbox{If $\Gamma ; \Psi \sem \sigma = \sigma' : \Phi$
and $\Gamma ; \Phi \sem \sigma_1 = \sigma_2 : \Phi'$
then $\Gamma ; \Psi \sem \lfs\sigma\Phi \sigma_1 = \lfs{\sigma'}\Phi \sigma_2 : \Phi'$.}\\[1em]
Proof by induction on $\Gamma ; \Phi \sem \sigma = \sigma' : \Phi$ using the fact that weak head reduction is preserved under substitution (Lemma \ref{lm:lfwhnfsub}).
\\[1em]
\SUBSTCLO{
 \pcase{$\ianc {
       \begin{array}{@{}lll@{}}
 \Gamma ; \Phi \vdash \sigma_1 \lfwhnf (\sclo \phi {\unbox{t_1}{\sigma'_1}})  : \phi
    & \typeof (\Gamma \vdash t_1) = \cbox{\Phi_1 \vdash \Phi'_1}  & \Gamma \vdash \Phi'_1 \equiv (\phi, \wvec{x{:}A}) : \ctx
 \\
 \Gamma ; \Phi \vdash \sigma_2 \lfwhnf (\sclo \phi {\unbox{t_2}{\sigma'_2}}) : \phi
     & \typeof (\Gamma \vdash t_2) = \cbox{\Phi_2 \vdash \Phi'_2} & \Gamma \vdash \Phi'_2 \equiv (\phi, \wvec{x{:}A}) : \ctx
 \\
 \Gamma \vdash t_1 \equiv t_2 : \cbox{\Phi_1 \vdash \Phi'_1}
     & \Gamma ; \Phi \Vdash \sigma'_1 = \sigma'_2 : \Phi_1
     & \Gamma \vdash \Phi_1 \equiv \Phi_2 : \ctx
 \end{array}
 }
 {\Gamma ; \Phi \Vdash \sigma_1 = \sigma_2 : \phi}{}$}
 \prf{$\Gamma ; \Psi \vdash \lfs{\sigma}{\Phi}(\sigma_1) \lfwhnf \lfs{\sigma}{\Phi} (\sclo \phi {\unbox{t_1}{\sigma'_1}})  : \phi$  \hfill by Lemma \ref{lm:lfwhnfsub}}
 \prf{$\Gamma ; \Psi \vdash \lfs{\sigma'}{\Phi}(\sigma_2) \lfwhnf \lfs{\sigma'}{\Phi} (\sclo \phi {\unbox{t_2}{\sigma'_2}})  : \phi$   \hfill by Lemma \ref{lm:lfwhnfsub}}
 \prf{$\lfs{\sigma}{\Phi} (\sclo \phi {\unbox{t_1}{\sigma'_1}}) = \sclo \phi {\unbox{t_1}{\lfs{\sigma}{\phi}{\sigma'_1}}}$ \hfill by LF subst. def.}
 \prf{$\lfs{\sigma'}{\Phi} (\sclo \phi {\unbox{t_2}{\sigma'_2}}) = \sclo \phi {\unbox{t_2}{\lfs{\sigma'}{\phi}{\sigma'_2}}}$ \hfill by LF subst. def.}
 \prf{$\Gamma ; \Psi \Vdash \lfs{\sigma}{\Phi}\sigma'_1 = \lfs{\sigma'}{\Phi}\sigma'_2: \Phi_1$ \hfill by IH}
 \prf{$\Gamma ; \Psi \Vdash \lfs{\sigma}{\Phi}\sigma_1 = \lfs{\sigma'}{\Phi}\sigma_2:\phi$ \hfill by sem. equ. def.}
 \\[1em]
}
 \pcase{$
 \ianc {
 \Gamma ; \Phi \vdash \sigma_1 \lfwhnf \cdot : \cdot  \quad
 \Gamma ; \Phi \vdash \sigma_2 \lfwhnf \cdot : \cdot }
 {\Gamma ; \Phi \Vdash \sigma_1 = \sigma_2 : \cdot}{}
 $}
\prf{$\Gamma ; \Psi \vdash \lfs{\sigma}{\Phi} \sigma_1 \lfwhnf \lfs{\sigma}{\Phi} \cdot : \cdot$ \hfill by Lemma \ref{lm:lfwhnfsub}}
\prf{$\Gamma ; \Psi \vdash \lfs{\sigma}{\Phi} \sigma_2 \lfwhnf \lfs{\sigma}{\Phi} \cdot : \cdot$ \hfill by Lemma \ref{lm:lfwhnfsub}}
\prf{$\lfs{\sigma}\Phi \cdot = \cdot$ \hfill by LF subst. def.}
\prf{$\Gamma \Psi \Vdash \lfs{\sigma}{\Phi}\sigma_1 = \lfs{\sigma}{\Phi} \sigma_2 : \cdot$ \hfill by sem. eq. def.}
 \\[1em]
 \pcase{$\ianc
 {
 \Gamma ; \phi, \wvec{x{:}A} \vdash \sigma_1 \lfwhnf \wk{\phi} : \phi \qquad
 \Gamma ; \phi, \wvec{x{:}A} \vdash \sigma_2 \lfwhnf \wk{\phi} : \phi
 }
 {\Gamma ; \phi, \wvec{x{:}A} \Vdash \sigma_1 = \sigma_2 : \phi}{}
 $}
\prf{$\Gamma ; \Psi \vdash [\sigma / \phi, \vec x] \sigma_1 \lfwhnf [\sigma / \phi, \vec x] \wk{\phi} : \phi$ \hfill by Lemma \ref{lm:lfwhnfsub}}
\prf{$[\sigma / \phi, \vec x] \wk{\phi} = \trunc_{\phi} (\sigma / \phi, \vec x) = \sigma'_1$ where $\Gamma ; \Psi \vdash \sigma'_1 : \phi$ \hfill by LF subst. def.}
\prf{$\Gamma ; \Psi \vdash [\sigma' / \phi, \vec x] \sigma_2 \lfwhnf [\sigma' / \phi, \vec x] \wk{\phi} : \phi$ \hfill by Lemma \ref{lm:lfwhnfsub}}
\prf{$[\sigma' / \phi, \vec x] \wk{\phi} = \trunc_{\phi} (\sigma' / \phi, \vec x) = \sigma'_2$ where $\Gamma ; \Psi \vdash \sigma'_2 : \phi$ \hfill by LF subst. def.}
\prf{$\Gamma ; \Psi \Vdash \sigma'_1 = \sigma'_2 : \phi$ \hfill since $\Gamma ; \Psi \Vdash \sigma = \sigma' : \phi, \vec x$ }
 \\[1em]
 \pcase{$
 \ianc {\begin{array}{@{}ll@{}}
 \Gamma ; \Phi \vdash \sigma_1 \lfwhnf \sigma'_1, M : \Phi', x{:}A  & \\
 \Gamma ; \Phi \vdash \sigma_2 \lfwhnf \sigma'_2, N : \Phi', x{:}A &
 \Gamma ; \Phi \Vdash \sigma'_1 = \sigma'_2 : \Phi' \qquad \Gamma ; \Psi  \Vdash M = N : \lfs{\sigma'_1}{\Phi'} A
 \end{array}
 }
 {\Gamma ; \Phi \Vdash \sigma_1 = \sigma_2 : \Phi', x{:}A}{}
 $
 }
\prf{$\Gamma ; \Psi \vdash \lfs{\sigma}{\Phi}\sigma_1 \lfwhnf \lfs{\sigma}\Phi {(\sigma'_1,M)} :  \Phi', x{:}A$ \hfill by Lemma \ref{lm:lfwhnfsub}}
\prf{$\Gamma ; \Psi \vdash \lfs{\sigma'}{\Phi}\sigma_2 \lfwhnf \lfs{\sigma'}\Phi {(\sigma'_2,N)} :  \Phi', x{:}A$ \hfill by Lemma \ref{lm:lfwhnfsub}}
\prf{$\lfs{\sigma'}\Phi {(\sigma'_2,N)} = \lfs{\sigma'}\Phi \sigma'_2,~\lfs{\sigma'}\Phi N$ \hfill by LF subst. def.}
\prf{$\lfs{\sigma}\Phi {(\sigma'_1,M)} = \lfs{\sigma}\Phi \sigma'_1,~\lfs{\sigma}\Phi M$ \hfill by LF subst. def.}
\prf{$\Gamma ; \Psi \Vdash \lfs{\sigma}\Phi \sigma'_1 = \lfs{\sigma'}\Phi \sigma'_2 :  \Phi'$ \hfill by IH}
\prf{$\Gamma ; \Psi \Vdash \lfs{\sigma}\Phi M = \lfs{\sigma'}\Phi N : \lfs{\sigma}{\Phi}(\lfs{\sigma'_1}{\Phi'}A)$ \hfill by IH}
\prf{$\Gamma ; \Psi \lfs\sigma\Phi \sigma_1 = \lfs{\sigma'}\Phi \sigma_2 : \Phi', x{:}A$ \hfill by sem. equ. def.}
\\[1em]
}
 \end{proof}
}
\LONGVERSION{
 \begin{lemma}[Semantic Weakening Substitution Exist]\label{lm:semlfwk}\quad \\
If $\Gamma ; \Psi, \wvec{x{:}A} \vdash \wk{\hatctx\Psi} : \Psi$
then $\Gamma ; \Psi, \wvec{x{:}A} \Vdash \wk{\hatctx\Psi} = \wk{\hatctx\Psi} : \Psi$.
 \end{lemma}
 \begin{proof}
By induction on the LF context $\Psi$.
\LONGVERSIONCHECKED{
\pcase{$\Psi = \cdot$.}
\prf{$\Gamma ; \cdot, \wvec{x{:}A} \vdash \wk\cdot : \cdot$ \hfill by  assumption}
\prf{$\Gamma ; \cdot, \wvec{x{:}A} \vdash \wk\cdot \whnf \cdot : \cdot$ \hfill by $\whnf$ rule and typing}
\prf{$\Gamma ; \cdot, \wvec{x{:}A} \Vdash \wk\cdot = \wk\cdot : \cdot$ \hfill by semantic. def.}

\pcase{$\Psi = \Psi', y{:}B$}
\prf{$\Gamma; \Psi', y{:}B, \wvec{x{:}A} \vdash\wk{\hatctx\Psi',y} : \Psi', y{:}B$ \hfill by assumption}
\prf{$\Gamma ; \Psi', y{:}B, \wvec{x{:}A} \vdash \wk{\hatctx\Psi'} : \Psi'$ \hfill by typing }
\prf{$\Gamma ; \Psi', y{:}B, \wvec{x{:}A} \Vdash \wk{\hatctx\Psi'} = \wk{\hatctx\Psi'} : \Psi'$ \hfill by IH}
\prf{$\Gamma ; \Psi', y{:}B, \wvec{x{:}A} \Vdash y = y : B$ \hfill by semantic eq. for LF terms, the fact that $\norm x$, and $B = \lfs{\wk{\hatctx\Psi'}}{\Psi'} B$}
\prf{$\Gamma ; \Psi', y{:}B, \wvec{x{:}A} \vdash \wk{\hatctx\Psi',y} \whnf \wk{\hatctx\Psi'}, y : \Psi', y{:}B$ \hfill by $\whnf$ and typing rules}
\prf{$\Gamma ; \Psi', y{:}B, \wvec{x{:}A} \vdash \wk{\hatctx\Psi',y} = \wk{\hatctx\Psi',y} : \Psi', y{:}B$ \hfill by sem. eq. for LF substitutions}

\pcase{$\Psi = \psi$}
\prf{$\Gamma ; \psi, \wvec{x{:}A} \vdash \wk\psi : \psi$ \hfill by assumption}
\prf{$\Gamma ; \psi, \wvec{x{:}A} \vdash \wk\psi : \psi$ \hfill by $\whnf$ and typing and the fact that $\norm \wk\psi$}
\prf{$\Gamma ; \psi, \wvec{x{:}A} \vdash \wk\psi = \wk\psi: \psi$  \hfill by sem. eq. for LF subst.}
}
 \end{proof}
}
\LONGVERSION{
\begin{lemma}[Semantic LF Context Conversion]\label{lm:semlfctxconv}
\quad
\begin{enumerate}
\item If $\Gamma ; \Psi, x{:}A_1 \Vdash M = N : B$ and $\Gamma ; \Psi \vdash A_1 \equiv A_2 : \lftype$
then $\Gamma ; \Psi, x{:}A_2 \Vdash M = N : B$

\item If $\Gamma ; \Psi, x{:}A_1 \Vdash \sigma = \sigma' : \Phi$ and $\Gamma ; \Psi \vdash A_1 \equiv A_2 : \lftype$
then $\Gamma ; \Psi, x{:}A_2 \Vdash \sigma = \sigma' : \Phi$.

\end{enumerate}
\end{lemma}
\begin{proof} The idea is to use  $\Gamma ; \Psi \vdash A_1 \equiv A_2 : \lftype$  and build LF weakening substitutions
$\Gamma ; \Psi,  x{:}A_2, y{:}A_1 \Vdash \wk{\hatctx\Psi}, y = \wk{\hatctx\Psi}, y : \Psi, x{:}A_1$
and
$\Gamma ; \Psi,  x{:}A_2 \Vdash \wk{\hatctx\Psi}, x, x = \wk{\hatctx\Psi}, x, x : \Psi, x{:}A_2, y{:}A_1$.
Using semantic LF subst. (Lemma \ref{lem:semlfeqsub}), we can then
move $\Gamma ; \Psi, x{:}A_1 \Vdash M = N : B$ to the new LF context $\Psi, x{:}A_2$.
~
\LONGVERSIONCHECKED{\\[0.5em]
 (1): \fbox{ If $\Gamma ; \Psi, x{:}A_1 \Vdash M = N : B$ and $\Gamma ; \Psi \vdash A_1 \equiv A_2 : \lftype$
then $\Gamma ; \Psi, x{:}A_2 \Vdash M = N : B$}
\\[1em]
\prf{$\Gamma \vdash \Psi : \ctx$ \hfill by Well-Formedness of Sem. LF
  Equ. (Lemma \ref{lm:semlfwf}) \\
\mbox{$\quad$}\hfill and Well-formedness of LF context (Lemma \ref{lm:lfctxwf}) }
 \prf{$\Gamma \vdash \Psi, x{:}A_2 : \ctx$ \hfill by context well-formedness rules}
 \prf{$\Gamma ; \Psi \vdash A_2 \equiv A_1 : \lftype$ \hfill by symmetry}
 \prf{$\Gamma ; \Psi, x{:}A_2 \vdash A_2 \equiv A_1 : \lftype$ \hfill by LF weakening  }
 \prf{$\Gamma ; \Psi, x{:}A_2 \vdash x : A_2$ \hfill by typing rule using $\Gamma \vdash \Psi, x{:}A_2 : \ctx$}
 \prf{$\Gamma ; \Psi, x{:}A_2 \vdash x : A_1$ \hfill conversion using $\Gamma ; \Psi, x{:}A_2 \vdash A_2 \equiv A_1 : \lftype$}
 \prf{$\Gamma ; \Psi, x{:}A_2 \vdash x : \lfs{\wk{\hatctx\Psi}}{\Psi} A_1$ \hfill as $A_1 = \lfs{\wk{\hatctx\Psi}}{\Psi} A_1$ }
 \prf{$\Gamma ; \Psi, x{:}A_2 \vdash \wk{\hatctx{\Psi}}, x, x : \Psi, x{:}A_2, y{:}A_1$ \hfill by typing rules for LF substitution}
 \prf{$\Gamma ; \Psi, x{:}A_2, y{:}A_1 \vdash \wk{\hatctx\Psi}, y : \Psi, x{:}A_1$ \hfill by typing rule for LF substitution}
\prf{$\Gamma ; \Psi, x{:}A_2, y{:}A_1 \vdash \wk{\hatctx\Psi} : \Psi$ \hfill by typing}
\prf{$\Gamma ; \Psi, x{:}A_2, y{:}A_1 \Vdash \wk{\hatctx\Psi} = \wk{\hatctx\Psi} : \Psi$ \hfill by Lemma \ref{lm:semlfwk}}
\prf{$\Gamma ; \Psi,  x{:}A_2, y{:}A_1 \Vdash \wk{\hatctx\Psi}, y = \wk{\hatctx\Psi}, y : \Psi, x{:}A_1$
   \hfill by sem. equ. for LF subst. \\\mbox{\hspace{1cm}}\hfill using the fact that $\norm y$}
\prf{$\Gamma ; \Psi, x{:}A_2 \Vdash \wk{\hatctx\Psi} = \wk{\hatctx\Psi} : \Psi$ \hfill by Lemma \ref{lm:semlfwk}}
\prf{$\Gamma ; \Psi,  x{:}A_2 \Vdash \wk{\hatctx\Psi}, x, x = \wk{\hatctx\Psi}, x, x : \Psi, x{:}A_2, y{:}A_1$
   \hfill by sem. equ. for LF subst. \\ \mbox{\hspace{1cm}}\hfill using the fact that $\norm x$}
 \prf{$\Gamma ; \Psi, x{:}A_2 \Vdash \lfs{\wk{\hatctx{\Psi}}, x, x}  {\Psi, x, y}M' = \lfs{\wk{\hatctx{\Psi}}, x, x}{\Psi, x, y}N' :
   \lfs{\wk{\hatctx{\Psi}}, x, x}{\Psi, x, y} B$\\
\mbox{$\quad$}\hfill where $M' = \lfs{\wk{\hatctx\Psi}, y}{\Psi, x}M$
                    and $N' = \lfs{\wk{\hatctx\Psi}, y}{\Psi, x}N$  by semantic LF subst. (Lemma \ref{lem:semlfeqsub} twice)}
 \prf{$\lfs{\wk{\hatctx{\Psi}}, x, x}{\Psi, x, y}(\wk{\hatctx\Psi}, y) = \wk{\hatctx\Psi}, x$ \hfill by subst. def.}
 \prf{$\Gamma ; \Psi, x{:}A_2 \vdash \lfs{\wk{\hatctx\Psi}, x}{\Psi, x} M = \lfs{\wk{\hatctx\Psi}, x}{\Psi, x} N : \lfs{\wk{\hatctx\Psi}, x}{\Psi, x} B$ \hfill by previous lines}
 \prf{$\Gamma ; \Psi, x{:}A_2 \Vdash M = N : B$ \hfill using the fact that $\lfs{\wk{\hatctx\Psi}, x}{\Psi, x} M = M$, etc.}
\\[1em]
We prove (2): \fbox{If $\Gamma ; \Psi, x{:}A_1 \Vdash \sigma = \sigma' : \Phi$ and $\Gamma ; \Psi \vdash A_1 \equiv A_2 : \lftype$
  then $\Gamma ; \Psi, x{:}A_2 \Vdash \sigma = \sigma' : \Phi$.}
\\[1em]
\prf{$\Gamma ; \Psi,  x{:}A_2, y{:}A_1 \Vdash \wk{\hatctx\Psi}, y = \wk{\hatctx\Psi}, y : \Psi, x{:}A_1$
   \hfill constructed as for case (1)}
\prf{$\Gamma ; \Psi,  x{:}A_2 \Vdash \wk{\hatctx\Psi}, x, x = \wk{\hatctx\Psi}, x, x : \Psi, x{:}A_2, y{:}A_1$
   \hfill constructed as for case (1)}
 \prf{$\Gamma ; \Psi, x{:}A_2 \Vdash \lfs{\wk{\hatctx{\Psi}}, x, x}  {\Psi, x, y}\sigma_1 = \lfs{\wk{\hatctx{\Psi}}, x, x}{\Psi, x, y}\sigma_2 : \Phi$ \\
\mbox{$\quad$}\hfill where $\sigma_1 = \lfs{\wk{\hatctx\Psi}, y}{\Psi, x} \sigma$
                    and $\sigma_2 = \lfs{\wk{\hatctx\Psi}, y}{\Psi, x}\sigma'$  by semantic LF subst. (Lemma \ref{lem:semlfeqsub} twice)}
 \prf{$\lfs{\wk{\hatctx{\Psi}}, x, x}{\Psi, x, y}(\wk{\hatctx\Psi}, y) = \wk{\hatctx\Psi}, x$ \hfill by subst. def.}
 \prf{$\Gamma ; \Psi, x{:}A_2 \vdash \lfs{\wk{\hatctx\Psi}, x}{\Psi, x} \sigma = \lfs{\wk{\hatctx\Psi}, x}{\Psi, x} \sigma' : \Phi$ \hfill by previous lines}
 \prf{$\Gamma ; \Psi, x{:}A_2 \Vdash \sigma = \sigma' : \Phi$ \hfill
   using the fact that $\lfs{\wk{\hatctx\Psi}, x}{\Psi, x} \sigma =
   \sigma$, etc.}
}
\end{proof}
}

Our semantic definitions are reflexive, symmetric, and transitive. Further they are stable under type conversions. We state the lemma below only for terms, but it must in fact be proven mutually with the corresponding property for LF substitutions. We first establish these properties for LF and subsequently for computations.  Establishing these properties is tricky and intricate. All proofs can be found in the long version.

\begin{lemma}[Reflexivity, Symmetry, Transitivity, and Conversion of Semantic Equality for LF]\label{lem:semsymlf}
Let $\Psi \Vdash M_1 = M_2 : A$. Then:
\begin{enumerate}
\item \label{it:reflclf} 
$\Gamma ; \Psi \Vdash M_1 = M_1 : A$.
\item \label{it:symclf}
  $\Gamma ; \Psi \Vdash M_2 = M_1 : A$.
\item \label{it:transclf} 
  If 
  $\Gamma ; \Psi \Vdash M_2 = M_3 : A$
    then $\Gamma ; \Psi \Vdash M_1 = M_3 : A$.
\item \label{it:convclf} 
    If $\Gamma ; \Psi \vdash A \equiv A' : \lftype$
     then $\Gamma ; \Psi \Vdash M_1 = M_2 : A'$.
\end{enumerate}

\LONGVERSION{B.~For LF Substitutions:
\begin{enumerate}
\item \label{it:reflsub}
 $\Gamma ; \Psi \Vdash \sigma = \sigma: \Phi$.
\item \label{it:symsub} 
      If $\Gamma ; \Psi \Vdash \sigma = \sigma' : \Phi$ then $\Gamma ; \Psi \Vdash \sigma' = \sigma : \Phi$.
\item \label{it:transsub} 
    If $\Gamma ; \Psi \Vdash \sigma_1 = \sigma_2 : \Phi$ and $\Gamma ; \Psi \Vdash \sigma_2 = \sigma_3 : \Phi$
    then $\Gamma ; \Psi \Vdash \sigma_1 = \sigma_3 : \Phi$.
\item \label{it:convsub} (Conversion:)
    If\/ $\Gamma \vdash \Phi \equiv \Phi' : \ctx$ and $\Gamma ; \Psi \Vdash \sigma = \sigma' : \Phi$
    then $\Gamma \Vdash \sigma = \sigma' : \Phi'$.
\end{enumerate}}
\end{lemma}
\begin{proof}
Reflexivity follows directly from symmetry and transitivity. For LF terms (and LF substitutions), we prove symmetry and conversion by induction on the derivation $\Gamma ; \Psi \Vdash M = N : A$ and $\Gamma; \Psi \Vdash \sigma = \sigma' : \Phi$ respectively. For transitivity, we use lexicographic induction.
%
The proofs relies on symmetry of declarative equality ($\equiv$),
determinacy of weak head reductions, and crucially relies on
well-formedness of semantic equality and functionality of LF typing
(Lemma~\ref{lm:func-lftyping}).
\LONGVERSIONCHECKED{\\[0.5em]
\fbox{
Transitivity: If $\Gamma ; \Psi \Vdash \sigma_1 = \sigma_2 : \Phi$ and $\Gamma ; \Psi \Vdash \sigma_2 = \sigma_3 : \Phi$
    then $\Gamma ; \Psi \Vdash \sigma_1 = \sigma_3 : \Phi$.   }
\\[1em]
By lexicographic induction on the first two derivations;
\\
\pcase{$\ianc{\Gamma ; \Psi \vdash \sigma_1 \lfwhnf \cdot : \cdot  \quad
              \Gamma ; \Psi \vdash \sigma_2 \lfwhnf \cdot : \cdot }
             {\Gamma ; \Psi \Vdash \sigma_1 = \sigma_2 : \cdot}{}$
}
\\
\prf{$\Gamma ; \Psi \Vdash \sigma_2 = \sigma_3 : \cdot$ \hfill by assumption}
\prf{$\Gamma ; \Psi \vdash \sigma_2 \lfwhnf \cdot : \cdot $ \hfill by inversion and determinacy (Lemma \ref{lem:detwhnf})}
\prf{$\Gamma ; \Psi \vdash \sigma_3 \lfwhnf \cdot : \cdot$ \hfill by inversion}
\prf{$\Gamma ; \Psi \Vdash \sigma_1 = \sigma_3 : \cdot$
\hfill using $\Gamma ; \Psi \vdash \sigma_1 \lfwhnf \cdot : \cdot$
        and $\Gamma ; \Psi \vdash \sigma_3 \lfwhnf \cdot : \cdot$}

\pcase{$\ianc{\Gamma ; \Psi, \wvec{x{:}A} \vdash \sigma_1 \lfwhnf \wk{\hatctx\Psi} : \Psi \quad
              \Gamma ; \Psi, \wvec{x{:}A} \vdash \sigma_2 \lfwhnf \wk{\hatctx\Psi} : \Psi}
             {\Gamma ; \Psi \Vdash \sigma_1 = \sigma_2 : \Psi}{}$
}
\prf{$\Gamma ; \Psi \vdash \sigma_2 \lfwhnf \cdot : \cdot $ \hfill by assumption}
\prf{$ \Gamma ; \Psi, \wvec{x{:}A} \vdash \sigma_2 \lfwhnf \wk{\hatctx\Psi} : \Psi$ \hfill  by inversion and determinacy (Lemma \ref{lem:detwhnf})}
\prf{$ \Gamma ; \Psi, \wvec{x{:}A} \vdash \sigma_3 \lfwhnf \wk{\hatctx\Psi} : \Psi$ \hfill by inversion}
\prf{$\Gamma \Vdash \sigma_1 = \sigma_3 : \cdot$
\hfill using $\Gamma ; \Psi, \wvec{x{:}A} \vdash \sigma_1 \lfwhnf \wk{\hatctx\Psi} : \Psi$
         and $\Gamma ; \Psi, \wvec{x{:}A} \vdash \sigma_3 \lfwhnf \wk{\hatctx\Psi} : \Psi$}

\pcase{$\ianc
  {\begin{array}{@{}ll@{}}
\Gamma ; \Psi \vdash \sigma_1 \lfwhnf \sigma'_1, M : \Phi, x{:}A  & \\
\Gamma ; \Psi \vdash \sigma_2 \lfwhnf \sigma'_2, N : \Phi, x{:}A &
\Gamma ; \Psi \Vdash \sigma'_1 = \sigma'_2 : \Phi \qquad \Gamma ; \Psi  \Vdash M = N : \lfs{\sigma'_1}{\Phi} A
\end{array}
}
{\Gamma ; \Psi \Vdash \sigma_1 = \sigma_2 : \Phi, x{:}A}{}$}
\prf{$\Gamma ; \Psi \Vdash \sigma_2 = \sigma_3 : \Phi, x{:}A$ \hfill by assumption}
\prf{$\Gamma ; \Psi \vdash \sigma_2 \lfwhnf \sigma'_2, N : \Phi, x{:}A$ \hfill by inversion and determinacy (Lemma \ref{lem:detwhnf})}
\prf{$\Gamma ; \Psi \vdash \sigma_3 \lfwhnf \sigma'_3, Q : \Phi, x{:}A$ \hfill by inversion}
\prf{$\Gamma ; \Psi  \Vdash N = Q : \lfs{\sigma_2}{\Phi} A$ \hfill inversion }
\prf{$\Gamma ; \Psi \Vdash \sigma'_2 = \sigma'_3 : \Phi, x{:}A$ \hfill by inversion}
\prf{$\Gamma ; \Psi \Vdash \sigma'_1 = \sigma'_3 : \Phi, x{:}A$ \hfill by IH}
\prf{$\Gamma ; \Psi \vdash \sigma'_2, N: \Phi, x{:}A$ \hfill by def. of well-typed whnf}
\prf{$\Gamma \vdash \Phi, x{:}A : \ctx$ \hfill by well-formedness of LF typing}
\prf{$\Gamma ; \Phi \vdash A : \lftype$ \hfill by well-formedness of LF contexts}
\prf{$\Gamma ; \Psi \vdash \sigma_1 \equiv \sigma_2 : \Phi$ \hfill by well-formedness of semantic equality (Lemma \ref{lm:semlfwf}) }
\prf{$\Gamma ; \Psi \vdash \lfs{\sigma_1}{\Phi} A \equiv \lfs{\sigma_2}{\Phi} A : \lftype$ \hfill by functionality of LF typing (Lemma \ref{lm:func-lftyping})}
\prf{$\Gamma ; \Psi \Vdash N = Q : \lfs{\sigma_1}{\Phi}A$ \hfill by IH (Conversion \ref{it:convclf})}
\prf{$\Gamma ; \Psi \Vdash M \equiv Q \lfs{\sigma_1}{\Phi} A$ \hfill by IH}
\prf{$\Gamma ; \Psi \Vdash \sigma_1 = \sigma_3 : : \Phi, x{:}A$ \hfill by sem. def. }

\SUBSTCLO{
\pcase{$\ianc {
      \begin{array}{@{}lll@{}}
      \Gamma ; \Psi \vdash \sigma_1 \lfwhnf (\sclo {\hatctx{\Phi}} {\unbox{t_1}{\sigma'_1})}  : \Phi
   & \typeof (\Gamma \vdash t_1) = \cbox{\Phi_1 \vdash \Phi'_1}  & \Gamma \vdash \Phi'_1 \equiv (\Phi, \wvec{x{:}A}) : \ctx
\\
\Gamma ; \Psi \vdash \sigma_2 \lfwhnf (\sclo {\hatctx\Phi} {\unbox{t_2}{\sigma'_2}}) : \Phi
    & \typeof (\Gamma \vdash t_2) = \cbox{\Phi_2 \vdash \Phi'_2} & \Gamma \vdash \Phi'_2 \equiv (\Phi, \wvec{x{:}A}) : \ctx
\\
\Gamma \vdash t_1 \equiv t_2 : \cbox{\Phi_1 \vdash \Phi'_1}
    & \Gamma ; \Psi \Vdash \sigma'_1 = \sigma'_2 : \Phi_1
    & \Gamma \vdash \Phi_1 \equiv \Phi_2 : \ctx
\end{array}
}
{\Gamma ; \Psi \Vdash \sigma_1 = \sigma_2 : \Phi}
{}$
}
\prf{$\Gamma ; \Psi \vdash \sigma_2 = \sigma_3 : \Phi$ \hfill by assumption}
\prf{$\Gamma ; \Psi \vdash \sigma_2 \lfwhnf (\sclo {\hatctx{\Phi}} {\unbox{t_2}{\sigma'_2})}  : \Phi$  \hfill by inversion and determinacy (Lemma \ref{lem:detwhnf})}
\prf{$\Gamma ; \Psi \vdash \sigma_3 \lfwhnf (\sclo {\hatctx{\Phi}} {\unbox{t_3}{\sigma'_3})}  : \Phi$  \hfill by inversion}
\prf{$\typeof (\Gamma \vdash t_2) = \cbox{\Phi_2 \vdash \Phi'_2}$ \hfill by inversion and uniqueness of $\typeof$}
\prf{$\typeof (\Gamma \vdash t_3) = \cbox{\Phi_3 \vdash \Phi'_3}$ \hfill by inversion}
\prf{$\Gamma ; \Psi \Vdash \sigma'_2 = \sigma'_3 : \Phi_2  $ \hfill by inversion}
\prf{$\Gamma \vdash t_2 \equiv t_3 : \cbox{\Phi_2 \vdash\Phi'_2} $ \hfill by inversion}
\prf{$\Gamma \vdash \Phi_2 \equiv \Phi_3 : \ctx $ \hfill by inversion}
\prf{$\Gamma \vdash \Phi'_2 \equiv (\Phi, \wvec{x{:}A'}) : \ctx $ \hfill by inversion}
\prf{$\Gamma \vdash \Phi'_3 \equiv (\Phi, \wvec{x{:}A'}) : \ctx $ \hfill by inversion}
\prf{$\Gamma \vdash (\Phi, \wvec{x{:}A})  \equiv (\Phi, \wvec{x{:}A'}) : \ctx$ \hfill by transitivity and symmetry}
\prf{$\Gamma \vdash \Phi'_3 \equiv (\Phi, \wvec{x{:}A}) :\ctx$ \hfill by transitivity and symmetry}
\prf{$\Gamma \vdash \Phi_1 \equiv \Phi_3 : \ctx$ \hfill by transitivity}
\prf{$\Gamma ; \Psi \Vdash \sigma'_2 = \sigma'_3 : \Phi_1$ \hfill by IH (\ref{it:convsub}) using $\Gamma \vdash \Phi_1 \equiv \Phi_2 : ctx$}
\prf{$\Gamma ; \Psi \Vdash \sigma'_1 = \sigma'_3 : \Phi_1$ \hfill by IH (\ref{it:transsub})}
\prf{$\Gamma \vdash t_1 \equiv t_3 : \cbox{\Phi_2 \vdash\Phi'_2} $ \hfill by transitivity}
\prf{$\Gamma ; \Psi \Vdash \sigma_1 = \sigma_3 : \Phi$ \hfill by def. of sem. typing}
}
\fbox{
Symmetry: If $\Gamma ; \Psi \Vdash \sigma_1 = \sigma_2 : \Phi$ then $\Gamma ; \Psi \Vdash \sigma_2 = \sigma_1 : \Phi$}
\\[1em]
By induction on semantic equivalence relation where we consider any $\sigma'$ smaller than $\sigma$ if $\sigma \lfwhnf \sigma'$; the proof is mostly straightforward exploiting symmetry of decl. equivalence ($\equiv$), but also relies again on  well-formedness of semantic equality (Lemma \ref{lm:semlfwf}) and  functionality of LF typing (Lemma \ref{lm:func-lftyping}) for the case where $\sigma_i \lfwhnf \sigma'_i, M_i$. We show the interesting case.
\\[1em]
\pcase{$\ianc
  {\begin{array}{@{}ll@{}}
\Gamma ; \Psi \vdash \sigma_1 \lfwhnf \sigma'_1, M : \Phi, x{:}A  & \\
\Gamma ; \Psi \vdash \sigma_2 \lfwhnf \sigma'_2, N : \Phi, x{:}A &
\Gamma ; \Psi \Vdash \sigma'_1 = \sigma'_2 : \Phi \qquad \Gamma ; \Psi  \Vdash M = N : \lfs{\sigma'_1}{\Phi} A
\end{array}
}
{\Gamma ; \Psi \Vdash \sigma_1 = \sigma_2 : \Phi, x{:}A}{}$}
\\[1em]
\prf{$\Gamma ; \Psi \Vdash \sigma'_2 = \sigma'_1 : \Phi$ \hfill by IH}
\prf{$\Gamma \vdash \Phi, x{:}A : \ctx$ \hfill by well-formedness of typing }
\prf{$\Gamma ; \Phi \vdash A : \lftype$ \hfill by well-formedness of LF contexts}
\prf{$\Gamma ; \Psi \vdash \sigma'_1 \equiv \sigma'_2 : \Phi$ \hfill by well-formedness of semantic equality (Lemma \ref{lm:semlfwf}) }
\prf{$\Gamma ; \Psi \vdash \lfs{\sigma_1}{\Phi} A \equiv \lfs{\sigma_2}{\Phi} A : \lftype$ \hfill by functionality of LF typing (Lemma \ref{lm:func-lftyping})}
\prf{$\Gamma ; \Psi \Vdash M = N : \lfs{\sigma_2}{\Phi}A$ \hfill by IH (Conversion \ref{it:convclf})}
\prf{$\Gamma ; \Psi \Vdash N = M : \lfs{\sigma_2}{\Phi}A$ \hfill by IH }
\prf{$\Gamma ; \Psi \Vdash \sigma_2 = \sigma_1 : \Phi, x{:}A$ \hfill by def. of semantic equivalence}

\vspace{1cm}
We now consider some cases for establishing symmetry and transitivity for semantic equality of LF terms.
\\[1em]
\pcase{$\ianc{
\begin{array}{@{}lll@{}}
\Gamma ; \Psi \vdash M_1 \lfwhnf \unbox {t_1} {\sigma_1} : \tm & \typeof(\Gamma \vdash t_1) = \cbox{\Phi_1 \vdash \tm}
& \Gamma \vdash \Phi_1 \equiv \Phi_2 : \ctx \\
\Gamma ; \Psi \vdash M_2 \lfwhnf \unbox {t_2} {\sigma_2} : \tm &
\typeof(\Gamma \vdash t_2) = \cbox{\Phi_2 \vdash \tm} &
\Gamma \vdash t_1 \equiv t_2 : \cbox{\Phi_1 \vdash \tm} \quad
\Gamma \Vdash \sigma'_1 = \sigma'_2 : \Phi_1
\end{array} }
            {\Gamma ; \Psi \Vdash M_1 = M_2 : \tm}{}$}
Symmetry for LF Terms.\\
\prf{$\Gamma ; \Psi \Vdash \sigma'_2 = \sigma'_1 : \Phi_1$ \hfill by IH}
\prf{$\Gamma ; \Psi \Vdash \sigma'_2 = \sigma'_1 : \Phi_2$ \hfill by IH (Conversion (\ref{it:convsub}))}
\prf{$\Gamma \vdash \cbox{\Phi_1 \vdash \tm} \equiv \cbox{\Phi_2 \vdash \tm} : u$ \hfill since $\Gamma \vdash \Phi_1 \equiv \Phi_2 : ctx$}
\prf{$\Gamma \vdash t_1 \equiv t_2 : \cbox{\Phi_2 \vdash \tm}$ \hfill by conversion  using $\Gamma \vdash \cbox{\Phi_1 \vdash \tm} \equiv \cbox{\Phi_2 \vdash \tm} : u$}
\prf{$\Gamma \vdash t_2 \equiv t_1 : \cbox{\Phi_2 \vdash \tm}$ \hfill by transitivity of $\equiv$}
\prf{$\Gamma ; \Psi \Vdash N = M : A$ \hfill by sem. def.}
\\
Transitivity for LF Terms.\\
\prf{$\Gamma \Vdash M_2 = M_3 : \tm$ \hfill by assumption}
\prf{$\Gamma \vdash M_2 \lfwhnf \unbox{t_2}{\sigma_2} : \tm$ \hfill by inversion and  determinacy (Lemma \ref{lem:detwhnf})}
\prf{$\Gamma \vdash M_3 \lfwhnf \unbox{t_3}{\sigma_3}$  \hfill by inversion}
\prf{$\typeof (\Gamma \vdash t_2) = \cbox{\Phi_2 \vdash \tm}$ \hfill by inversion and uniqueness of $\typeof$}
\prf{$\typeof (\Gamma \vdash t_3) = \cbox{\Phi_3 \vdash \tm}$ \hfill by inversion}
\prf{$\Gamma \vdash \Phi_2 \equiv \Phi_3 : \ctx$ \hfill by inversion}
\prf{$\Gamma \vdash \Phi_1 \equiv \Phi_3 : \ctx$ \hfill by transitivity ($\equiv$)}
\prf{$\Gamma \Vdash \sigma_2 = \sigma_3 : \Phi_2$ \hfill by inversion}
\prf{$\Gamma \Vdash \sigma_2 = \sigma_3 : \Phi_1$ \hfill by IH (Conversion \ref{it:convsub} using $\Gamma \Vdash \Phi_2 = \Phi_1 : \ctx$)}
\prf{$\Gamma \Vdash \sigma_1 = \sigma_3 : \Phi_1$ \hfill by IH}
\prf{$\Gamma \vdash t_2 \equiv t_3 : \cbox{\Phi_2 \vdash \tm}$ \hfill by inversion on $\Gamma ; \Psi \Vdash M_2 = M_3 : \tm$}
\prf{$\Gamma \vdash (\Phi_1 \vdash \tm) \equiv (\Phi_2 \vdash \tm) : u$ \hfill since $\Gamma \vdash \Phi_1 \equiv \Phi_2 : ctx$}
\prf{$\Gamma \vdash t_2 \equiv t_3 : \cbox{\Phi_1 \vdash \tm}$ \hfill by type conversion using $\Gamma \vdash (\Phi_1 \vdash \tm) \equiv (\Phi_2 \vdash \tm) : u$)}
\prf{$\Gamma \Vdash t_1 \equiv t_3 : \cbox{\Phi_1 \vdash \tm}$ \hfill by Transitivity of $\equiv$}
\prf{$\Gamma \Vdash M_1 = M_3 : \tm$}
}
\LONGVERSIONCHECKED{
\\[1em]
We concentrate here on proving the conversion properties:
\\[1em]
\fbox{(Conversion:) If\/ $\Gamma ; \Psi \vdash A \equiv A' : \lftype$
     and $\Gamma ; \Psi \Vdash M = N : A$
     then $\Gamma ; \Psi \Vdash M = N : A'$.}
\\[1em]
\pcase{
$\ianc
{ \Gamma ; \Psi \vdash M \lfwhnf \lambda x.M' : \Pi x{:}A.B \quad
  \Gamma ; \Psi \vdash N \lfwhnf \lambda x.N' : \Pi x{:}A.B \quad
 \Gamma ; \Psi, x{:}A \Vdash M' = N': B
}
{\Gamma ; \Psi \Vdash M = N: \Pi x{:}A. B}{}$
}
\prf{$\Gamma ; \Psi \vdash \Pi x{:}A.B \equiv \Pi x{:}A'.B' : \lftype$ \hfill by assumption}
\prf{$\Gamma ; \Psi, x{:}A \vdash B \equiv B' : \lftype$ and $\Gamma ; \Psi \vdash A \equiv A' : \lftype$ \hfill by injectivity of $\Pi$-types (Lemma \ref{lm:lfpi-inj})}
\prf{$\Gamma; \Psi, x{:}A \Vdash M' = N' : B'$ \hfill by IH}
\prf{$\Gamma ; \Psi, x{:}A' \Vdash M' = N' : B'$ \hfill by Semantic LF context conversion (Lemma \ref{lm:semlfctxconv})}
\prf{$\Gamma ; \Psi \Vdash M = N : \Pi x{:}A'.B'$ \hfill by sem. def.}
\\[1em]
Other cases are trivial since they are at type $\tm$.
}
\end{proof}

\subsection{Semantic Properties of Computations}
If a term is semantically well-typed, then it is also syntactically well-typed. Furthermore, our definition of semantic typing is stable under weakening.
\LONGVERSION{
 \begin{lemma}[Well-Formedness of Semantic Typing] \quad \label{lm:semwf}
    If $\Gamma \Vdash t = t' : \ann\tau$ then
         $\Gamma \vdash t : \ann\tau$ and $\Gamma \vdash t' : \ann\tau$ and $\Gamma \vdash t \equiv t' : \ann\tau$.
\end{lemma}
\begin{proof}
By induction on the induction on $\Gamma \Vdash \ann\tau : u$. . In each case, we refer the Def.~\ref{def:typedwhnf}.
\end{proof}

\begin{lemma}[Semantic Weakening for Computations]\label{lem:compsemweak}\quad
\begin{enumerate}
  \item \label{it:weaksemtau} If $\Gamma \Vdash \ann\tau :u$ and $\Gamma' \leq_\rho \Gamma$ then $\Gamma' \Vdash \{\rho\}\ann\tau : u$.
  \item If\/ $\Gamma \Vdash \ann\tau = \ann\tau' : u$ and $\Gamma' \leq_\rho \Gamma$
    then $\Gamma' \Vdash \{\rho\}\ann\tau = \{\rho\}\ann\tau' : u$.
  \item  If\/ $\Gamma \Vdash t = t' : \ann\tau$ and $\Gamma' \leq_\rho \Gamma$
    then $\Gamma' \Vdash \{\rho\}t = \{\rho\} t : \{\rho\}\ann\tau$.
  \end{enumerate}
\end{lemma}
\begin{proof}
By induction on $\Gamma \Vdash \ann\tau : u$.
\LONGVERSIONCHECKED{
\\[0.25em]
We note that the theorem is trivial for $\ann\tau = \tmctx$. Hence we concentrate on proving it where $\ann\tau = \tau$ (i.e. it is a proper type).
\\[0.25em]
For better and easier readability we simply write for example
 $\tau = (y:\ann\tau_1) \arrow \tau_2$  instead of
\\
$\Gamma \Vdash \tau : u$ where
\begin{enumerate}
\item $\Gamma \vdash \tau \whnf (x:\ann\tau_1) \arrow \tau_2$
\item $(\forall \Gamma' \leq_\rho \Gamma. \Gamma' \Vdash \{\rho\}\ann\tau_1 : u_1)$
\item $\forall \Gamma'\leq_\rho \Gamma.~ \Gamma' \Vdash s = s' :\{\rho\}\ann\tau_1
    \Longrightarrow \Gamma' \Vdash \{\rho, s/x\} \tau_2 = \{\rho,s'/x\}\tau_2 : u_2$ and
\item $(u_1, u_2, u_3) \in \Ru$.
\end{enumerate}

\noindent
Weakening of semantic typing $\Gamma \Vdash \tau : u$: \\
By case analysis on $\Gamma \Vdash \tau : u$.
\\[0.5em]
    \pcase{$\tau  = (x: \ann\tau_1) \arrow \tau_2 $}
    %
    %
    \prf{$\Gamma' \vdash \{\rho\}\tau \whnf \{\rho\}((y: \ann\tau_1) \arrow \tau_2) : u_3$ \hfill by Lemma \ref{lem:weakwhnf} using $\Gamma \Vdash \tau : u$}
    \prf{$\Gamma' \vdash \{\rho\}\tau \whnf (x:\{\rho\}\ann\tau_1) \arrow \{\rho,x/x\}\tau_2: u_3$ \hfill by substitution def.}
\\
    \prf{Suppose that $\Gamma_1 \leq_{\rho_1} \Gamma'$}
    \prf{$\Gamma' \leq_\rho \Gamma$ \hfill by assumption}
    \prf{$\Gamma_1 \leq_{\{\rho_1\}\rho} \Gamma$ \hfill by composition of substitution}
    \prf{$\Gamma_1 \Vdash \{\{\rho_1\}\rho\}\ann\tau_1 : u_1$  \hfill by $\Gamma \Vdash \tau : u$}
    \prf{$\Gamma_1 \Vdash \{\rho_1\}\{\rho\}\ann\tau_1  : u_1$ \hfill by composition of substitution}
\\
    \prf{Suppose that $\Gamma_1 \leq_{\rho_1} \Gamma'$ and $\Gamma_1 \Vdash s = s' : \{\rho_1\}(\{\rho\}\ann\tau_1)$}
    \prf{$\Gamma_1 \leq_{\{\rho_1\}\rho} \Gamma$ \hfill by composition of substitution}
    \prf{$\Gamma_1 \Vdash s = s' : \{\{\rho_1\}\rho\}\ann\tau_1$\hfill by composition of substitution}
    \prf{$\Gamma_1 \Vdash \{\{\rho_1\}\rho, s/x\} \tau_2 = \{\{\rho_1\}\rho, s'/x\}\tau_2: u_2$ \hfill by $\Gamma \Vdash \tau : u$}
    \prf{$\Gamma_1 \Vdash \{\rho_1, s/x\}(\{\rho,x/x\} \tau_2)= \{\rho_1, s'/x\}(\{\rho,x/x\} \tau_2) : u_2$ \hfill by composition of substitution}
\\
    \prf{$\Gamma' \Vdash \{\rho\}\tau : u_3$ \hfill by abstraction, since $\Gamma_1, \rho_1$ where arbitrary}
\\
    \pcase{$\tau = \cbox{T}$}
    \prf{$\Gamma' \der\{\rho\}\tau \whnf \{\rho\}\cbox{T}:u$ \hfill by Lemma \ref{lem:weakwhnf}  using $\Gamma \Vdash \tau : u$}
    \prf{$\Gamma' \der\{\rho\}\tau \whnf \cbox{\{\rho\}T}:u$ \hfill by substitution def.}
    \prf{$\Gamma \semlf T = T : u$ \hfill by assumption}
    \prf{$\Gamma \vdash T \equiv T : u$\hfill by inversion}
    \prf{$\Gamma' \vdash \{\rho\}T \equiv \{\rho\}T : u$\hfill by substitution lemma}
    \prf{$\Gamma' \Vdash \{\rho\}\tau : u$ \hfill by def.}
\\
    %
     \pcase{$\Gamma \Vdash t = t' : \cbox {\Psi \vdash A}$}
     \prf{$\Gamma \vdash t \whnf w   : \cbox {\Psi \vdash A}$ \hfill by assumption}
     \prf{$\Gamma' \vdash \{\rho\}t \whnf \{\rho\}w   : \{\rho\}\cbox {\Psi \vdash A}$ \hfill by \ref{lem:weakwhnf}}
     \prf{$\Gamma \vdash t' \whnf w'   : \cbox {\Psi \vdash A}$ \hfill by assumption}
     \prf{$\Gamma' \vdash \{\rho\}t' \whnf \{\rho\}w' :  \{\rho\}\cbox{\Psi \vdash A}$ \hfill by \ref{lem:weakwhnf}}
     \prf{$\Gamma ; \Psi \Vdash \unbox w \id = \unbox{w'}{\id} : A$ \hfill by assumption}
     \prf{$\Gamma' ; \{\rho\}\Psi \Vdash \{\rho\}\unbox{w}{\id} = \{\rho\}\unbox{w'}{\id} : \{\rho\} A$ \hfill by IH}
     \prf{$\Gamma' \Vdash \{\rho\}t = \{\rho\}t' : \{\rho\}\cbox{\Psi \vdash A}$ \hfill by sem. def. using subst. properties}
\\
    \pcase{$\Gamma \Vdash t = t' : (y: \ann\tau_1) \arrow \tau_2$}
    \prf{$\Gamma \vdash t \whnf w : (y: \ann\tau_1) \arrow \tau_2$ \hfill by assumption}
    \prf{$\Gamma' \vdash \{\rho\}t \whnf \{\rho\}w :\{\rho\}((y: \ann\tau_1) \arrow \tau_2)$ \hfill by Lemma \ref{lem:weakwhnf}}
    \prf{$\Gamma \vdash t'\whnf w': (y: \ann\tau_1) \arrow \tau_2$ \hfill by assumption}
    \prf{$\Gamma' \vdash \{\rho\}t' \whnf \{\rho\}w' :\{\rho\}((y: \ann\tau_1) \arrow \tau_2)$ \hfill by Lemma \ref{lem:weakwhnf}}
    \prf{Suppose that $\Gamma_1 \leq_{\rho_1} \Gamma'$ and $\Gamma_1 \vdash t_1 : \{\rho_1\}(\{\rho\}\ann\tau_1)$ and $\Gamma_1 \vdash t_2 : \{\rho_1\}(\{\rho\}\ann\tau_1)$}
    \prf{$\Gamma_1\leq_{\{\rho_1\}\rho}\Gamma'$ \hfill by definition}
    \prf{$\Gamma_1 \vdash t_1 : \{\{\rho_1\}\{\rho\}\ann\tau_1$\hfill by composition of substitution}
    \prf{$\Gamma_1 \vdash t_2 :\{ \{\rho_1\}\rho\}\ann\tau_1$ \hfill by composition of substitution}
    \prf{$\Gamma_1 \Vdash t_1 = t_2 : \{\{\rho_1\}\rho\}\ann\tau_1 \Longrightarrow \Gamma_1 \Vdash \{\{\rho_1\}\rho\}w~t_1 = \{\{\rho_1\}\rho\}w'~t_2 : \{\{\rho_1\}\rho, t_1/x\}\tau_2$ \hfill by assumption }
    \prf{$\Gamma_1 \Vdash t_1 = t_2 : \{\rho_1\}(\{\rho\}\ann\tau_1) \Longrightarrow \Gamma_1 \Vdash(\{\rho_1\}(\{\rho\}w))~t_1 = (\{\rho_1\}\{\rho\}w')~t_2 : \{\rho_1, t_1/x\}(\{\rho\}\tau_2)$ \hfill by composition of substitution}
    \prf{$\Gamma' \Vdash \{\rho\}t = \{\rho\}t' : \{\rho\}((y: \ann\tau_1) \arrow \tau_2)$ \hfill by def.}
}
\end{proof}

}%
%
%
Our semantic equality definition is symmetric and transitive. It is also reflexive -- however, note that we prove a weaker reflexivity statement which says that if $t_1$ is semantically equivalent to another term $t_2$ then it is also equivalent to itself. This suffices for our proofs. We also note that our semantic equality takes into account extensionality for terms at function types and contextual types; this is in fact baked into our semantic equality definition.

\begin{lemma}[Symmetry, Transitivity, and Conversion of Semantic Equality]\label{lem:semsym}
Let $\Gamma \Vdash \ann\tau : u$ and $\Gamma \Vdash \ann\tau' : u$ and $\Gamma \sem \ann\tau = \ann\tau' : u$ and $\Gamma \Vdash t_1 = t_2 : \ann\tau$.  Then:
\begin{enumerate}
\item\label{it:refl} (Reflexivity) $\Gamma \sem t_1 = t_1 : \ann\tau$.
\item\label{it:sym} (Symmetry) $\Gamma \Vdash t_2 = t_1 : \ann\tau$.
\item\label{it:trans} (Transitivity) If\/ $\Gamma \Vdash t_2 = t_3 : \ann\tau$ then
  $\Gamma \Vdash t_1 = t_3 : \ann\tau$.
\item\label{it:conv} (Conversion:) $\Gamma \Vdash t_1 = t_2 : \ann\tau'$.
\end{enumerate}
\end{lemma}
\begin{proof}
  Reflexivity follows directly from symmetry and transitivity.
  We prove symmetry and transitivity for terms using a lexicographic induction on $u$ and
  $\Gamma \Vdash \tau : u$; we appeal to the induction hypothesis and
  use the corresponding properties on types if
  the universe is smaller; if the universe stays the same,
  then we may appeal to the property for terms if $\Gamma \Vdash \tau
  : u$ is smaller; to prove conversion and symmetry for types, we may also appeal to the induction hypothesis if $\Gamma \Vdash \tau' : u$ is smaller.
\LONGVERSIONCHECKED{
  \\[1em]
\pcase{\fbox{\fbox{$\ann\tau = \tmctx$}}}
Symmetry for Terms (Prop. \ref{it:sym}): To Show: \fbox{$\Gamma \Vdash t_2 = t_1 : \tmctx$} \\
\prf{$t_2$ and $t_1$ stand for LF context $\Psi_2$ and $\Psi_1$ respectively}
\prf{$\Gamma \vdash \Psi_1 \equiv \Psi_2 : \tmctx$ \hfill by def. $\Gamma \Vdash t_1 = t_2 : \tau$}
\prf{$\Gamma \vdash \Psi_2 \equiv \Psi_1 : \tmctx$ \hfill by symmetry of $\equiv$}
\prf{$\Gamma \Vdash t_2 = t_2 : \tmctx$ \hfill by sem. equ. def.}
\\[0.5em]
Transitivity for Terms (Prop. \ref{it:trans}): To Show:
\fbox{If\/ $\Gamma \Vdash t_2 = t_3 : \tmctx $ then  $\Gamma \Vdash t_1 = t_3 : \tmctx$.}\\
\prf{$t_1$, $t_2$, and $t_3$ stand for LF context $\Psi_1$, $\Psi_2$, and $\Psi_3$ respectively}
\prf{$\Gamma \vdash \Psi_1 \equiv \Psi_2 : \tmctx$ \hfill by $\Gamma \Vdash t_1 = t_2 : \tmctx$}
\prf{$\Gamma \vdash \Psi_2 \equiv \Psi_3 : \tmctx$ \hfill by $\Gamma \Vdash t_2 = t_3 : \tmctx$}
\prf{$\Gamma \vdash \Psi_1 \equiv \Psi_3 : \tmctx$ \hfill by transitivity of $\equiv$}
\prf{$\Gamma \Vdash t_1 = t_3 : \tmctx$ \hfill by semantic equ. def.}
%
\\[0.5em]
Other cases are trivial.
\\[1.5em]
%
%
\pcase{\fbox{\fbox{$\tau = \cbox T$,~i.e. $\Gamma \vdash \tau \whnf \cbox T : u$
  and $\Gamma \Vdash_\LF T = T $ where $T = \Psi \vdash A$}}}
Symmetry for Terms (Prop. \ref{it:sym}): To Show: \fbox{$\Gamma \Vdash t_2 = t_1 : \cbox T$.}\\
\prf{$\Gamma \vdash t_1 \whnf w_1 : \tau$ \hfill by definition of  $\Gamma \Vdash t_1 = t_2 : \tau$}
\prf{$\Gamma \vdash t_2 \whnf w_2 : \tau$ \hfill by definition of  $\Gamma \Vdash t_1 = t_2 : \tau$}
\prf{\emph{Sub-Case}: $\Gamma \vdash t_1 \whnf \cbox{C} : \cbox{T}$
  and $\Gamma \vdash t_2 \whnf \cbox{C'} : \cbox T$
  and $\Gamma \Vdash_{\LF} C = C' : T$}
\prf{Consider $C = (\hatctx{\Psi} \vdash M)$ and $C' = (\hatctx{\Psi} \vdash N)$ and $T = \Psi \vdash A$ (proof is the same for case $C = (\hatctx{\Psi} \vdash \sigma)$)}
\prf{$\Gamma ; \Psi \Vdash M = N : A$ \hfill by $\Gamma \Vdash_\LF C = C' : T$}
\prf{$\Gamma ; \Psi \Vdash N = M : A$ \hfill by Lemma \ref{lem:semsymlf} (\ref{it:symclf})}
\prf{$\Gamma \Vdash_\LF C' = C : T$ \hfill by sem. def.}
\prf{$\Gamma \Vdash t_2 = t_1 : \cbox T$ \hfill by semantic equ. def.}
\\[-0.75em]
\prf{\emph{Sub-Case}: $\neut w_1, w_2$ and $\Gamma \vdash w_1 \equiv w_2 : \cbox{T}$}
\prf{$\Gamma \vdash w_2 \equiv w_1 : \cbox{T}$ \hfill by symmetry of $\equiv$}
\prf{$\Gamma \Vdash t_2 = t_1 : \cbox T$ \hfill by semantic equ. def.}
\\
Transitivity for Terms (Prop.\ref{it:trans}):
To Show \fbox{If\/ $\Gamma \Vdash t_2 = t_3 : \cbox T $ then  $\Gamma \Vdash  t_1 = t_3 : \cbox T$.}
 \\
\prf{$\Gamma \vdash t_1 \whnf w_1 : \tau$ \hfill by definition of  $\Gamma \Vdash t_1 = t_2 : \tau$}
\prf{$\Gamma \vdash t_2 \whnf w_2 : \tau$ \hfill by definition of  $\Gamma \Vdash t_1 = t_2 : \tau$}
\prf{$\Gamma \vdash t_2 \whnf w_2' : \tau$ \hfill by definition of  $\Gamma \Vdash t_2 = t_3 : \tau$}
\prf{$\Gamma \vdash t_3 \whnf w_3 : \tau$ \hfill by definition of  $\Gamma \Vdash t_2 = t_3 : \tau$}
\prf{$w_2 = w_2'$ \hfill by Lemma~\ref{lem:detwhnf} (\ref{it:comp-detwhnf})}
\prf{$\Gamma ; \Psi \Vdash \unbox{w_1}{\id} = \unbox{w_2}{\id} : A$ \hfill by def.  $\Gamma \Vdash t_1 = t_2 : \tau$}
\prf{$\Gamma ; \Psi \Vdash \unbox{w_2}{\id} = \unbox{w_3}{\id} : A$ \hfill by def.  $\Gamma \Vdash t_2 = t_3 : \tau$}
\prf{$\Gamma ; \Psi \Vdash \unbox{w_1}{\id} = \unbox{w_3}{\id} : A$ \hfill by Lemma~\ref{lem:semsymlf} (\ref{it:transclf})}
\prf{$\Gamma \Vdash t_1 = t_3 : \cbox{\Psi \vdash A}$ \hfill by sem. equ. def.}
\\
Symmetry for Types (Prop.\ref{it:symtyp}): To Show: \fbox{$\Gamma \sem \tau' =  \cbox T : u$ where $T = \Psi \vdash A$}\\
\prf{$\Gamma \vdash \tau' \whnf \cbox{T'} : u$ and $\Gamma \vdash T \equiv T'$
  \hfill by $\Gamma \Vdash \tau = \tau' : u$}
\prf{$\Gamma \vdash T' \equiv T$ \hfill by symmetry for LF equ.}
\prf{$\Gamma \Vdash \tau' = \tau : u$ \hfill by semantic equ. def.}
\\
Transitivity for Types (Prop.\ref{it:transtyp}): To Show:
\fbox{ If\/ $\Gamma \Vdash \tau' = \tau'' : u$ and $\Gamma \Vdash \tau'' : u$ then
$\Gamma \Vdash \cbox T = \tau'' : u$.}\\
\prf{$\Gamma \vdash \tau' \whnf \cbox{T'} : u$ and $\Gamma \vdash T \equiv T'$
  \hfill by $\Gamma \Vdash \tau = \tau' : u$}
\prf{$\Gamma \vdash \tau'' \whnf \cbox{T''} : u$ and $\Gamma \vdash T' \equiv T''$
  \hfill by $\Gamma \Vdash \tau' = \tau'' : u$}
\prf{$\Gamma \vdash T \equiv T''$ \hfill by transitivity for LF equ.}
\prf{$\Gamma \Vdash \tau = \tau'' : u$ \hfill by semantic equ. def.}
\\
Conversion for Terms (Prop.\ref{it:conv}): To Show:
\fbox{$\Gamma \Vdash t_1 = t_2 : \tau'$.} \\ 
\prf{$\Gamma \vdash t_1 \whnf w_1 : \cbox T$ \hfill by definition of  $\Gamma \Vdash t_1 = t_2 : \tau$}
\prf{$\Gamma \vdash t_2 \whnf w_2 : \cbox T$ \hfill by definition of  $\Gamma \Vdash t_1 = t_2 : \tau$}
\prf{$\Gamma \vdash \tau' \whnf \cbox{T'} : u$ and $\Gamma \vdash T \equiv T'$ \hfill by def. of $\Gamma \Vdash \tau = \tau' : u$}
\prf{$\Gamma \vdash \cbox{T} \equiv \cbox{T'} : u$ \hfill by decl. equ. def.}
\prf{$\Gamma \vdash \tau \equiv \cbox{T} : u$ \hfill by $\Gamma \vdash
  \tau \whnf \cbox{T}$ (since $\whnf$ rules are a subset of $\equiv$)}
\prf{$\Gamma \vdash \tau' \equiv \cbox{T'} : u$ \hfill by $\Gamma \vdash \tau' \whnf \cbox{T'}$ (since $\whnf$ rules are a subset of $\equiv$)}
\prf{$\Gamma \vdash \tau \equiv \tau' : u$ \hfill by transitivity and
  symmetry of decl. equality ($\equiv$)}
\prf{$\Gamma \vdash t_i : \tau'$ and $\Gamma \vdash w_i : \tau'$ for $i = 1,2$
  \hfill by typing rules using $\Gamma \vdash t_i : \cbox T$}
\prf{$\Gamma \vdash t_1 \whnf w_1 : \tau' $ and $\Gamma \vdash t_2  \whnf w_2 : \tau' $
  \hfill by Def.~\ref{def:typedwhnf}}
\prf{$\Gamma \Vdash t_1 = t_2 : \tau'$ \hfill by semantic equ. def.}
\\[0.5em]
\pcase{\fbox{\fbox{$\tau = (y : \ann\tau_1) \arrow \tau_2$ i.e.\ $\Gamma \vdash \tau \whnf (y : \ann\tau_1) \arrow \tau_2 : u$}}}
Symmetry for Terms (Prop. \ref{it:sym}):  To Show:
\fbox{$\Gamma \Vdash t_2 = t_1 : (y : \ann\tau_1) \arrow \tau_2$ }\\
\prf{$\Gamma \vdash t_1 \whnf w_1 : \tau$ \hfill by definition of  $\Gamma \Vdash t_1 = t_2 : \tau$}
\prf{$\Gamma \vdash t_2 \whnf w_2 : \tau$ \hfill by definition of  $\Gamma \Vdash t_1 = t_2 : \tau$}
\prf{Assume $\Gamma' \leq_\rho \Gamma$ and $\Gamma' \Vdash s_2 = s_1 : \{\rho\}\ann\tau_1$}
\prf{\smallerderiv{$\Gamma' \Vdash \{\rho\}\ann\tau_1 : u_1$} \hfill by $\Gamma \Vdash \tau : u$}
\prf{$\Gamma' \sem s_1 = s_2 : \{\rho\}\ann\tau_1$
  \hfill by induction hypothesis (Prop. \ref{it:sym}), symmetry}
\prf{$\Gamma' \sem \{\rho\}w_1~s_1 = \{\rho\}w_2~s_2 : \{\rho,s_1/y\}\tau_2$
  \hfill by definition of  $\Gamma \Vdash t_1 = t_2 : \tau$}
\prf{\smallerderiv{$\Gamma' \sem \{\rho,s_1/y\}\tau_2 : u_2$} \hfill by assumption $\Gamma \sem \tau : u$ }
\prf{$\Gamma' \sem \{\rho\}w_2~s_2 = \{\rho\}w_1~s_1 : \{\rho,s_1/y\}\tau_2$
  \hfill by induction hypothesis (Prop. \ref{it:sym}), symmetry}
\prf{$\Gamma \vdash \tau' \whnf (y : \ann\tau_1') \arrow \tau_2' : u$
  \hfill by definition of $\Gamma \Vdash \tau = \tau' : u$}
\prf{$\Gamma' \Vdash \{\rho,s_1/y\}\tau_2 = \{\rho,s_2/y\}\tau_2' : u_2$
  \hfill by definition of $\Gamma \Vdash \tau = \tau' : u$}
\prf{$\Gamma' \Vdash s_2 = s_2 : \{\rho\}\ann\tau_1$
  \hfill by induction hypothesis (Prop. \ref{it:refl}), reflexivity}
\prf{$\Gamma' \Vdash \{\rho,s_2/y\}\tau_2 = \{\rho,s_2/y\}\tau_2' : u_2$
  \hfill by definition of $\Gamma \Vdash \tau = \tau' : u$}
\prf{$\Gamma' \Vdash \{\rho,s_2/y\}\tau_2' = \{\rho,s_2/y\}\tau_2 : u_2$
  \hfill by induction hypothesis (Prop. \ref{it:symtyp}), symmetry}
\prf{$\Gamma' \Vdash \{\rho,s_1/y\}\tau_2 = \{\rho,s_2/y\}\tau_2 : u_2$
  \hfill by induction hypothesis (Prop. \ref{it:transtyp}), transitivity}
\prf{$\Gamma' \sem \{\rho\}w_2~s_2 = \{\rho\}w_1~s_1 : \{\rho,s_2/y\}\tau_2$
  \hfill by induction hypothesis (Prop. \ref{it:conv}), conversion}
\prf{$\Gamma \Vdash t_2 = t_1 : \tau$ \hfill since $\Gamma',\rho,s_2,s_1$ were arbitrary}
\\
Transitivity for Terms (Prop. \ref{it:trans}):
\fbox{ If\/ $\Gamma \Vdash t_2 = t_3 : (y : \ann\tau_1) \arrow \tau_2$ then
  $\Gamma \Vdash t_1 = t_3 : (y : \ann\tau_1) \arrow \tau_2$.}\\
\prf{$\Gamma \vdash t_1 \whnf w_1 : \tau$ \hfill by definition of  $\Gamma \Vdash t_1 = t_2 : \tau$}
\prf{$\Gamma \vdash t_2 \whnf w_2 : \tau$ \hfill by definition of  $\Gamma \Vdash t_1 = t_2 : \tau$}
\prf{$\Gamma \vdash t_2 \whnf w_2' : \tau$ \hfill by definition of  $\Gamma \Vdash t_2 = t_3 : \tau$}
\prf{$\Gamma \vdash t_3 \whnf w_3 : \tau$ \hfill by definition of  $\Gamma \Vdash t_2 = t_3 : \tau$}
\prf{$w_2 = w_2'$ \hfill by determinacy of weak head evaluation (Lemma \ref{lem:detwhnf})}
\prf{Assume $\Gamma' \leq_\rho \Gamma$ and $\Gamma' \Vdash s_1 = s_3 : \{\rho\}\ann\tau_1$}
\prf{\smallerderiv{$\Gamma' \Vdash \{\rho\}\ann\tau_1 : u_1$} \hfill by $\Gamma \Vdash \tau : u$}
\prf{$\Gamma' \sem s_1 = s_1 : \{\rho\}\ann\tau_1$
  \hfill by induction hypothesis (Prop. \ref{it:refl}), reflexivity}
%
\prf{\smallerderiv{$\Gamma' \sem \{\rho,s_1/y\}\tau_2 : u_2$} \hfill  by assumption $\Gamma \sem \tau : u$ }
\prf{$\Gamma' \sem \{\rho\}w_1~s_1 = \{\rho\}w_2~s_1 : \{\rho,s_1/y\}\tau_2$
  \hfill by assumption $\Gamma \sem t_1 = t_2 : \tau$}
\prf{$\Gamma' \sem \{\rho\}w_2~s_1 = \{\rho\}w_3~s_3 : \{\rho,s_1/y\}\tau_2$
  \hfill by assumption $\Gamma \sem t_2 = t_3 : \tau$}
\prf{$\Gamma' \sem \{\rho\}w_1~s_1 = \{\rho\}w_3~s_3 : \{\rho,s_1/y\}\tau_2$
  \hfill by induction hypothesis (Prop. \ref{it:trans}), transitivity}
\prf{$\Gamma \Vdash t_1 = t_3 : \tau$  \hfill since $\Gamma',\rho,s_1,s_3$ were arbitrary}
\\
Symmetry for Types (Prop. \ref{it:symtyp}):\fbox{$\Gamma \sem \tau' = \tau : u$}\\
\prf{$\Gamma \sem (y : \ann\tau_1) \arrow \tau_2  = \tau' : u$ \hfill by assumption}
\prf{$\Gamma \vdash \tau' \whnf (y : \ann\tau_1') \arrow \tau_2'$
  \hfill by definition of $\Gamma \sem \tau = \tau' : u$}
\prf{Assume $\Gamma' \leq_\rho \Gamma$ and $\Gamma' \Vdash s' = s : \{\rho\}\ann\tau_1'$.}
\prf{\smallerderiv{$\Gamma' \Vdash \{\rho\}\ann\tau_1 : u_1$} \hfill by $\Gamma \Vdash \tau : u$}
\prf{$\Gamma' \Vdash \{\rho\}\ann\tau_1 = \{\rho\}\ann\tau_1' : u_1$ \hfill by $\Gamma \Vdash \tau = \tau' : u$}
\prf{$\Gamma' \Vdash \{\rho\}\ann\tau_1' = \{\rho\}\ann\tau_1 : u_1$ \hfill by induction hypothesis (Prop. \ref{it:symtyp})}
\prf{\highlight{$\Gamma' \vdash \{\rho\}\ann\tau_1' : u_1$} \hfill by $\Gamma \Vdash \tau' : u$}
\prf{$\Gamma' \Vdash s' = s : \{\rho\}\ann\tau_1$ \hfill by induction hypothesis (Prop. \ref{it:conv}), conversion}
\prf{$\Gamma' \Vdash s = s' : \{\rho\}\ann\tau_1$ \hfill by induction hypothesis (Prop. \ref{it:sym}), symmetry for terms}
\prf{$\Gamma' \sem \{\rho, s/y\}\tau_2 = \{\rho, s'/y\}\tau_2' : u_2$  \hfill by $\Gamma \sem \tau = \tau' : u$}
\prf{$\Gamma' \sem s = s : \{\rho\}\ann\tau_1$
  \hfill by induction hypothesis (Prop. \ref{it:refl}), reflexivity}
\prf{\smallerderiv{$\Gamma' \Vdash \{\rho, s/y\}\tau_2 : u_2$} \hfill by $\Gamma \Vdash \tau : u$}
\prf{$\Gamma' \sem \{\rho, s'/y\}\tau_2' =  \{\rho, s/y\}\tau_2 : u$ \hfill by induction hypothesis (Prop. \ref{it:symtyp}), symmetry for types}
%
%
\prf{$\Gamma \sem \tau' = \tau : u$ \hfill since $\Gamma', \rho, s, s'$ were arbitrary}
\\
Transitivity for Types (Prop. \ref{it:transtyp}):
\fbox{ If\/ $\Gamma \Vdash \tau' = \tau'' : u$ and $\Gamma \Vdash \tau'' : u$ then
$\Gamma \Vdash \tau = \tau'' : u$.}\\
\prf{$\Gamma \sem (y : \ann\tau_1') \arrow \tau_1 = \tau' : u$ \hfill by assumption}
\prf{$\Gamma \vdash \tau' \whnf (y:\ann\tau_2') \arrow \tau_2 : u$ \hfill by definition of $\Gamma \sem \tau = \tau':u$
and determinacy of reduction}
\prf{$\Gamma \vdash \tau'' \whnf (y:\ann\tau_3') \arrow \tau_3 : u$ \hfill by definition of $\Gamma \sem \tau' = \tau'' : u$ }
\prf{Assume $\Gamma' \leq_\rho \Gamma$ and $\Gamma' \sem s_1 = s_3 : \{\rho\}\ann\tau_1'$}
\prf{\smallerderiv{$\Gamma' \Vdash \{\rho\}\ann\tau_1' : u_1$} \hfill by $\Gamma \Vdash \tau : u$}
\prf{$\Gamma' \sem s_1 = s_1 : \{\rho\}\ann\tau_1'$
  \hfill  by induction hypothesis (Prop. \ref{it:refl}), reflexivity}
\prf{$\Gamma' \sem \{\rho,s_1/y\}\tau_1 = \{\rho, s_1/y\}\tau_2 : u_2$
   \hfill by definition of $\Gamma \sem \tau = \tau' :u$}
\prf{$\Gamma' \sem \{\rho\}\ann\tau_1' = \{\rho\}\ann\tau_2' : u_1$
   \hfill by definition of $\Gamma \sem \tau = \tau' :u$}
\prf{$\Gamma' \sem s_1 = s_1 : \{\rho\}\ann\tau_2'$ \hfill
      \hfill by induction hypothesis (Prop. \ref{it:conv}), type conversion}
\prf{$\Gamma' \sem \{\rho,s_1/y\}\tau_2 = \{\rho, s_1/y\}\tau_3 : u_2$
   \hfill by definition of $\Gamma \sem \tau' = \tau'' : u$}
\prf{\smallerderiv{$\Gamma' \Vdash \{\rho, s_1/y\}\tau_1 : u_2$} \hfill by $\Gamma \Vdash \tau : u$ }
\prf{$\Gamma' \sem \{\rho,s_1/y\}\tau_1 = \{\rho, s_1/y\}\tau_3 : u_2$
   \hfill by induction hypothesis (Prop. \ref{it:transtyp}), transitivity for types}
\prf{$\Gamma' \sem \{\rho\}\ann\tau_2' = \{\rho\}\ann\tau_3' : u_1$
   \hfill by definition of $\Gamma \sem \tau' = \tau'' :u$}
\prf{$\Gamma' \sem \{\rho\}\ann\tau_1' = \{\rho\}\ann\tau_3' : u_1$
   \hfill by induction hypothesis (Prop. \ref{it:transtyp}), transitivity for types}
\prf{$\Gamma' \sem s_1 = s_3 : \{\rho\}\tau_3'$
    \hfill by  induction hypothesis (\ref{it:conv}), type conversion}
\prf{$\Gamma' \sem \{\rho, s_1/y\}\tau_3 = \{\rho, s_3/y\}\tau_3 : u_2$
    \hfill by $\Gamma \sem \tau'' : u$}
\prf{$\Gamma \sem \{\rho, s_1/y\}\tau_1 = \{\rho, s_3/y\}\tau_3 : u_2$
   \hfill by induction hypothesis (Prop. \ref{it:transtyp}), transitivity for types}
\prf{$\Gamma \sem \tau = \tau'' : u$ \hfill since $\Gamma', \rho, s_1, s_3$ were arbitrary}
\\
Conversion (Prop. \ref{it:conv}). \fbox{$\Gamma \Vdash t_1 = t_2 : \tau'$.}
\\
\prf{$\Gamma \der \tau \equiv \tau' : u$  \hfill by Well-formedness Lemma~\ref{lm:semwf}}
\prf{$\Gamma \vdash t_1 \whnf w_1 : \tau$ \hfill by definition of  $\Gamma \Vdash t_1 = t_2 : \tau$}
\prf{$\Gamma \vdash t_1 \whnf w_1 : \tau'$ \hfill by the conversion rule}
\prf{$\Gamma \vdash t_2 \whnf w_2 : \tau'$ \hfill ditto}
\prf{$\Gamma \vdash \tau' \whnf (y : \ann\tau_1') \arrow \tau_2' : u$ \hfill by definition of  $\Gamma \Vdash \tau = \tau' : u$}
\prf{Assume $\Gamma' \leq_\rho \Gamma$ and $\Gamma' \Vdash s_1 = s_2 : \{\rho\}\ann\tau_1'$}
\prf{$\Gamma' \Vdash \{\rho\}\ann\tau_1 =\{\rho\}\ann\tau_1' : u_1$
   \hfill by definition of $\Gamma \sem \tau = \tau' : u_1$}
\prf{\smallerderiv{$\Gamma' \Vdash \{\rho\}\ann\tau_1 : u_1$} \hfill by $\Gamma \Vdash \tau : u$}
 \prf{$\Gamma' \Vdash \{\rho\}\ann\tau_1' =\{\rho\}\ann\tau_1 : u_1$
    \hfill by induction hypothesis (Prop. \ref{it:symtyp}), symmetry}
\prf{\highlight{$\Gamma' \Vdash s_1 = s_2 : \{\rho\}\ann\tau_1$
  \hfill by induction hypothesis (Prop. \ref{it:conv}) on $\Gamma' \sem \{\rho\}\ann\tau_1' : u_1$, conversion}}
\prf{$\Gamma' \sem \{\rho\}w_1~s_1 = \{\rho\}w_2~s_2 : \{\rho,s_1/y\}\tau_2$
  \hfill by assumption $\Gamma \sem t_1 = t_2 : \tau$}
\prf{$\Gamma' \sem s_1 = s_1 : \{\rho\}\ann\tau_1$
  \hfill by induction hypothesis (\ref{it:refl}), reflexivity}
\prf{\smallerderiv{$\Gamma' \sem \{\rho,s_1/y\}\tau_2 : u_2$} \hfill  by definition of $\Gamma \sem \tau : u$}
\prf{$\Gamma' \sem \{\rho,s_1/y\}\tau_2 = \{\rho,s_1/y\}\tau_2'$
  \hfill by definition of $\Gamma \sem \tau = \tau' : u$}
\prf{$\Gamma' \sem \{\rho\}w_1~s_1 = \{\rho\}w_2~s_2 : \{\rho,s_1/y\}\tau_2'$
  \hfill by induction hypothesis (Prop. \ref{it:conv}), conversion}
\prf{$\Gamma \Vdash t_1 = t_2 : \tau'$  \hfill since $\Gamma',\rho,s_1,s_2$ were arbitrary}
%
\\
\pcase{\fbox{\fbox{$\tau = u'$,~i.e. $\Gamma \Vdash \tau : u$ where
$\Gamma \vdash \tau \whnf u' : u$ and $u' < u$}}}
Symmetry for Terms (Prop. \ref{it:sym}): To Show:
\fbox{$\Gamma \Vdash t_2 = t_1 : u'$}\\
\prf{$\Gamma \Vdash t_2 = t_1 : u'$ \hfill by IH using Symmetry for Types (Prop. \ref{it:symtyp}) (since $u' < u$)}
%
\\[1em]
Transitivity for Terms (Prop. \ref{it:trans}): To Show:
\fbox{If\/ $\Gamma \Vdash t_2 = t_3 : u'$ then  $\Gamma \Vdash t_1 = t_3 : u'$.}
\\[1em]
\prf{$\Gamma \vdash t_1 = t_3:u'$ \hfill by IH using Transitivity for Types (Prop. \ref{it:transtyp}) (since $u' < u$)}
\\[1em]
Symmetry for Types (Prop. \ref{it:symtyp}): To Show:
\fbox{$\Gamma \sem \tau' = u' : u$}
\\[1em]
\prf{$\Gamma \vdash \tau' \whnf u' : u$ \hfill by $\Gamma \Vdash u' = \tau' : u$}
\prf{$\Gamma \vdash \tau' = \tau : u$ \hfill since $\Gamma \vdash \tau \whnf u' : u$ and $u' < u$  (by assumption), and $\Gamma \vdash \tau' \whnf u' : u$ }
\\[1em]
Transitivity for Types (Prop. \ref{it:transtyp}): To Show:
\fbox{ If\/ $\Gamma \Vdash \tau' = \tau'' : u$ and $\Gamma \Vdash \tau'' : u$ then
$\Gamma \Vdash u' = \tau'' : u$.}
\\[1em]
\prf{$\Gamma \vdash \tau' \whnf u' : u$ \hfill $\Gamma \Vdash u' =  \tau' : u$ }
\prf{$\Gamma \vdash \tau'' \whnf u' : u$ \hfill by $\Gamma \Vdash \tau' = \tau'' : u$ }
\prf{$\Gamma \Vdash u' = \tau'' : u$ \hfill using sem. equ. def. and the assumption $u' < u$}
\\[1em]
Conversion (Prop. \ref{it:conv}): To Show: \fbox{$\Gamma \Vdash t_1 = t_2 : \tau'$.}
\\[1em]
\prf{$\Gamma \Vdash t_1 = t_2 : \tau$ and
     $\Gamma \Vdash \tau : u$ where $\Gamma \vdash \tau \whnf u' : u$ and $u' < u$\hfill by assumption}
\prf{$\Gamma \Vdash u' = \tau' : u$ and $u' < u$ \hfill by assumption}
\prf{$\Gamma \vdash \tau' \whnf u' : u$ \hfill by $\Gamma \Vdash u' = \tau' : u$ }
\prf{$\Gamma \Vdash t_1 = t_2 : \tau'$ \hfill since $\Gamma \Vdash \tau' : u$}
\\
\pcase{\fbox{\fbox{$\tau = x~\vec{t}$ and $\Gamma \vdash \tau \whnf x~\vec{t} : u$ and $\neut (x~\vec{t})$}}}
\\
Symmetry for Terms (Prop. \ref{it:sym}): To Show: \fbox{$\Gamma \Vdash t_2 = t_1 : x~\vec{s}$.}
\\
\prf{$\Gamma \Vdash t_1 = t_2 : x~\vec{s}$ \hfill by assumption}
\prf{$\Gamma \vdash t_1 \whnf n_1 : x ~\vec{s}$,~~$\Gamma \vdash t_2 \whnf n_2 : x~\vec{s}$,~~
$\Gamma \vdash n_1 \equiv n_2 : x~\vec{s}$,~~$\neut n_1, n_2$ \hfill by $\Gamma \Vdash t_1 = t_2 : x~\vec{s}$}
\prf{$\Gamma \vdash n_2 \equiv n_1 : x~\vec{s}$ \hfill by symmetry of $\equiv$}
\prf{$\Gamma \Vdash t_2 = t_1 : x~\vec{s}$ \hfill by sem. equ. definition}
\\[0.5em]
Transitivity for Terms (Prop. \ref{it:trans}):
To Show \fbox{If\/ $\Gamma \Vdash t_2 = t_3 : x~\vec{s}$ then  $\Gamma \Vdash  t_1 = t_3 : x~\vec{s}$.}
\\
\prf{$\Gamma \Vdash t_1 = t_2 :  x~\vec{s}$ \hfill by assumption}
\prf{$\Gamma \vdash t_1 \whnf n_1 : x ~\vec{s}$,~~$\Gamma \vdash t_2 \whnf n_2 : x~\vec{s}$,~~$\Gamma \vdash n_1 \equiv n_2 : x~\vec{s}$,~~$\neut n_1, n_2$ \hfill by $\Gamma \Vdash t_1 = t_2 : x~\vec{s}$}
\prf{$\Gamma \vdash t_3 \whnf n_3 : x ~\vec{s}$,
     $\Gamma \vdash n_2 \equiv n_3 : x~\vec{s}$,~~$\neut n_3$ \hfill by $\Gamma \Vdash t_2 = t_3 : x~\vec{s}$}
\prf{$\Gamma \vdash n_1 \equiv n_3 : x~\vec{s}$ \hfill by transitivity of $\equiv$}
\prf{$\Gamma \Vdash t_1 = t_3 : x~\vec{s}$ \hfill by sem. equ. definition}
\\[0.5em]
Symmetry for Types (Prop. \ref{it:symtyp}): To Show: \fbox{$\Gamma \sem \tau' = x~\vec{s} : u$}
\\
 \prf{$\Gamma \Vdash x~\vec{s} = \tau' : u$ where $\Gamma \vdash \tau \whnf x~\vec{s} : u$ \hfill by assumption}
 \prf{$\Gamma \vdash \tau' \whnf x~\vec{t} : u$ and $\Gamma \vdash x~\vec{s} \equiv x~\vec{t} : u$}
 \prf{$\Gamma \vdash x~\vec{t} \equiv x~\vec{s} : u$ \hfill by symmetry of $\equiv$}
 \prf{$\Gamma \Vdash \tau' = \tau : u$ \hfill by sem. equ. definition}
\\[0.5em]
Transitivity for Types (Prop. \ref{it:transtyp}): To Show: \fbox{If\/ $\Gamma \Vdash \tau' = \tau'' : u$ and $\Gamma \Vdash \tau'' : u$ then
 $\Gamma \Vdash x~\vec{s} = \tau'' : u$.}
\\
\prf{$\Gamma \Vdash x~\vec{s} = \tau' : u$ \hfill by assumption}
\prf{$\Gamma \vdash \tau' \whnf x~\vec{s'} : u$ and $\Gamma \vdash x~\vec{s} \equiv x~\vec{s'} : u$ \hfill by $\Gamma \Vdash x~\vec{s} = \tau' : u$ }
\prf{$\Gamma \Vdash \tau' = \tau'' : u$ \hfill by assumption}
\prf{$\Gamma \vdash \tau'' \whnf x~\vec{s''}:u$ and $\Gamma \vdash x~\vec{s'} \equiv x~\vec{s''}:u$ \hfill by $\Gamma \Vdash \tau' = \tau'' : u$}
\prf{$\Gamma \vdash x~\vec{s} = x~\vec{s''} : u$ \hfill by transitivity of $\equiv$}
\prf{$\Gamma \vdash \tau = \tau'' : u$ \hfill by sem. equ. def.}

Conversion (Prop. \ref{it:conv}): \fbox{$\Gamma \Vdash t_1 = t_2 : \tau'$.}
\\
\prf{$\Gamma \Vdash t_1 = t_2 : x~\vec{s} $ \hfill by assumption}
\prf{$\Gamma \vdash t_1 \whnf n_1 : x~\vec{s}$ and $\Gamma \vdash t_2 \whnf n_2 : x~\vec{s}$ and $\Gamma \vdash n_1 \equiv n_2 : x~\vec{s}$ \hfill by $\Gamma \Vdash t_1 = t_2 : x~\vec{s} $}
\prf{$\Gamma \Vdash x~\vec{s} = \tau' : u$ \hfill by assumption}
\prf{$\Gamma \vdash \tau' \whnf x~\vec{s'} : u$ and $\Gamma \vdash x~\vec{s} \equiv x~\vec{s'} : u$ \hfill by $\Gamma \Vdash x~\vec{s} = \tau' : u$ }
\prf{$\Gamma \vdash n_1 \equiv n_2 : x~\vec{s'} $ \hfill using type conversion}
\prf{$\Gamma \vdash t_1 \whnf n_1 : x~\vec{s'}$ and $\Gamma \vdash t_2 \whnf n_2 : x~\vec{s'}$ \hfill using type conversion}
\prf{$\Gamma \Vdash t_1 = t_2 : \tau'$ \hfill by sem. equ. def.}
}
\end{proof}

Finally we establish various elementary properties about our semantic definition that play a key role in the fundamental lemma which we prove later.

 \begin{lemma}[Neutral Soundness]\label{lem:NeutSound}$\quad$$\quad$\\
If\/ $\Gamma \Vdash \ann\tau : u$ and $\Gamma \vdash t : \ann\tau$ and $\Gamma \vdash t' : \ann\tau$ and $\Gamma \vdash t \equiv t' : \ann\tau$ and $\neut t, t'$ then $\Gamma \Vdash t = t' : \ann\tau$.

 \end{lemma}
 \begin{proof}
By induction on $\Gamma \Vdash \tau : u$.
%
\LONGVERSIONCHECKED{
\\
\noindent
\pcase{$\tau = u$}
\prf{$\neut t$ and $\neut t'$ \hfill by assumption}
\prf{$t = x~\vec{s}$ and $t' = x~\vec{s'}$ \hfill since $\neut t$ and $\neut t$, $\Gamma \vdash t : u$, $\Gamma \vdash t':u$ and $\Gamma \vdash t \equiv t' : u$}
\prf{$\Gamma \vdash x~\vec{s} \equiv x~\vec{s'} : u$ \hfill by assumption $\Gamma \vdash t \equiv t' : u$}
\prf{$\Gamma \vdash t' \whnf  x~\vec{s'}$ and $\Gamma \vdash t \whnf  x~\vec{s}$ \hfill since $\neut t$ and $\neut t'$}
\prf{$\Gamma \vdash x~\vec{s} = t' : u$ \hfill by sem. def.}
%
%
\\
\pcase{$\tau =\cbox{T}$ where $\Psi \vdash \tm$}
\prf{$\Gamma \vdash t : \cbox{T}$ and $\Gamma \vdash t' : \cbox{T}$\hfill by assumption}
\prf{$\neut t$ and $\neut t'$\hfill by assumption}
\prf{$\norm t$ and $\norm t'$\hfill by def. of $\norm / \neut$}
\prf{$t \whnf t$ and $t' \whnf t'$ \hfill by def. of $\whnf$}
\prf{$\Gamma \vdash t \equiv t' : \cbox{T}$ \hfill by assumption}
\prf{$\Gamma \vdash t \whnf t : \cbox{T}$  and $\Gamma \vdash t' \whnf t' : \cbox{T}$ \hfill by Def.~\ref{def:typedwhnf}}
\prf{$\Gamma ; \Psi \vdash \unbox t \id \lfwhnf \unbox t \id : A$ \hfill since $\neut t$}
\prf{$\Gamma ; \Psi \vdash \unbox {t'} \id \lfwhnf \unbox {t'} \id : A$ \hfill since $\neut t'$}
\prf{$\Gamma ; \Psi \vdash \typeof (\Gamma \vdash t) = \cbox{\Phi \vdash \tm}$ and $\Gamma \vdash \Psi \equiv \Phi : \ctx$ \hfill by Lemma \ref{lm:typeof} }
\prf{$\Gamma ; \Psi \vdash \typeof (\Gamma \vdash t') = \cbox{\Phi' \vdash \tm}$ and $\Gamma \vdash \Psi \equiv \Phi' : \ctx$ \hfill by Lemma \ref{lm:typeof} }
\prf{$\Gamma \vdash \Phi \equiv \Phi'  : \ctx$ \hfill by symmetry and transitivity of $\equiv$}
\prf{$\Gamma ; \Psi \Vdash \id = \id : \Psi$ \hfill by Lemma \ref{lm:semlfwk}}
\prf{$\Gamma ; \Psi \Vdash \unbox t \id =  \unbox {t'}\id : A$ \hfill by sem. def.}
\prf{$\Gamma \Vdash t = t' : \cbox{T}$ \hfill by semantic def.}
\\
\pcase{$\tau = (y:\ann\tau_1) \arrow \tau_2$}
\prf{$\Gamma \vdash t : (y:\ann\tau_1) \arrow \tau_2$ and $\Gamma \vdash t' : (y:\ann\tau_1) \arrow \tau_2$ \hfill by assumption}
\prf{$\neut t$ and $\neut t'$ \hfill by assumption}
\prf{$\norm t$ and $\norm t'$ \hfill by def. of $\norm / \neut$}
\prf{$t \whnf t$ and $t' \whnf t'$ \hfill by def. of $\whnf$}
\prf{$\Gamma \vdash t \equiv t' : (y:\ann\tau_1)\arrow \tau_2$ \hfill by assumption}
\prf{$\Gamma \vdash t \whnf t : (y:\ann\tau_1) \arrow \tau_2$ \hfill by Def.~\ref{def:typedwhnf}}
\prf{$\Gamma \vdash t' \whnf t' : (y:\ann\tau_1) \arrow \tau_2$ \hfill by Def.~\ref{def:typedwhnf}}
\prf{Assume $\forall \Gamma' \leq_\rho \Gamma.~\Gamma' \Vdash s = s' :  \{\rho\}\ann\tau_1$}
\prf{\mbox{$\quad$}$\Gamma' \vdash \{\rho\}t \equiv \{\rho\}t': \{\rho\}((y:\ann\tau_1) \arrow \tau_2)$ \hfill by Weakening Lemma \ref{lem:weakcomp}}
\prf{\mbox{$\quad$}$\Gamma' \vdash \{\rho\}t \equiv \{\rho\}t': (y:\{\rho\}\ann\tau_1) \arrow \{\rho, y/y\}\tau_2$ \hfill by subst. def.}
\prf{\mbox{$\quad$}$\Gamma' \vdash s \equiv s' : \{\rho\}\ann\tau_1$ \hfill by  Well-formedness Lemma~\ref{lm:semwf} }
\prf{\mbox{$\quad$}$\Gamma' \vdash \{\rho\}t~s \equiv \{\rho\}t'~s' : \{\rho, s/y\}\tau_2$\hfill by rule}
\prf{\mbox{$\quad$}$\neut \{\rho\}t$ and $\neut \{\rho\}t'$ \hfill by Lemma \ref{lem:weaknorm}}
\prf{\mbox{$\quad$}$\neut \{\rho\}t~s$ and $\neut \{\rho\}t'~s'$ \hfill by def. of $\norm / \neut$}
\prf{\mbox{$\quad$}\smallerderiv{$\Gamma' \Vdash \{\rho,s/y\}\tau_2 : u_2$} \hfill by $\Gamma \Vdash \tau : u$}
\prf{\mbox{$\quad$}$\Gamma' \Vdash \{\rho\}t~s = \{\rho\}t~s' : \{\rho, s/y\}\tau_2$\hfill by IH}
\prf{$\Gamma \Vdash t = t' : (y:\ann\tau_1) \arrow \tau_2$ \hfill by semantic def.}

\pcase{$\tau = x~\vec{s}$}
\prf{$\Gamma \vdash \tau \whnf x~\vec{s} : u$ and $\neut(x~\vec{s})$ \hfill by $\Gamma \Vdash \tau : u$ }
\prf{$\Gamma \vdash t \equiv t' : x~\vec{s}$ and $\neut t, t'$ \hfill by assumption}
\prf{$\Gamma \vdash t \whnf t : x~\vec{s}$ \hfill since $\neut t$}
\prf{$\Gamma \vdash t' \whnf t' : x~\vec{s}$ \hfill since $\neut t'$}
\prf{$\Gamma \Vdash t = t' : x~\vec{s}$ \hfill by sem. equ. def.}

}
 \end{proof}


\begin{lemma}[Backwards Closure for Computations]\label{lem:bclosed}\quad
\\
If $\Gamma \sem t_1 = t_2: \ann\tau$ and $\Gamma \vdash t_1 \whnf w : \ann\tau$ and $\Gamma \vdash t_1' \whnf w : \ann\tau$
then $\Gamma \sem t_1' = t_2: \ann\tau$.
\end{lemma}
\begin{proof}
By case analysis of $\Gamma \sem t_1 = t_2: \ann\tau$ considering different cases of $\Gamma \sem \ann\tau : u$.
\end{proof}

\begin{lemma}[Typed Whnf Is Backwards Closed]\label{lem:typeclosed}
$\;$\\ If \mbox{$\Gamma \vdash t \whnf w : (y:\ann\tau_1) \arrow \tau_2$}
    and  $\Gamma \vdash s : \ann\tau_1$
    \\and  $\Gamma \vdash w~s \whnf v : \{s/y\}\tau_2$
    then $\Gamma \vdash t~s \whnf v : \{s / y \}\tau_2$.
\end{lemma}
\begin{proof}
By unfolding the definitions and considering different cases for $w$.
\LONGVERSIONCHECKED
{
 \quad\\
\prf{$\Gamma \vdash t : (y:\ann\tau_1) \arrow \tau_2$ \hfill by def. of $\Gamma \vdash t \whnf w : (y:\ann\tau_1) \arrow \tau_2$}
\prf{$\Gamma \vdash s : \ann\tau_1$ \hfill by assumption}
\prf{$\Gamma \vdash t~s : \{s/y\}\tau_2$ \hfill by typing rule}
\prf{$\Gamma \vdash s \equiv s : \ann\tau_1$ \hfill by reflexivity of $\equiv$}
\prf{$\Gamma \vdash t \equiv w : (y:\ann\tau_1) \arrow \tau_2$ \hfill by def. of $\Gamma \vdash t \whnf w : (y:\ann\tau_1) \arrow \tau_2$}
\prf{$\Gamma \vdash t~s \equiv w~s : \{s/y\}\tau_2$ \hfill by congruence rules of $\equiv$}
\prf{$\Gamma \vdash w~s \equiv v : \{s/y\}\tau_2$ \hfill by def. of $\Gamma \vdash w~s \whnf v : \{s/y\}\tau_2$ }
\prf{$\Gamma \vdash t~s \equiv v : \{s/y\}\tau_2$ \hfill by symmetry and transitivity of $\equiv$}
\prf{$t \whnf w$ \hfill by  def. of $\Gamma \vdash t \whnf w : (y:\ann\tau_1) \arrow \tau_2$}
\prf{$w~s \whnf v$ \hfill by def. of $\Gamma \vdash w~s \whnf v : \{s/y\}\tau_2$ }
\prf{$\Gamma \vdash v : \{s/y\}\tau_2$ \hfill by def. of $\Gamma \vdash w~ s \whnf v : \{s/y\}\tau_2$ }
\prf{$\norm w$ \hfill by definition of $\whnf$}
\\
\prf{\emph{Sub-case}: $t \whnf \tmfn x {t'}$ and $w = \tmfn x {t'}$}
\prf{$(\tmfn x t')~s \whnf v$ where $\{s/x\}t' \whnf v$ \hfill by $\Gamma \vdash w~s\whnf v: \{s/y\}\tau_2$}
\prf{$t~s \whnf v$ \hfill since $t \whnf \tmfn x t'$}
\prf{$\Gamma \vdash t~s \whnf v : \{s / y \}\tau_2$ \hfill by def. }
\\
\prf{\emph{Sub-case}: $t \whnf w$ where $\neut w$ 
                       }
\prf{$w~s\whnf w~s$ \hfill  since $\neut (w~s)$}
\prf{$t~s \whnf w~s$ \hfill by rule}
\prf{$\Gamma \vdash t~s \whnf v : \{s / y \}\tau_2$ \hfill by def. }
}
\end{proof}

\begin{lemma}[Semantic Application]\label{lem:SemTypeApp}$\;$\\
If $\Gamma \Vdash t = t': (y:\ann\tau_1) \arrow \tau_2$ and $\Gamma \Vdash s = s': \ann\tau_1$ then
$\Gamma \Vdash t~s = t'~s': \{s/y\}\tau_2$.
 \end{lemma}
 \begin{proof}
Using well-formedness of semantic equality, Backwards closed properties (Lemma \ref{lem:typeclosed} and \ref{lem:bclosed}), and Symmetry of semantic equality (Lemma Prop. \ref{it:sym}).
\LONGVERSIONCHECKED{
$\;$\\ [1em]
\prf{$\Gamma \vdash t \whnf w   : (y:\ann\tau_1) \arrow \tau_2$ ~\mbox{and}~\\
$\Gamma \vdash t' \whnf w' : (y:\ann\tau_1) \arrow \tau_2$ ~\mbox{and}~
\\
$\forall \Gamma' \leq_\rho \Gamma.~\Gamma' \Vdash s = s' : \{\rho\}\ann\tau_1
\Longrightarrow \Gamma' \Vdash (\{\rho\}w)~s = \{\rho\}w'~s' : \{\rho, s/y\}\tau_2$
\hfill by sem. def.}
\prf{$\Gamma \Vdash w~s = w'~s': \{s/y\}\tau_2$ \hfill choosing $\Gamma$ for $\Gamma'$, $\rho$ to be the identity substitution }
\prf{$\Gamma \vdash w~s \whnf v : \{s/y\}\tau_2$ \hfill by def. of $\Gamma \Vdash w~s = w'~s': \{s/y\}\tau_2$}
\prf{$\Gamma \vdash s : \tau_1$ and $\Gamma \vdash s' \vdash \tau_1$ \hfill by Well-formedness Lemma \ref{lm:semwf}}
\prf{$\Gamma \vdash t~s \whnf v : \{s/y\}\tau_2$ \hfill by Whnf Backwards closed (Lemma \ref{lem:typeclosed})}
\prf{$\Gamma \vdash w'~s' \whnf v: \{s/y\}\tau_2$ \hfill by def. of $\Gamma \Vdash w~s = w'~s': \{s/y\}\tau_2$}
\prf{$\Gamma \vdash t'~s' \whnf v : \{s/y\}\tau_2$ \hfill by Whnf Backwards closed (Lemma \ref{lem:typeclosed})}
\prf{$\Gamma \Vdash t~s = t'~s': \{s/y\}\tau_2$ \hfill by Semantic Backwards
  Closure for Computations (Lemma \ref{lem:bclosed}) \\
\mbox{\quad} \hfill and Symmetry (Lemma Prop. \ref{it:sym})}
}
 \end{proof}


\section{Validity in the Model}\label{sec:validity}
For normalization, we need to establish that well-typed terms are
logically related. In other words, we show that syntactically
well-typed terms are also semantically well-typed. However, as we
traverse syntactically well-typed terms, they do not remain
closed. Hence, we need to prove a generalization where we
show that every syntactically well-typed term in a context $\Gamma$ is
semantically well-typed in an extension of $\Gamma$.
As is customary, we extend our logical relation
to substitutions defining semantic substitutions which allow us to
move between $\Gamma$ and $\Gamma'$.

\[
  \begin{array}{l}
\fbox{$\Gamma' \Vdash \theta  = \theta': \Gamma$}\\
\infer{\Gamma' \Vdash \cdot = \cdot : \cdot}{\vdash \Gamma' }
\qquad
\infer{\Gamma' \Vdash \theta, t/x = \theta', t'/x : \Gamma, x{:}\ann\tau}
{
    \begin{array}{@{}ll@{}}
\Gamma' \Vdash \theta = \theta' : \Gamma &  \Gamma' \Vdash \{\theta\}\ann\tau = \{\theta'\}\ann\tau : u \\
\Gamma' \Vdash \{\theta\}\ann\tau : u & \Gamma' \Vdash t = t' : \{\theta\}\ann\tau
    \end{array}
}
  \end{array}
\]

Semantic substitutions are well-formed (i.e. they imply that substitutions are well-typed), stable under weakening and preserve equivalences. They are also reflexive, symmetric, and transitive. Further, given a valid context where each of the declarations is valid, we can always generate $\Gamma \Vdash \id(\Gamma) = \id(\Gamma) : \Gamma$, where $\id$ is the identity substitution.

\LONGVERSION{
\begin{lemma}[Context Satisfiability]\label{lem:ctxsat}
If $\models \Gamma$ then $\vdash \Gamma$ and $\Gamma \Vdash \id(\Gamma) = \id(\Gamma) : \Gamma$ where
\[
\begin{array}{lcl}
  \id(\cdot) & = & \cdot\\
  \id(\Gamma, x{:}\tau) & = & \id(\Gamma), x/x
\end{array}
\]
\end{lemma}
\begin{proof}
By induction on $\Gamma$ using  Neutral Soundness (Lemma \ref{lem:NeutSound}) and Semantic Weakening (Lemma \ref{lem:weakcsub}).
\LONGVERSIONCHECKED{By induction on $\Gamma$.\\[0.5em]
\pcase{$\Gamma = \cdot$}
\prf{$\vdash \cdot $ \hfill by rules}
\prf{$\id(\cdot) = \cdot$ \hfill by def. of $\id$}
\prf{$\Gamma' \Vdash \cdot = \cdot : \cdot$ \hfill by semantic def.}

\pcase{$\Gamma = \Gamma_0, x{:}\ann\tau$}
\prf{$\models \Gamma_0, x{:}\ann\tau$ \hfill by assumption}
\prf{$\models \Gamma_0$ and $\Gamma_0 \models \ann\tau : u$ \hfill by def. validity}
\prf{$\vdash \Gamma_0$ and $\Gamma_0 \Vdash \id(\Gamma_0) = \id(\Gamma_0) : \Gamma_0$ \hfill by IH}
\prf{$\forall \Gamma',~\theta,~\theta'.~\Gamma' \Vdash \theta = \theta' : \Gamma
\Longrightarrow \Gamma' \Vdash \{\theta\}\ann\tau  =  \{\theta\}\ann\tau : \{\theta\}u$  \hfill by def. validity}
\prf{$\Gamma_0 \Vdash \{\id(\Gamma_0)\}\ann\tau = \{\id(\Gamma_0)\}\ann\tau : \{\id(\Gamma_0)\}u$ \hfill by previous lines}
\prf{$\Gamma_0 \Vdash \ann\tau = \ann\tau : u$ \hfill by subst. def.}
\prf{$\Gamma_0 \vdash \ann\tau : u$ \hfill by semantic typing}
\prf{$\vdash \Gamma_0, x{:}\ann\tau$ \hfill by context def.}
\prf{$\neut x$ \hfill by def.}
\prf{$\Gamma_0, x{:}\ann\tau \vdash x : \ann\tau$ \hfill by typing rule}
\prf{$\Gamma_0, x{:}\ann\tau \vdash x \equiv x : \ann\tau$ \hfill by typing rule}
\prf{$\Gamma_0, x{:}\ann\tau \Vdash x = x : \ann\tau$ \hfill by Neutral Soundness Lemma \ref{lem:NeutSound}}
\prf{$\Gamma_0, x{:}\ann\tau \Vdash \id({\Gamma_0}) = \id({\Gamma_0}) : \Gamma_0$ \hfill by Sem. Weakening Lemma \ref{lem:weakcsub} }
\prf{$\Gamma_0, x{:}\ann\tau \Vdash \id({\Gamma_0}), x/x = \id({\Gamma_0}), x/x : \Gamma_0, x{:}\ann\tau$ \hfill by semantic def.}
\prf{$\Gamma_0, x{:}\ann\tau \Vdash \id({\Gamma_0, x{:}\tau}) = \id({\Gamma_0, x{:}\tau}) : \Gamma_0, x{:}\ann\tau$  \hfill by def. of $\id$}
}
\end{proof}
}

\LONGVERSION{
\begin{lemma}[Well-formedness of Semantic Substitutions]\label{lem:wfsemsub}
If $\Gamma' \Vdash \theta = \theta' : \Gamma$ then $\Gamma' \vdash \theta : \Gamma$ and $\Gamma' \vdash \theta' : \Gamma$ and $\Gamma' \vdash \theta \equiv \theta' : \Gamma$.
\end{lemma}
\begin{proof}
By induction on   $\Gamma' \Vdash \theta = \theta' : \Gamma$.
\end{proof}

\begin{lemma}[Semantic Weakening of Computation-level Substitutions]\label{lem:weakcsub}
If $\Gamma' \Vdash \theta = \theta' : \Gamma$ and $\Gamma'' \leq_\rho \Gamma'$
then $\Gamma'' \Vdash \{\rho\}\theta = \{\rho\}\theta' : \Gamma$.
\end{lemma}
\begin{proof}
By induction on $\Gamma' \Vdash \theta = \theta' : \Gamma$ and using semantic weakening lemma \ref{lem:semweak}.
\end{proof}

 \begin{lemma}[Semantic Substitution Preserves Equivalence]\label{lem:semsubst}
  Let $\Gamma' \sem \theta = \theta' : \Gamma$;
    \begin{enumerate}
    \item If\/$\Gamma ; \Psi \vdash M \equiv M : A$ then
      $\Gamma;  \{\theta\}\Psi \vdash \{\theta\}M \equiv \{\theta'\}M : \{\theta\}A$.
    \item If\/$\Gamma ; \Psi \vdash \sigma \equiv \sigma : \Phi$ then
      $\Gamma;  \{\theta\}\Psi \vdash \{\theta\}\sigma \equiv \{\theta'\}\sigma : \{\theta\}\Phi$.
    \item If\/ $\Gamma \vdash t \equiv t : \ann\tau$
        then $\Gamma' \vdash \{\theta\}t \equiv \{\theta'\}t : \{\theta\}\ann\tau$.
    \end{enumerate}
  \end{lemma}
  \begin{proof}
 By induction on $M$, $\sigma$ $\tau$, and $t$. The proof is mostly straightforward; in the case where $t = x$ we know by assumption that $t_x/x \in \theta$ and $t'_x/x \in \theta'$ where $\Gamma' \sem t_x = t'_x : \{\theta\}\tau_x$. From Well-formedness of semantic typing (Lemma \ref{lm:semwf}), we know that $\Gamma' \vdash t_x \equiv t'_x : \{\theta\}\tau_x$.
  \end{proof}

\begin{lemma}[Symmetry and Transitivity of Semantic Substitutions]\quad
  \label{lm:symsemsub}
  Assume $\models \Gamma$.
  \begin{enumerate}
  \item If $\Gamma' \Vdash \theta_1 = \theta_2 : \Gamma$ then
    $\Gamma'\Vdash \theta_2 = \theta_1: \Gamma$.
  \item If $\Gamma' \Vdash \theta_1 = \theta_2 : \Gamma$
    and $\Gamma' \Vdash \theta_2 = \theta_3 : \Gamma$
    then $\Gamma' \Vdash \theta_1 = \theta_3 : \Gamma$.
  \end{enumerate}
\end{lemma}
\begin{proof}
We prove symmetry by induction on the derivation and transitivity by induction on both derivations using Symmetry, Transitivity, and Conversion for semantic equality (Lemma \ref{lem:semsym});
; reflexivity follows from symmetry and transitivity.
\LONGVERSIONCHECKED{\\[1em]Symmetry:  By induction on derivation.
\\[1em]
\pcase{\ianc{}{\Gamma' \Vdash \cdot = \cdot : \cdot}{}}
\prf{$\Gamma' \Vdash \cdot = \cdot : \cdot$ \hfill by def.}
\\
\pcase{$\D = \ibnc{\Gamma' \Vdash \theta = \theta' : \Gamma}
            {\Gamma' \Vdash t = t' : \{\theta\}\ann\tau}
            {\Gamma' \Vdash \theta, t/x = \theta', t'/x : \Gamma, x{:}\ann\tau}{}$}
\prf{$\Gamma' \Vdash \theta' = \theta : \Gamma$ \hfill by IH}
\prf{$\Gamma' \Vdash t' = t : \{\theta\}\ann\tau$\hfill by Lemma \ref{lem:semsym} (Symmetry)}
\prf{$\Gamma' \Vdash \theta', t'/x = \theta, t/x : \Gamma, x{:}\ann\tau$ \hfill by rule}
\\
Transitivity: By induction on both derivations.\\[0.5em]
\pcase{$\theta_1= \cdot$, $\theta_2 = \cdot$ and $\theta_3 = \cdot$}
\prf{$\Gamma' \Vdash \cdot = \cdot : \cdot$ \hfill by def.}
\\
\pcase{$\theta_1= \theta_1', t_1/x$, $\theta_2 = \theta_2', t_2/x$ and $\theta_3 =\theta_3', t_3/x$
   and $\Gamma = \Gamma_0, x{:}\ann\tau$}
\prf{$\Gamma' \Vdash \theta_1' = \theta_2' : \Gamma_0$ \hfill by inversion}
\prf{$\Gamma' \Vdash \theta_2' = \theta_3' : \Gamma_0$ \hfill by inversion}
\prf{$\Gamma' \Vdash \theta_1' = \theta_3' : \Gamma_0$ \hfill by IH}
\prf{$\Gamma' \Vdash t_1 = t_2 : \{\theta'_1\}\ann\tau$ \hfill by inversion}
\prf{$\Gamma' \Vdash t_2 = t_3: \{\theta'_2\}\ann\tau$ \hfill by inversion}
\prf{$\Gamma' \Vdash \{\theta'_1\}\ann\tau = \{\theta'_2\}\ann\tau : u$\hfill by inversion}
\prf{$\Gamma' \Vdash \{\theta'_2\}\ann\tau = \{\theta'_3\}\ann\tau : u$\hfill by inversion}
\prf{$\Gamma' \Vdash \{\theta'_2\}\ann\tau = \{\theta'_1\}\ann\tau : u$\hfill by Lemma \ref{lem:semsym} (Symmetry)}
\prf{$\Gamma' \Vdash t_2 = t_3 : \{\theta'_1\}\ann\tau$ \hfill by Lemma \ref{lem:semsym} (Conversion)}
\prf{$\Gamma' \Vdash t_1 = t_3: \{\theta'_1\}\ann\tau$ \hfill by Lemma \ref{lem:semsym} (Transitivity)}
\prf{$\Gamma' \Vdash \{\theta'_1\}\ann\tau = \{\theta'_3\}\ann\tau$ \hfill by Lemma \ref{lem:semsym} (Transitivity)}
\prf{$\Gamma' \Vdash \theta_1 = \theta_3 : \Gamma$ \hfill by rule}
}
\end{proof}

}
Last, we define validity of LF objects, types, and terms (Fig.~\ref{fig:valid}). Our notion
of validity generalizes our definition of semantic typing and
equality. Intuitively, we say that a term $t$ is valid, if for any
semantic substitution $\theta$, $\{\theta\}t$  is semantically well-typed.
 This allows us to define compactly the fundamental lemma which now
 states that well typed terms correspond to valid terms in our model.

Note that we do not work directly with semantically well-typed terms. Instead
we say that a term is semantically well-typed, if it is semantically
equal to itself. Our definition of validity is built on the same
idea. Concretely, we say that two terms $t$ and $t'$ are equal in our
model, i.e.~$\Gamma \models t = t' : \ann\tau$,
if for all semantically equal substitutions $\theta$ and $\theta'$,
we have that $\{\theta\}t$ and
$\{\theta'\}t$ are semantically equal. Our definition of validity is
symmetric and transitive.

\begin{figure}[htb]
  \centering\small
\[
  \begin{array}{l}
\LONGVERSION{
\mbox{Validity of context}: \fbox{$\models \Gamma$}\quad
\raisebox{-1.5ex}{
\infer{\models \cdot}{}
\quad
\infer{\models \Gamma, x:\ann\tau}{\models \Gamma & \Gamma \models \ann\tau : u}}
\\[1.25em]}
\mbox{Validity of LF objects}: \fbox{$\Gamma ; \Psi \models M = N : A$} ~\mbox{where}~\models \Gamma\\[1em]
\infer{\Gamma ; \Psi \models M = N : A}
{
                    \begin{array}{@{}r@{}l@{}}
 \forall \Gamma',\theta,\theta'.& \Gamma' \Vdash \theta = \theta' : \Gamma \\
& \Longrightarrow \Gamma' ; \{\theta\}\Psi \Vdash \{\theta\}M  = \{\theta'\}N : \{\theta\}A
                    \end{array}}
\\[1em]
\mbox{Validity of LF substitutions}: \fbox{$\Gamma ; \Psi \models \sigma = \sigma' : \Phi$} ~\mbox{where}~\models \Gamma \\[1em]
\infer{\Gamma ; \Psi \models \sigma = \sigma' : \Phi}
{
                    \begin{array}{@{}r@{}l@{}}
 \forall \Gamma',~\theta,~\theta'.&\Gamma' \Vdash \theta = \theta' :\Gamma \\
 & \Longrightarrow \Gamma' ; \{\theta\}\Psi \Vdash \{\theta\}\sigma_1  = \{\theta'\}\sigma' : \{\theta\}\Phi
                    \end{array}
}
\\[1em]
\mbox{Validity of types}: \fbox{$\Gamma \models \ann\tau = \ann\tau':u$} \quad\mbox{and}\quad \fbox{$\Gamma \models \ann\tau : u$}
\\[1em]
\infer{\Gamma \models \ann\tau = \ann\tau' : u}
      {
    \begin{array}{@{}r@{}l@{}}
       \forall \Gamma',~\theta,~\theta'. & \Gamma' \Vdash \theta = \theta' : \Gamma \\
       & \Longrightarrow \Gamma' \Vdash \{\theta\}\ann\tau = \{\theta'\}\ann\tau' : u
    \end{array}
}
\qquad
\infer{\Gamma \models \ann\tau : u}{\Gamma \models \ann\tau = \ann\tau : u}
\\[1em]
\mbox{Validity of terms}: \fbox{$\Gamma \models t = t' : \ann\tau$} \quad\mbox{and}\quad \fbox{$\Gamma \models t : \ann\tau$}
\\[1em]
\infer{\Gamma \models t = t' : \ann\tau}
      {
     \begin{array}{@{}l@{}}
\models \Gamma \quad \forall \Gamma',\theta,\theta'. \Gamma' \Vdash \theta = \theta' : \Gamma
\\
\Gamma \models \ann\tau : u\quad\quad
\Longrightarrow \Gamma' \Vdash \{\theta\}t  =\{\theta'\}t' : \{\theta\}\ann\tau
   \end{array}
}
\quad
\infer{\Gamma \models t : \ann\tau}
{\Gamma \models t = t : \ann\tau}
  \end{array}
\]

  \caption{Validity Definition}\label{fig:valid}
\end{figure}

\LONGVERSION{
\begin{lemma}[Symmetry and Transitivity of Validity]\quad
  \begin{enumerate}
  \item If $\Gamma \models t = t' : \ann\tau$ then $\Gamma \models t' = t : \ann\tau$.
  \item If $\Gamma \models t_1 = t_2 : \ann\tau$ and $\Gamma \models t_2 =  t_3 : \ann\tau$
    then $\Gamma \models t_1 = t_3 : \ann\tau$.
  \end{enumerate}
\end{lemma}
\begin{proof}
Using Lemma \ref{lem:ctxsat} (Context Satisfiability), Lemma \ref{lm:symsemsub} (Symmetry and Transitivity for Substitutions),  Lemma \ref{lem:semsym} (Symmetry and Transitivity for Terms),  and Lemma \ref{lem:semsym} (Conversion).
\LONGVERSIONCHECKED{
\\[1em]
    \pcase{$\ibnc{\models \Gamma}
                {\Gamma \models \ann\tau : u \qquad
                 \forall \Gamma',~\theta,~\theta'.~\Gamma' \Vdash \theta = \theta' : \Gamma
         \Longrightarrow \Gamma' \Vdash \{\theta\}t  = \{\theta'\}t' : \{\theta\}\ann\tau}{\Gamma \models t = t' : \ann\tau}{}$}
\\
\prf{Assume $\Gamma' \Vdash \theta' = \theta : \Gamma$}
     \prf{$\Gamma' \Vdash \theta = \theta' : \Gamma$ \hfill by Lemma \ref{lm:symsemsub} (Symmetry)}
    \prf{$\Gamma' \Vdash \{\theta\}t  = \{\theta'\}t' : \{\theta\}\ann\tau$ \hfill by assumption $\Gamma \Vdash t = t' : \ann\tau$}
    \prf{$\Gamma' \Vdash \{\theta\}t'  = \{\theta'\}t : \{\theta\}\ann\tau$ \hfill by Lemma \ref{lem:semsym} (Symmetry)}
    \prf{$\models \Gamma$ \hfill by assumption}
    \prf{$\Gamma \Vdash \id(\Gamma) = \id(\Gamma) : \Gamma$ \hfill by Lemma \ref{lem:ctxsat}}
    \prf{$\Gamma \models \ann\tau : u$ \hfill by assumption}
    \prf{$\Gamma \models \ann\tau = \ann\tau : u$ \hfill by def.}
    \prf{$\Gamma' \Vdash \{\theta'\}\ann\tau = \{\theta\}\ann\tau : u$\hfill by $\Gamma \models \ann\tau = \ann\tau : u$}
    \prf{$\Gamma' \Vdash \{\theta'\}t'  = \{\theta\}t : \{\theta'\}\ann\tau$ \hfill by Lemma \ref{lem:semsym} (Conversion)}
    \prf{$\Gamma \models t' = t : \ann\tau$ \hfill by rule}
\\
\pcase{$\ibnc{\models \Gamma}{
             \Gamma \models \ann\tau : u \quad
             \forall \Gamma',~\theta,~\theta'.~\Gamma' \Vdash \theta = \theta' : \Gamma
\Longrightarrow \Gamma' \Vdash \{\theta\}t_1  =  \{\theta'\}t_2 : \{\theta\}\ann\tau}
           {\Gamma \models t_1 = t_2 : \ann\tau}{}$ and}
\pcase{$\ibnc{\models \Gamma}{
             \Gamma \models \ann\tau : u \quad
             \forall \Gamma',~\theta,~\theta'.~\Gamma' \Vdash \theta = \theta' : \Gamma
\Longrightarrow \Gamma' \Vdash \{\theta\}t_2  =   \{\theta'\}t_3 : \{\theta\}\ann\tau}{\Gamma \models t_2 = t_3 : \ann\tau}{}$}
\prf{Assume $\Gamma' \Vdash \theta = \theta' : \Gamma$}
\prf{$\Gamma' \Vdash \theta = \theta : \Gamma$ \hfill by symmetry and transitivity of substitution (Lemma \ref{lm:symsemsub})}
\prf{$\Gamma' \Vdash \{\theta\}t_1  = \{\theta\}t_2 : \{\theta\}\ann\tau$ \hfill by $\Gamma \models t_1 = t_2 : \ann\tau$}
\prf{$\Gamma' \Vdash \{\theta\}t_2  = \{\theta'\}t_3 : \{\theta\}\ann\tau$ \hfill by $\Gamma \models t_2 = t_3 : \ann\tau$}
\prf{$\Gamma' \Vdash \{\theta\}t_1  = \{\theta'\}t_3 : \{\theta\}\ann\tau$ \hfill by Lemma \ref{lem:semsym} (Transitivity)}
\prf{$\Gamma \models t_1 = t_3 : \ann\tau$ \hfill by rule}
}
\end{proof}}

\begin{lemma}[Function Type Injectivity Is Valid]
If\\ $\Gamma \models {(y:\ann\tau_1) \arrow \tau_2} = {(y:\ann\tau_1') \arrow \tau'_2} : u_3$, then
{$\Gamma\! \models \ann\tau_1 = \ann\tau_1' \!:\! u_1$} and
$\Gamma, y{:}\ann\tau_1 \models \tau_2 = \tau'_2 : u_2$ and \mbox{$(u_1,~u_2,~u_3) \in \Ru$}.
\end{lemma}
\begin{proof}
Proof by unfolding the semantic definitions.
\LONGVERSIONCHECKED{
\quad\\[0.5em]
\prf{$\Gamma \models (y:\ann\tau_1) \arrow \tau_2 = (y:\ann\tau_1') \arrow \tau'_2 : u_3$ \hfill by assumption}
\prf{$\models \Gamma$  \hfill}
\prf{$ \forall \Gamma',~\theta,~\theta'.~\Gamma' \Vdash \theta = \theta' : \Gamma
     \Longrightarrow \Gamma' \Vdash \{\theta\}(y:\ann\tau_1) \arrow \tau_2  =
                                     \{\theta'\}(y:\ann\tau_1') \arrow \tau'_2  : \{\theta\}u_3$ \hfill by def. of validity}
\\[0.15em]
\prf{To prove: $\Gamma \models \ann\tau_1 = \ann\tau_1' : u_1$ }
\prf{Assume an arbitrary $\Gamma' \Vdash \theta = \theta' : \Gamma$.}
\prf{$\Gamma' \Vdash \{\theta\}(y:\ann\tau_1) \arrow \tau_2  =
                    \{\theta'\}(y:\ann\tau_1') \arrow \tau'_2  : u_3$ \hfill by previous lines and $\{\theta\}u_3 = u_3$}
\prf{$\Gamma' \Vdash (y: \{\theta\}\ann\tau_1) \arrow \{\theta,~y/y\}\tau_2  =
                     (y: \{\theta'\}\ann\tau_1') \arrow \{\theta', y/y\}\tau'_2  : u_3$ \hfill by subst. def.}
\prf{$(y: \{\theta'\}\ann\tau_1') \arrow \{\theta',~y/y\}\tau'_2  \whnf
  (y: \{\theta'\}\ann\tau_1') \arrow \{\theta',~y/y\}\tau'_2$ \hfill by
  unfolding semantic def. and $\whnf$\\
\mbox{\quad}\hfill since $\norm ((y: \{\theta'\}\ann\tau_1') \arrow \{\theta',~y/y\}\tau'_2 )$}
\\[0.15em]
\prf{$\forall \Gamma_0 \leq_\rho \Gamma'.~
     \Gamma_0 \Vdash  \{\rho\}\{\theta\}\ann\tau_1 = \{\rho\}\{\theta\}\ann\tau_1' : u_1$ \hfill by sem. def.}
\prf{$\Gamma' \Vdash \{\theta\}\ann\tau_1 = \{\theta\}\ann\tau_1':u_1$
   \hfill choosing $\Gamma_0 = \Gamma'$ and $\rho = \id(\Gamma')$ and the fact that $\{\id(\Gamma')\}\theta = \theta$}
\prf{$\forall \Gamma', \theta, \theta'.~\Gamma' \Vdash \theta = \theta' : \Gamma \Longrightarrow \Gamma' \Vdash \{\theta\}\ann\tau_1 = \{\theta'\}\ann\tau_1' : u_1 $ \hfill previous lines }
\prf{$\Gamma \models \ann\tau_1 = \ann\tau_1' : u_1$ \hfill def. of validity ($\models$), since $\Gamma', \theta, \theta'$ arbitrary}
\\[0.15em]
\prf{To prove: $\Gamma, y{:}\ann\tau_1 \models \tau_2 = \tau'_2 : u_2$ }
\prf{Assume an arbitrary $\Gamma' \Vdash \theta_2 = \theta'_2 : \Gamma, y{:}\ann\tau_1$.}
\prf{$\theta_2 = \theta, s/y$ and $\theta'_2 = \theta', s'/y$ \\
 $\Gamma' \Vdash \theta = \theta' : \Gamma$ and
 $\Gamma' \Vdash s = s' : \{\theta\}\ann\tau_1$ \hfill by inversion on  $\Gamma' \Vdash \theta_2 = \theta'_2 : \Gamma, y{:}\ann\tau_1$}
\prf{$\forall \Gamma_0 \leq_\rho \Gamma'.
~\Gamma_0 \Vdash s = s' : \{\rho\}\{\theta\}\ann\tau_1  \Longrightarrow
 \Gamma' \Vdash \{\rho, s/y\}\{\theta, y/y\}\tau_2 = \{\rho,s'/y\}\{\theta',y/y\}\tau'_2$ \hfill by sem. def. }
\prf{$\Gamma' \Vdash s = s' : \{\theta\}\ann\tau_1 \Longrightarrow
     \Gamma' \Vdash \{s/y\}\{\theta,y/y\}\tau_2 = \{s'/y\}\{\theta',y/y\}\tau'_2$
  \hfill by choosing $ \rho = \id(\Gamma')$}
\prf{$\Gamma' \Vdash s = s' : \{\theta\}\ann\tau_1 \Longrightarrow
     \Gamma' \Vdash \{\theta,s/y\}\tau_2 = \{\theta',s'/y\}\tau'_2$ \hfill by subst. def.}
\prf{$\Gamma' \Vdash \{\theta, s/y\}\tau_2 = \{\theta', s'/y\}\tau'_2$ \hfill by previous line}
\prf{$\Gamma' \Vdash \{\theta_2\}\tau_2 = \{\theta'_2\}\tau'_2 $ \hfill by previous lines}
\prf{$\Gamma , y{:}\ann\tau_1 \models \tau_2 = \tau'_2 : u_2$ \hfill by def. of validity, since $\Gamma', \theta, \theta'$ arbitrary}
}
\end{proof}

\LONGVERSION{
\begin{theorem}[Completeness of Validity]\label{lem:SemComplete}
If $\Gamma \models t = t' : \ann\tau$ then $\Gamma \vdash t : \ann\tau$ and
$\Gamma \vdash t' : \ann\tau$ and
$\Gamma \vdash t \equiv t' : \ann\tau$ and
$\Gamma \vdash \ann\tau : u$.
\end{theorem}
\begin{proof}
Unfolding of validity definition relying on context satisfiability
(Lemma \ref{lem:ctxsat}) and Well-Formedness of Semantic Typing (Lemma \ref{lm:semwf}).
\LONGVERSIONCHECKED{\\[0.5em]\quad
\prf{$\Gamma \models t = t' : \ann\tau$ \hfill by assumption}
\prf{$\models \Gamma$ \hfill by validity definition    }
\prf{$\vdash \Gamma$ and $\Gamma \Vdash \id(\Gamma) = \id(\Gamma) : \Gamma$ \hfill by Lemma \ref{lem:ctxsat}}
\prf{$\forall \Gamma',~\theta,~\theta'.~\Gamma' \Vdash \theta = \theta' :
                  \Gamma \Longrightarrow \Gamma' \Vdash \{\theta\}t  =
                  \{\theta\}t' : \{\theta\}\ann\tau$ \hfill by validity definition    }
\prf{$\Gamma \Vdash \{\id(\Gamma)\}t = \{\id(\Gamma)\}t' : \{\id(\Gamma)\}\ann\tau$ \hfill by previous lines }
\prf{$\Gamma \Vdash t = t' : \ann\tau$ \hfill by subst. def. }
\prf{$\Gamma \vdash t' : \ann\tau$,~ $\Gamma \vdash t: \ann\tau$,~
$\Gamma \vdash t \equiv t' : \ann\tau$ and $\Gamma \vdash \ann\tau : u$ \hfill
    by Well-Formedness of Seman. Typ. (Lemma \ref{lm:semwf})}
}
\end{proof}
}


The fundamental lemma (Lemma \ref{lm:fundtheo}) states that well-typed
terms are valid. The proof proceeds by mutual induction on the typing derivation for LF-objects and computations. It relies on the validity of type conversion, computation-level functions, applications, and recursion. To establish these properties, we require symmetry, transitivity of semantic equality, and semantic type conversion (Lemma \ref{lem:semsym}).

\LONGVERSION{
\begin{lemma}[Validity of Type Conversion]\label{lem:validtypconv}
If $\Gamma \models \ann\tau = \ann\tau' : u$ and $\Gamma \models t : \ann\tau$ then $\Gamma \models t : \ann\tau'$.
\end{lemma}
\begin{proof}
By definition relying on semantic type conversion lemma (Lemma \ref{lem:semsym} (\ref{it:conv})).
 \LONGVERSIONCHECKED{
\quad\\[0.5em]
\prf{$\Gamma \models t : \ann\tau$ \hfill by assumption }
\prf{$\Gamma \models t = t : \ann\tau$ \hfill by validity def.}
\prf{Assume $\Gamma' \Vdash \theta = \theta' : \Gamma$}
\prf{$\Gamma' \Vdash \{\theta\}t = \{\theta'\}t:\{\theta\}\ann\tau$ \hfill by validity def. $\Gamma \models t = t : \ann\tau$}
\prf{$\Gamma \Vdash \{\theta\}\ann\tau = \{\theta'\}\ann\tau' : u$ \hfill by validity def. $\Gamma \models \ann\tau = \ann\tau' : u$}
\prf{$\Gamma \Vdash \{\theta\}t = \{\theta'\} : \{\theta\}\ann\tau'$ \hfill by Lemma \ref{lem:semsym} (Conversion)}
\prf{$\Gamma \models t = t : \ann\tau'$ \hfill since $\Gamma', \theta,\theta'$ arbitrary}
}
\end{proof}

\begin{lemma}[Validity of Functions]\label{lem:ValidFun}
If  $\Gamma, y{:}\ann\tau_1 \models t : \tau_2$ then $\Gamma \models \tmfn y t : (y:\ann\tau_1) \arrow \tau_2$.
\end{lemma}
\begin{proof}
We unfold the validity definitions, relying on 
completeness of validity (Lemma \ref{lem:SemComplete}),
semantic weakening of computation-level substitutions (Lemma \ref{lem:weakcsub}),
Well-formedness Lemma \ref{lm:semwf},  Backwards Closure Lemma \ref{lem:bclosed},  Symmetry property of semantic equality (Lemma \ref{lem:semsym}).
\LONGVERSIONCHECKED{
\quad\\[0.5em]
\prf{$\Gamma, y{:}\ann\tau_1 \models t : \tau_2$ \hfill by assumption}
\prf{$\Gamma, y:\ann\tau_1 \models t = t : \tau_2$ \hfill by def. validity}
\prf{$\forall \Gamma',~\theta,~\theta'.
~\Gamma' \Vdash \theta = \theta' :\Gamma, y{:}\ann\tau_1
\Longrightarrow
\Gamma' \Vdash \{\theta\}t  = \{\theta'\}t : \{\theta\}\tau_2$ \hfill by def. of validity}
\prf{
$\Gamma, y:\ann\tau_1 \models \tau_2 : u$ \hfill by def. of validity\\
$\Gamma, y:\ann\tau_1 \models \tau_2 = \tau_2 : u$ \hfill by inversion on validity\\
$\models \Gamma, y:\ann\tau_1$ \hfill by inversion on validity\\
$\models \Gamma$ \hfill by inversion on validity\\
$\Gamma \models \ann\tau_1 : u$ \hfill by inversion on validity \\
$\Gamma \models \ann\tau_1 = \ann\tau_1 : u$ \hfill by inversion on validity
\\[-0.5em]
}
TO SHOW:
  \begin{enumerate}
  \item $\models \Gamma$
  \item $\Gamma \models (y:\ann\tau_1) \arrow \tau_2 : u$, i.e. \\
$\forall \Gamma', ~\theta,~\theta'. ~\Gamma' \Vdash \theta = \theta' :\Gamma \Longrightarrow \Gamma' \Vdash \{\theta\}((y:\ann\tau_1) \arrow \tau_2) = \{\theta'\}((y:\ann\tau_1) \arrow \tau_2) : u$
  \item $\forall \Gamma', ~\theta,~\theta'. ~
    \Gamma' \Vdash \theta = \theta' :\Gamma
    \Longrightarrow
    \Gamma' \Vdash \{\theta\}(\tmfn y t) = \{\theta'\}(\tmfn y t) : \{\theta\}((y:\ann\tau_1) \arrow \tau_2)$
  \end{enumerate}
\prf{~\\[0.5em](1) SHOW: $\models \Gamma$}
\prf{$\models \Gamma, y{:}\ann\tau_1$ \hfill by assumption $\Gamma, y:\ann\tau_1 \models t = t : \tau_2$}
\prf{$\models \Gamma$ \hfill by inversion on $\models \Gamma, y{:}\ann\tau_1$\\[-0.5em]}
\prf{(2) SHOW: $\forall \Gamma' \Vdash \theta = \theta' :\Gamma \Longrightarrow \Gamma' \Vdash \{\theta\}((y:\ann\tau_1) \arrow \tau_2) = \{\theta'\}((y:\ann\tau_1) \arrow \tau_2) : u$}
\prf{Assume $\Gamma' \Vdash \theta = \theta' :\Gamma$ \\[-0.75em]}
\prf{(2.a) SHOW: $\forall \Gamma'' \leq_\rho \Gamma'. \Gamma'' \Vdash s = s' : \{\rho\}\{\theta\}\ann\tau_1 \Longrightarrow
 \Gamma'' \Vdash \{\{\rho\}\theta, s/y\}\tau_2 = \{\{\rho\}\theta', s'/y\}\tau_2 : u$}
\prf{\mbox{$\quad$}Assume $\Gamma'' \leq_\rho \Gamma'. \Gamma'' \Vdash s = s' : \{\rho\}\{\theta\}\ann\tau_1$}
\prf{\mbox{$\quad$}$\Gamma'' \Vdash \{\rho\}\theta = \{\rho\}\theta' : \Gamma$ \hfill by weakening sem. subst. (Lemma \ref{lem:weakcsub})}
\prf{\mbox{$\quad$}$\Gamma'' \Vdash \{\rho\}\theta, s/y = \{\rho\}\theta', s'/y : \Gamma, y{:}\ann\tau_1$ \hfill by sem. def.}
\prf{\mbox{$\quad$}$\Gamma'' \Vdash \{\{\rho\}\theta, s/y\}\tau_2 = \{\{\rho\}\theta', s'/y\}\tau_2 : u$ \hfill by def. of $\Gamma, y:\ann\tau_1 \models \tau_2 = \tau_2 : u$ \\[-0.5em]}
%
\prf{(2.b) SHOW: $\forall \Gamma'' \leq_\rho \Gamma'. \Gamma'' \Vdash \{\rho\}\{\theta\}\ann\tau_1 = \{\rho\}\{\theta'\}\ann\tau_1 : u$\\}
\prf{\mbox{$\quad$}$\Gamma'' \Vdash \{\rho\}\theta = \{\rho\}\theta' : \Gamma$ \hfill by weakening sem. subst. (Lemma \ref{lem:weakcsub})}
\prf{\mbox{$\quad$}$\Gamma'' \Vdash \{\rho\}\{\theta\}\ann\tau_1 = \{\rho\}\{\theta'\}\ann\tau_1 : u$ \hfill by def. of $\Gamma \models \ann\tau_1 = \ann\tau_1 : u$\\[-0.5em] }
\prf{(2.c) SHOW: $\Gamma' \vdash \{\theta'\}((y:\ann\tau_1) \arrow \tau_2) \whnf \{\theta'\}((y:\ann\tau_1) \arrow \tau_2) :u$}
\prf{\mbox{$\quad$}$\Gamma \vdash (y:\ann\tau_1) \arrow \tau_2 : u$ \hfill by Completeness of Validity (Lemma \ref{lem:SemComplete})\\
\mbox{\hspace{1cm}}\hfill (using validity of functions which we show under (3))}
\prf{\mbox{$\quad$}$\Gamma' \vdash \theta : \Gamma$ and $\Gamma' \vdash \theta' : \Gamma$ \hfill by Well-formedness of semantic subst. (Lemma \ref{lem:wfsemsub})}
\prf{\mbox{$\quad$}$\Gamma' \vdash \{\theta'\}((y:\ann\tau_1) \arrow \tau_2) : u$ \hfill by computation-level subst. lemma (Lemma \ref{lm:compsub} )}
\prf{\mbox{$\quad$}$\norm ((y:\{\theta'\}\ann\tau_1) \arrow \{\theta', y/y\}\tau_2)$ \hfill by def. and subst. prop.\\}
\prf{(3) SHOW: $\forall \Gamma', ~\theta,~\theta'. ~
    \Gamma' \Vdash \theta = \theta' :\Gamma
    \Longrightarrow
    \Gamma' \Vdash \{\theta\}(\tmfn y t) = \{\theta'\}(\tmfn y t) : \{\theta\}((y:\ann\tau_1) \arrow \tau_2)$}
\prf{Assume $\Gamma' \Vdash \theta = \theta' :\Gamma$; Showing: $\Gamma' \Vdash \{\theta\}(\tmfn y t) = \{\theta'\}(\tmfn y t) : \{\theta\}((y:\ann\tau_1) \arrow \tau_2)$ \\[-0.75em]}
\prf{(3.a) SHOW: $\Gamma' \vdash \{\theta\}(\tmfn y t) \whnf w : \{\theta\}((y:\ann\tau_1) \arrow \tau_2)$ \\\mbox{$\quad\qquad$}and
                 $\Gamma' \vdash \{\theta'\}(\tmfn y t) \whnf w' : \{\theta\}((y:\ann\tau_1) \arrow \tau_2)$ }
\prf{\mbox{$\quad$}$\tmfn y \{\theta,y/y\} t \whnf w$ \hfill where $w = \tmfn y \{\theta,y/y\}t$ since $\norm (\tmfn y \{\theta,y/y\}t)$}
\prf{\mbox{$\quad$}$\tmfn y \{\theta',y/y\} t \whnf w'$ \hfill where $w' = \tmfn y \{\theta',y/y\}t$ since $\norm (\tmfn y \{\theta',y/y\}t)$}
\prf{\mbox{$\quad$}$\Gamma, y{:}\ann\tau_1 \vdash t : \tau_2$ \hfill by Well-formedness of semantic typing (Lemma \ref{lm:semwf})}
\prf{\mbox{$\quad$}$\Gamma' \vdash \theta : \Gamma$ and $\Gamma' \vdash \theta' : \Gamma$ \hfill by Well-formedness of semantic subst. (Lemma \ref{lem:wfsemsub})}
\prf{\mbox{$\quad$}$\Gamma', y{:}\{\theta\}\ann\tau_1 \vdash \theta, y/y : \Gamma, y{:}\ann\tau_1$
     and $\Gamma', y{:}\{\theta'\}\ann\tau_1 \vdash \theta', y/y : \Gamma, y{:}\ann\tau_1'$  \hfill by comp. subst.}
\prf{\mbox{$\quad$}$\Gamma', y{:}\{\theta\}\tau_1 \vdash \{\theta, y/y\} t : \{\theta, y/y\}\tau_2$ and \\
     \mbox{$\quad$}$\Gamma', y{:}\{\theta'\}\tau_1 \vdash \{\theta', y/y\} t : \{\theta', y/y\}\tau_2$ \hfill by computation-level subst. lemma (Lemma \ref{lm:compsub})}
\prf{\mbox{$\quad$}$\Gamma' \vdash \{\theta\}(\tmfn y t) : \{\theta\}((y:\ann\tau_1) \arrow \tau_2)$ and\\
     \mbox{$\quad$}$\Gamma' \vdash \{\theta'\}(\tmfn y t) : \{\theta'\}((y:\ann\tau_1) \arrow \tau_2)$ \hfill by typing rule}
\prf{\mbox{$\quad$}$\Gamma' \vdash \{\theta\}((y:\ann\tau_1) \arrow \tau_2) = \{\theta'\}((y:\ann\tau_1) \arrow \tau_2) : u$ \hfill by (2)}
\prf{\mbox{$\quad$}$\Gamma' \vdash \{\theta'\}(\tmfn y t) : \{\theta\}((y:\ann\tau_1) \arrow \tau_2)$ \hfill Conversion (Lemma \ref{lem:semsym}(\ref{it:conv}))\\[0.5em]}
\prf{(3.b) SHOW: $\forall \Gamma'' \leq_\rho \Gamma'. \Gamma'' \Vdash s = s' : \{\rho\}\{\theta\}\ann\tau_1 \Longrightarrow
  \Gamma'' \Vdash \{\rho\}w~s = \{\rho\}w'~s' : \{\rho, s/y\}\{\theta, y/y\}\tau_2$}
\prf{\mbox{$\quad$}Assume $ \Gamma'' \leq_\rho \Gamma'. \Gamma'' \Vdash s = s' : \{\rho\}\{\theta\}\ann\tau_1$}
\prf{\mbox{$\quad$}$\Gamma'' \Vdash \{\rho\}\theta = \{\rho\}\theta' : \Gamma$ \hfill by weakening sem. subst. (Lemma \ref{lem:weakcsub})}
\prf{\mbox{$\quad$}$\Gamma'' \Vdash \{\rho\}\theta, s/y = \{\rho\}\theta', s'/y : \Gamma, y{:}\ann\tau_1$ \hfill by sem. def.}
\prf{\mbox{$\quad$}$\Gamma'' \Vdash \{\{\rho\}\theta, s/y\}t = \{\{\rho\}\theta', s'/y\}t :  \{\{\rho\}\theta, s/y\}\tau_2$ \hfill using $\Gamma, y:\ann\tau_1 \models t = t : \tau_2$}
\prf{\mbox{$\quad$}$\Gamma'' \Vdash  \{\{\rho\}\theta, s/y\}t \whnf v :  \{\{\rho\}\theta, s/y\}\tau_2$ \\
\mbox{$\quad$} $\Gamma''\Vdash \{\{\rho\}\theta', s'/y\}t \whnf v':  \{\{\rho\}\theta, s/y\}\tau_2$
\hfill sem. definition $\Vdash$}
\prf{\mbox{$\quad$}$\Gamma'' \Vdash \{\rho\}\{\theta\}\tmfn y t \whnf \{\rho\}\{\theta\}\tmfn y t : \{\rho\}\{\theta\}((y:\ann\tau_1) \arrow \tau_2)$ \hfill  since $\norm (\{\rho\}\{\theta\}\tmfn y t)$ \\
\mbox{$\quad$}\hfill and $\Gamma'' \vdash  \tmfn y \{\{\rho\}\theta, y/y\}t : \{\rho\}\{\theta\}((y:\ann\tau_1) \arrow \tau_2)$}
\prf{\mbox{$\quad$}$\Gamma'' \Vdash (\tmfn y  \{\{\rho\}\theta,y/y\}t)~s \whnf v : \{\{\rho\}\theta, s/y\}\tau_2$
\hfill by rules (typing and $\whnf$)}
\prf{\mbox{$\quad$}$\Gamma'' \Vdash (\tmfn y \{\{\rho\}\theta,y/y\}t) ~s =
       \{\{\rho\}\theta', s'/y\}t : \{\{\rho\}\theta, s/y\}\tau_2$ \hfill by Backwards Closed (Lemma \ref{lem:bclosed})
}
\prf{\mbox{$\quad$} $\Gamma'' \Vdash \{\{\rho\}\theta, s/y\}\tau_2 = \{\{\rho\}\theta', s'/y\}\tau_2 : u$
\hfill by $\Gamma, y:\ann\tau_1 \models \tau_2 = \tau_2 : u$ }
\prf{\mbox{$\quad$} $\Gamma'' \Vdash  \{\{\rho\}\theta', s'/y\}t = (\tmfn y \{\{\rho\}\theta,y/y\}t) ~s : \{\{\rho\}\theta, s/y\}\tau_2$
\hfill Symmetry (Lemma \ref{lem:semsym}(\ref{it:sym}))}
\prf{\mbox{$\quad$}$\Gamma'' \Vdash \{\rho\}\{\theta'\}\tmfn y t \whnf \{\rho\}\{\theta'\}\tmfn y t : \{\rho\}\{\theta\}((y:\ann\tau_1) \arrow \tau_2)$ \hfill  since $\norm (\{\rho\}\{\theta'\}\tmfn y t)$ \\
\mbox{$\quad$}\hfill and $\Gamma'' \vdash  \tmfn y \{\{\rho\}\theta', y/y\}t : \{\rho\}\{\theta\}((y:\ann\tau_1) \arrow \tau_2)$}
\prf{\mbox{$\quad$}$\Gamma'' \Vdash (\tmfn y  \{\{\rho\}\theta',y/y\}t)~s' \whnf v' : \{\{\rho\}\theta, s/y\}\tau_2$
\hfill by rules (typing and $\whnf$)}
\prf{\mbox{$\quad$}$\Gamma'' \Vdash (\tmfn y  \{\{\rho\}\theta',y/y\}t)~s' = (\tmfn y  \{\{\rho\}\theta,y/y\}t)~s : \{\rho\{\theta\}, s/y\}\tau_2$ \hfill  \\
\mbox{$\quad$}\hfill by Backwards Closed (Lemma \ref{lem:bclosed})}
\prf{\mbox{$\quad$}$\Gamma'' \Vdash (\tmfn y  \{\{\rho\}\theta,y/y\}t)~s = (\tmfn y  \{\{\rho\}\theta',y/y\}t)~s': \{\rho\{\theta\}, s/y\}\tau_2$   \\
\mbox{$\quad$}\hfill by  Symmetry (Lemma \ref{lem:semsym}(\ref{it:sym}))}
\prf{\mbox{$\quad$}$\Gamma'' \Vdash (\{\rho\}\tmfn y  \{\theta,y/y\}t)~s = (\{\rho\}\tmfn y  \{\theta',y/y\}t)~s': \{\rho, s/y\}\{\theta, y/y\}\tau_2$ \hfill by comp. subst.}
\prf{\mbox{$\quad$}$\Gamma'' \Vdash \{\rho\}w~s = \{\rho\}w'~s' : \{\rho, s/y\}\{\theta, y/y\}\tau_2$ \hfill since $w = \tmfn y \{\theta,y/y\}t$  \\\mbox{$\quad$}\hfill and $w' = \tmfn y \{\theta',y/y\}t$}
}
\end{proof}

\begin{lemma}[Validity of Recursion] \label{lem:ValidRec}
\[
  \begin{array}{ll@{~}c@{~}l}
\multicolumn{4}{p{14cm}}{Let
$\Gamma \vdash \tmrec {\IH} {b_v} {b_{\mathsf{app}}} {b_{\tlam}} \rappto \Psi~t: \{\Psi/\psi, t/m\}\tau$
and $\IH = (\psi:\tmctx) \arrow (m:\cbox{\unboxc{\psi} \vdash \tm}) \arrow \tau$.}
\\[0.25em]
\mbox{If} & \Gamma \models \IH : u ~\mbox{and}~
 \Gamma \models t : \cbox{\Psi \vdash \tm}~\mbox{and}\\
 & \Gamma, \psi:{\tmctx}, p:\cbox{\unboxc{\psi} \vdash_\# \tm} & \models &  t_v : \{p / y\}\tau ~\mbox{and}\\
& \Gamma,  \psi:{\tmctx}, m:\cbox{\unboxc{\psi} \vdash \tm}, n:\cbox{\unboxc{\psi} \vdash \tm},\\
& \qquad        f_m: \{m/y\}\tau, f_n: \{n/y\}\tau & \models & t_{\mathsf{app}} : \{\cbox{\unboxc{\psi} \vdash\tapp~\unbox{m}{\id}~\unbox{n}{\id}}/y\}\tau ~\mbox{and}\\
& \Gamma, \psi:{\tmctx},  m:\cbox{\unboxc{\psi}, x:\tm \vdash \tm}, & & \\
&  \qquad        f_m:\{\cbox{\unboxc{\psi}, x:\tm}/\psi, m /y \} \tau & \models &  t_{\tlam} : \{\cbox{\unboxc{\psi} \vdash \tlam~\lambda x.\unbox{m}{\id}~} / y\}\tau
\\[0.25em]
\multicolumn{4}{p{13cm}}{then $\Gamma \models \tmrec {\IH} {b_v} {b_{\mathsf{app}}} {b_{\tlam}} \rappto \Psi~t:   \{\Psi/\psi, t/m\}\tau$.}
  \end{array}
\]
\end{lemma}
\begin{proof}
We assume $\Gamma' \Vdash \theta = \theta' : \Gamma$, and show
$\Gamma' \Vdash \{\theta\}(\trec{\R}{\tm}\IH \rappto \Psi~t )
    = \{\theta'\}(\trec{\R}{\tm}\IH \rappto \Psi~t): \{\theta\}\{\Psi/\psi,~t/m\}\tau$ by considering different cases for
$\Gamma' \Vdash \{\theta\}t = \{\theta'\}t : \cbox{\{\theta\}\Psi \vdash \tm}$. Let $\Psi' = \{\theta\}\Psi$. In the case where $\Gamma' \vdash \{\theta\}t \whnf \cbox{\hatctx\Psi' \vdash M} : \cbox{\Psi' \vdash \tm}$ and $\Gamma' \vdash \{\theta'\}t \whnf \cbox {\Psi' \vdash N} : \cbox{\{\theta\}\Psi \vdash \tm}$ and $\Gamma' ; \Psi' \Vdash M = N : \tm$, we proceed by an inner induction on $\Gamma' ; \Psi' \Vdash M = N : \tm$ exploiting on Back. Closed Lemma (\ref{lem:bclosed}).
\LONGVERSIONCHECKED{
$\quad$\\[1em]
We now give the full proof.
\\[1em]
Let $\trec{\R}{\tm}\IH = \tmrec {\IH} {b_v} {b_{\mathsf{app}}} {b_{\tlam}}$
\\[1em]
\prf{Assume $\Gamma' \Vdash \theta = \theta' : \Gamma$; \\
     TO SHOW: $\quad\Gamma' \Vdash \{\theta\}(\trec{\R}{\tm}\IH \rappto \Psi~t )
    = \{\theta'\}(\trec{\R}{\tm}\IH \rappto \Psi~t): \{\theta\}\{\Psi/\psi,~t/m\}\tau$
}
\\
\prf{$\Gamma' \Vdash \{\theta\}t = \{\theta'\}t : \cbox{\{\theta\}\Psi \vdash \tm}$ \hfill by validity of $\Gamma \models t : \cbox{\Psi \vdash \tm}$\\[-0.5em]}
 \prf{Let $\Psi' = \{\theta\}\Psi$. We now proceed to prove:
 \\[0.25em]
\pcase{$\Gamma' \vdash \{\theta\}t \whnf w : \cbox{\Psi' \vdash \tm}$  \\
     \mbox{\hspace{0.2cm}}and $\Gamma' \vdash \{\theta'\}t \whnf w' : \cbox{\Psi' \vdash \tm}$ \\
     \mbox{\hspace{0.2cm}}and $\Gamma' ; \Psi' \Vdash \unbox w \id = \unbox {w'}\id : \tm$
}
\prf{We write $M$ for $\unbox w \id$ and $N$ for $\unbox {w'}\id$ below.}
 \begin{center}
 \begin{tabular}{p{13cm}}
 If $\Gamma' ; \Psi' \Vdash M = N : \tm$ \\
     then $\Gamma' \Vdash \{\theta\}(\trec{\R}{\tm}{\IH}~\rappto \cbox{\Psi})~\cbox{\hatctx{\Psi} \vdash M}
              = \{\theta'\}(\trec{\R}{\tm}{\IH}~\rappto \cbox{\Psi})~ \cbox{\hatctx{\Psi} \vdash N} :
      \{\theta,~\Psi'/\psi, \cbox{\hatctx{\Psi} \vdash M}/m\}\tau$   \\[0.25em] $\quad$
 \end{tabular}
 \end{center}
 by induction on $M$, i.e. we may appeal to the induction hypothesis if the term $M$ has made progress and has stepped using $\lfwhnf$ and hence is ``smaller'' \\[0.5em]}
 \prf{\emph{Sub-case}:
   $\Gamma'; \Psi' \vdash M \lfwhnf \tapp~M_1~M_2 : \tm$ and $\Gamma' ; \Psi' \vdash N \lfwhnf \tapp~N_1~N_2 : \tm$ \\
 \mbox{\qquad\qquad}$\Gamma'; \Psi' \Vdash M_1 = N_1 : \tm$ and $\Gamma' ; \Psi' \Vdash M_2 = N_2 : \tm$\\[-0.5em]}
 \prf{$\Gamma' \Vdash \{\theta\}(\trec{\R}{\tm}{\IH}~\rappto \cbox{\Psi})~\cbox{\hatctx{\Psi} \vdash M_1}
              = \{\theta'\}(\trec{\R}{\tm}{\IH}~\rappto \cbox{\Psi})~\cbox{\hatctx{\Psi} \vdash N_1} :
      \{\theta,~\Psi'/\psi, \cbox{\hatctx{\Psi} \vdash M_1/m}\}\tau$
      \hfill by IH(i)}
 \prf{$\Gamma' \Vdash \{\theta\}(\trec{\R}{\tm}{\IH}~\rappto \cbox{\Psi})~\cbox{\hatctx{\Psi'} \vdash M_2}
                    = \{\theta'\}(\trec{\R}{\tm}{\IH}~\rappto \cbox{\Psi'})~\cbox{\hatctx{\Psi'} \vdash N_2} :
      \{\theta,~\Psi'/\psi, \cbox{\hatctx{\Psi} \vdash M_2}/m\}\tau $
      \hfill by IH(i)}
 \prf{$\Gamma' \vdash \{\theta\}(\trec{\R}{\tm}{\IH}~\rappto \Psi)~\cbox{\hatctx \Psi \vdash M_i} :  \{\theta,~\Psi'/\psi, \cbox{\hatctx\Psi \vdash M_i}/m\}\tau$ \hfill by Well-form. of Sem. Def.}
 \prf{$\Gamma' \vdash \{\theta\}(\trec{\R}{\tm}{\IH}~\rappto \Psi)~\cbox{\hatctx \Psi \vdash N_i} :  \{\theta',~\Psi'/\psi, \cbox{\hatctx\Psi \vdash N_i}/m\}\tau$ \hfill by Well-form. of Sem. Def. and Type Conv.}
 \prf{let $\theta_\tapp = \theta, \Psi'/\psi, \cbox{\hatctx{\Psi} \vdash M_1} / m,~\cbox{\hatctx{\Psi} \vdash M_2} / n\\
 \mbox{\qquad\quad$\;$\quad} \{\theta\}(\trec{\R}{\tm}{\IH}~\rappto \cbox{\Psi})~\cbox{\hatctx{\Psi} \vdash M_1} / f_m,
  \{\theta\}(\trec{\R}{\tm}{\IH}~\rappto \Psi)~\cbox{\hatctx{\Psi} \vdash M_2} / f_n
          $
 }
 \prf{let $\theta'_\tapp = \theta', \Psi'/\psi, \cbox{\hatctx{\Psi} \vdash  N_1} / m,~\cbox{\hatctx{\Psi} \vdash  N_2} / n\\
 \mbox{\qquad\quad$\;$\quad} \{\theta'\}(\trec{\R}{\tm}{\IH} \rappto {\Psi})~\cbox{\hatctx{\Psi} \vdash N_1} / f_m,
  \{\theta'\}(\trec{\R}{\tm}{\IH} \rappto {\Psi})~\cbox{\hatctx{\Psi} \vdash N_2} / f_n
          $
 }
 \prf{let $\Gamma_\tapp =  \Gamma,  \psi:\cbox{\tmctx}, m:\cbox{\unboxc{\psi} \vdash \tm}, n:\cbox{\unboxc{\psi} \vdash \tm},  f_m: \{m/y\}\tau, f_n: \{n/y\}\tau$}
 \prf{$\Gamma' \Vdash \theta_\tapp = \theta'_\tapp : \Gamma_\tapp$}
 \prf{$\Gamma' \Vdash \{\theta_\tapp\}t_\tapp = \{\theta'_\tapp\}t_\tapp :
       \{\theta_\tapp\} \{\cbox{\psi \vdash\tapp~\unbox{m}{\id}~\unbox{n}{\id}}/y\}\tau$
    \hfill by sem. def. of $t_\tapp$}
 \prf{$\Gamma' \Vdash  \{\theta_\tapp\}t_\tapp = \{\theta'_\tapp\}t_\tapp :
      \{\theta, \Psi'/\psi,~\cbox{\hatctx \Psi \vdash\tapp~M_1~M_2}/m\}\tau$
    \hfill by subst. def.}
  \prf{$\Gamma' \Vdash C = C : (\{\theta\}\Psi \vdash \tm)$ \hfill by reflexivity (Lemma \ref{lem:semsymlf})
 where $C = \hatctx\Psi \vdash \tapp~M_1~M_2$ }
 \prf{$\Gamma' \Vdash \{\theta\}t = \cbox{\hatctx{\Psi} \vdash \tapp~M_1~M_2} : \Psi' \vdash \tm$
    \hfill since $\Gamma' \vdash \{\theta\}t \whnf \cbox{C} : (\{\theta\}\Psi \vdash \tm)$ }
 \prf{$\Gamma' \Vdash \{\theta, \Psi'/\psi,~\{\theta\}t/m\}\tau =  \{ \theta,~\Psi'/\psi, \cbox{\hatctx\Psi \vdash \tapp~M_1~M_2}/m\}\tau : u$
        \hfill by sem. def. of $\Gamma \Vdash \IH : u$}
 \prf{$\Gamma' \Vdash  \{\theta_\tapp\}t_\tapp = \{\theta'_\tapp\}t_\tapp :
         \{\theta, \Psi'/\psi,~\{\theta\}t/m\}\tau$ \hfill by type conversion}
 \prf{$\Gamma' \vdash \{\theta_\tapp\}t_\tapp \whnf v :  \{\theta, \Psi'/\psi,~\{\theta\}t/m\}\tau$ \hfill by previous sem. def. }
 \prf{$\Gamma' ; \Psi' \vdash  \{\theta\}(\trec{\R}{\tm}{\IH}~@~\Psi~t) \whnf v :   \{\theta, \Psi'/\psi,~\{\theta\}t/m\}\tau$ \\
 \mbox{\qquad} \hfill since $\Gamma' \vdash \{\theta\}t \whnf \cbox{\hatctx \Psi \vdash M}: \tm$ and $\Gamma' ; \Psi' \vdash M \lfwhnf \tapp~M_1~M_2 : \tm$}
 \prf{$\Gamma' \Vdash \{\theta\}(\trec{\R}{\tm}{\IH}~@~\Psi~t) =  \{\theta'_\tapp\}t_\tapp :
         \{\theta, \Psi'/\psi,~\{\theta\}t/m\}\tau$ \hfill by Back. Closed (Lemma \ref{lem:bclosed})}
 \prf{$\Gamma' \vdash \{\theta'_\tapp\}t_\tapp \whnf v' :  \{\theta, \Psi'/\psi,~\{\theta\}t/m\}\tau$ \hfill by previous sem. def. }
 \prf{$\Gamma' ; \Psi' \vdash  \{\theta'\}(\trec{\R}{\tm}{\IH}~@~\Psi~t) \whnf v' :   \{\theta, \Psi'/\psi,~\{\theta\}t/m\}\tau$ \hfill (using type conversion)\\
 \mbox{\qquad} \hfill since $\Gamma' \vdash \{\theta'\}t \whnf \cbox{\hatctx \Psi \vdash N} : \tm$ and $\Gamma' ; \Psi' \vdash N \lfwhnf \tapp~N_1~N_2 : \tm$}
 \prf{$\Gamma' \Vdash \{\theta\}(\trec{\R}{\tm}{\IH}~@~\Psi~t) = \{\theta'\}(\trec{\R}{\tm}{\IH}~@~\Psi~t) :
         \{\theta, \Psi'/\psi,~\{\theta\}t/m\}\tau$ \hfill by Back. Closed (Lemma \ref{lem:bclosed})}
 \prf{$\Gamma \models \trec{\R}{\tm}{\IH}\rappto \Psi~t : \IH$ \hfill since $\Gamma', \theta, \theta'$ were arbitrary\\[0.5em]}
 \prf{\emph{Sub-Case.}  Other Sub-Cases where $\Gamma' ; \Psi' \vdash \{\theta\}M \lfwhnf \tlam~\lambda x.M : \tm$ and
  $\Gamma' ; \Psi' \vdash \{\theta\}M \lfwhnf x : \tm$ where $x:\tm \in \Psi'$ are similar.\\[0.5em]}
 \prf{\emph{Sub-Case.} $\Gamma' ; \Psi' \vdash \{\theta\}M \lfwhnf r_1  : \tm$ where $r_1 = \unbox{t_1}{\sigma_1}$ and $\neut t_1$\\
 \mbox{\hspace{0.8cm}}and $\Gamma' ; \Psi' \vdash \{\theta'\}M \lfwhnf r_2 : \tm$ where $r_2 = \unbox{t_2}{\sigma_2}$ and $\neut t_2$ }
\prf{$\Gamma' ; \Psi' \Vdash \sigma_1 = \sigma_2 : \Phi$ \hfill since $\Gamma' ; \Psi' \Vdash_{LF} \{\theta\}M = \{\theta'\}M : \tm$}
\prf{$\Gamma' ; \Psi' \vdash \sigma_1 \equiv \sigma_2 : \Phi$ \hfill by Well-formedness Lemma}
\prf{$\Gamma' \vdash t_1 \equiv t_2 : \cbox{\Phi \vdash \tm}$ \hfill since $\Gamma' ; \Psi' \Vdash_{LF} \{\theta\}M = \{\theta'\}M : \tm$}
\prf{$\Gamma' ; \Psi' \vdash r_1 \equiv r_2 : \tm$}
\prf{$\Gamma' \vdash \{\theta\}\Psi \equiv \{\theta'\}\Psi : \tmctx$ \hfill Sem. Subst. Preserve Equiv. (Lemma \ref{lem:semsubst})\\
\mbox{\hspace{1cm}}\hfill since $\Gamma' \vdash \theta \equiv \theta' : \Gamma$}
\prf{$\Gamma' \vdash \{\theta\}(\trec{\R}{\tm}{\IH}\rappto \Psi)~\cbox{\hatctx\Psi \vdash r_1}
\equiv \{\theta'\}(\trec{\R}{\tm}{\IH}\rappto \Psi)~\cbox{\hatctx\Psi \vdash r_2} : \{\theta, \Psi'/\psi,~\{\theta\}t/m\}\tau$ \hfill by $\equiv$}
\prf{$\Gamma' \vdash \{\theta\}(\trec{\R}{\tm}{\IH}\rappto \Psi)~\{\theta\}t \whnf \{\theta\}(\trec{\R}{\tm}{\IH}\rappto \Psi)~\cbox{\hatctx\Psi \vdash r_1}  : \{\theta, \Psi'/\psi,~\{\theta\}t/m\}\tau$  \hfill by $\whnf$}
\prf{$\Gamma' \vdash \{\theta'\}(\trec{\R}{\tm}{\IH}\rappto \Psi)~\{\theta'\}t \whnf \{\theta'\}(\trec{\R}{\tm}{\IH}\rappto \Psi)~\cbox{\hatctx\Psi \vdash r_2} :  \{\theta, \Psi'/\psi,~\{\theta\}t/m\}\tau$ \hfill by $\whnf$ and type conversion}
\prf {$\Gamma' \Vdash \{\theta\}(\trec{\R}{\tm}{\IH}~@~\Psi~t) = \{\theta'\}(\trec{\R}{\tm}{\IH}~@~\Psi~t) :
         \{\theta, \Psi'/\psi,~\{\theta\}t/m\}\tau$ \hfill Back. Closed (Lemma \ref{lem:bclosed})\\
\mbox{\hspace{1cm}}\hfill and Symmetry}
 %
 \pcase{$\neut w$ and $\neut w'$ and $\Gamma' \vdash w \equiv w' : \cbox{\Psi' \vdash \tm}$}
\prf{$\neut \{\theta\}(\trec{\R}{\tm}{\IH}~@~\Psi)~w$ and $\neut \{\theta'\}(\trec{\R}{\tm}{\IH}~@~\Psi)~w'$ \hfill by def. of $\neut$}
\prf{$\Gamma' \vdash  \{\theta\}\trec{\R}{\tm}{\IH} \equiv  \{\theta'\}\trec{\R}{\tm}{\IH} $ (short for all the branches remain equivalent) \hfill Sem. Subst. Preserve Equiv. (Lemma \ref{lem:semsubst})}
\prf{$\Gamma' \vdash \{\theta\}\Psi \equiv \{\theta'\}\Psi : \tmctx$ \hfill Sem. Subst. Preserve Equiv. (Lemma \ref{lem:semsubst})}
\prf{$\Gamma' \Vdash \{\theta\}(\trec{\R}{\tm}{\IH}~@~\Psi)~w \equiv \{\theta'\}(\trec{\R}{\tm}{\IH}~@~\Psi)~w' :
         \{\theta, \Psi'/\psi,~\{\theta\}t/m\}\tau$ \hfill by $\equiv$ }
%
 \prf{$\Gamma' \vdash  \{\theta\}(\trec{\R}{\tm}{\IH} \rappto \Psi)~t \whnf  \{\theta\}(\trec{\R}{\tm}{\IH}\rappto \Psi)~w
           :  \{\theta, \Psi/\psi,~\{\theta\} t/m\}\tau $
     \hfill by type conversion and $\whnf$ rule}
 \prf{$\Gamma' \vdash  \{\theta\}(\trec{\R}{\tm}{\IH} \rappto \Psi)~t \whnf  \{\theta\}(\trec{\R}{\tm}{\IH}~\rappto \Psi)~w'
           :  \{\theta, \Psi/\psi,~\{\theta\} t/m\}\tau $
     \hfill by type conversion and $\whnf$ rule}
  \prf{$\Gamma' \Vdash \{\theta\}(\trec{\R}{\tm}{\IH}~\rappto\Psi~t) = \{\theta'\}(\trec{\R}{\tm}{\IH}~\rappto \Psi ~t):
             \{\theta, \Psi/\psi,~\{\theta\}t/m\}\tau$
    \hfill by Back. Closed Lemma (twice) (\ref{lem:bclosed}) \\ \mbox{\qquad}\hfill and Symmetry }
 \prf{$\Gamma \models \trec{\R}{\tm}{\IH}\rappto \Psi~t : \IH$ \hfill since $\Gamma', \theta, \theta'$ were arbitrary\\[0.5em]}
}
\end{proof}

\begin{lemma}[Validity of Application]\label{lem:ValidTypeApp}
  If $\Gamma \models t : (y:\ann\tau_1) \arrow \tau_2$ and $\Gamma \models s : \ann\tau_1$
then $\Gamma \models t~s : \{s/y\}\tau_2$.
\end{lemma}
\begin{proof}
By definition relying on semantic type application lemma (Lemma \ref{lem:SemTypeApp}).
\end{proof}
}

 \begin{theorem}[Fundamental Theorem]\quad
   \label{lm:fundtheo}
  \begin{enumerate}
  \item If $\vdash \Gamma$ then $\models \Gamma$.
  \item If $\Gamma ; \Psi \vdash M : A$ then $\Gamma ; \Psi \models M = M : A$.
  \item If $\Gamma ; \Psi \vdash \sigma  : \Phi$ then $\Gamma ; \Psi \models \sigma = \sigma : \Phi$.
  \item If $\Gamma ; \Psi \vdash M \equiv N : A$ then $\Gamma ; \Psi \models M = N : A$.
  \item If $\Gamma ; \Psi \vdash \sigma \equiv \sigma' : \Phi$ then $\Gamma ; \Psi \models \sigma = \sigma' :  \Phi$.
  \item If $\Gamma \vdash t : \tau$ then $\Gamma \models t : \tau$.
  \item If $\Gamma \vdash t \equiv t' : \tau$ then $\Gamma \models t = t' : \tau$.
  \end{enumerate}
\end{theorem}
\begin{proof}
By induction on the first derivation using validity of application,
functions, recursion, and type conversion, Backwards Closed
(\ref{lem:bclosed}), Well-formedness of Semantic
Typing\LONGVERSION{\ (Lemma~\ref{lm:semwf})}, 
Semantic Weakening.
\LONGVERSIONCHECKED{
\\[0.5em]
\fbox{ If $\Gamma \vdash t : \tau$ then $\Gamma \models t : \tau$.}
\\[0.5em]
\paragraph{Case} $\D = \ianc{\Gamma \vdash t : \cbox{\Phi \vdash A} \quad \Gamma ; \Psi \vdash \sigma : \Phi}
                            {\Gamma ; \Psi \vdash \unbox t \sigma : [\sigma]A}{}$\\[0.5em]
\prf{$\Gamma \models t : \cbox{\Phi \vdash A}$ \hfill by IH}
\prf{$\Gamma ; \Psi \models \sigma = \sigma: \Phi$ \hfill by IH}\\[-0.75em]
\prf{Assume $\Gamma' \Vdash \theta = \theta' : \Gamma$}
\prf{$\Gamma' \Vdash \{\theta\}t = \{\theta'\}t : \{\theta\}\cbox{\Phi \vdash A}$ \hfill by def. $\Gamma \models t : \cbox{\Phi \vdash A}$}
\\[-0.75em]
\prf{\emph{Sub-case}: $\Gamma' \vdash \{\theta\}t \whnf \cbox C : \{\theta\}\cbox{\Phi \vdash A}$ and $\Gamma' \vdash \{\theta'\}t \whnf \cbox {C'} : \{\theta\}\cbox{\Phi \vdash A}$ \\
\mbox{\qquad} and $\Gamma' \Vdash C = C' : \{\theta\}(\Phi \vdash A)$}
\prf{$C = \hatctx{\Phi} \vdash M$ and $C' = \hatctx{\Phi} \vdash N$
   \hfill by inversion on $\Gamma' \Vdash C = C' : \{\theta\}(\Phi \vdash A)$}
\prf{$\Gamma' ; \{\theta\}\Psi \Vdash \{\theta\} \sigma = \{\theta'\}\sigma: \{\theta\}\Phi$
   \hfill by $\Gamma ; \Psi \models \sigma : \Phi$  }
\prf{$\Gamma' ; \{\theta\}\Phi \Vdash M = N: \{\theta\}A$
   \hfill by def. of $\Gamma' \Vdash C = C' : \{\theta\}(\Phi \vdash A)$}
\prf{$\Gamma' ; \{\theta\}\Psi \Vdash [\{\theta\}\sigma] M = [\{\theta'\}\sigma]N: \{\theta\}[\sigma]A$
   \hfill by Lemma \ref{lem:semlfeqsub} }
\prf{$\Gamma' ; \{\theta\}\Psi \vdash [\{\theta\}\sigma] M \lfwhnf W : \{\theta\}[\sigma]A$ \hfill by previous line}
\prf{$\Gamma' \vdash \{\theta\}t \whnf \cbox C : \{\theta\}\cbox{\Phi \vdash A}$ \hfill by restating the condition of the case we are in}
\prf{$\Gamma' ; \{\theta\}\Psi \vdash \{\theta\}(\unbox t \sigma) \lfwhnf W : \{\theta\}[\sigma]A$ \hfill by whnf rules}
\prf{$\Gamma' ; \{\theta\}\Psi \vdash [\{\theta'\}\sigma]N \lfwhnf W' : \{\theta\}[\sigma]A$ \hfill by previous line}
\prf{$\Gamma' \vdash \{\theta'\}t \whnf \cbox {C'} : \{\theta\}\cbox{\Phi \vdash A}$ \hfill by restating the condition of the case we are in}
\prf{$\Gamma' ; \{\theta\}\Psi \vdash \{\theta'\}(\unbox t \sigma) \lfwhnf W'  : \{\theta\}[\sigma]A$ \hfill by whnf rules}
\prf{$\Gamma' ; ; \{\theta\}\Psi  \Vdash \{\theta\}(\unbox t \sigma) = \{\theta'\}(\unbox t \sigma) : \{\theta\}[\sigma]A$
   \hfill by Backwards Closed Lemma \ref{lem:bclosed} (twice) \\
\mbox{$\quad$}\hfill and symmetry.}
\prf{$\Gamma ; \Psi \models \unbox t \sigma : [\sigma]A$
  \hfill by abstraction, since $\Gamma', \theta, \theta'$ were arbitrary}
\\
\prf{\emph{Sub-case}: $\Gamma' \vdash \{\theta\}t \whnf w : \{\theta\}\cbox{\Phi \vdash A}$ and
                      $\Gamma' \vdash \{\theta'\}t \whnf w': \{\theta\}\cbox{\Phi \vdash A}$ \\
\mbox{\qquad} and $\neut w$ and $\neut w'$ and $\Gamma \vdash w \equiv w' : \{\theta\}\cbox{\Phi \vdash A}$}
\\[-0.75em]
\prf{$\Gamma' \vdash \theta(\unbox t \sigma) \whnf \unbox w {\{\theta\}\sigma}$ \hfill by whnf rules}
\prf{$\Gamma' \vdash \theta'(\unbox t \sigma) \whnf \unbox {w'} {\{\theta\}\sigma}$ \hfill by whnf rules}
\prf{$\Gamma' \{\theta\}\Psi \Vdash \{\theta\}\sigma = \{\theta'\}\sigma : \{\theta\}\Phi$ \hfill by def. $\Gamma ; \Psi \models \sigma = \sigma : \Phi$}
\prf{$\Gamma' ; \{\theta\}\Psi \vdash \{\theta\}\sigma \equiv \{\theta'\}\sigma : \{\theta\}\Phi$ \hfill by Well-formedness Lemma~\ref{lm:semwf}}
\prf{$\Gamma' ; \{\theta\}\Psi \vdash \{\theta\}(\unbox t \sigma) \equiv \{\theta'\}(\unbox t \sigma) : \{\theta\}[\sigma]A$ \hfill by $\equiv$ rules}
\prf{$\Gamma' \Vdash \{\theta\}\Phi = \{\theta\}\Phi : \ctx$ \hfill Reflexivity}
\prf{$\typeof (\Gamma' \vdash w) = \cbox{\{\theta\}\Phi \vdash \tm}$ \hfill since $\Gamma' \vdash w : \{\theta\}\cbox{\Phi \vdash A}$ }
\prf{$\typeof (\Gamma' \vdash w') = \cbox{\{\theta\}\Phi \vdash \tm}$ \hfill since $\Gamma' \vdash w' : \{\theta\}\cbox{\Phi \vdash A}$ }
\prf{$\Gamma' ; \{\theta\}\Psi \Vdash \{\theta\}(\unbox t \sigma) = \{\theta'\}(\unbox t \sigma) : \{\theta\}[\sigma]A$
    \hfill by semantic def. }
\prf{$\Gamma ; \Psi \models \unbox t \sigma : [\sigma]A$
   \hfill by abstraction, since $\Gamma', \theta, \theta'$ were
   arbitrary}
}
\LONGVERSIONCHECKED{\\
\pcase{$\D = \ianc {y:\ann\tau \in \Gamma}{\Gamma \vdash y : \ann\tau}{}$}
\prf{$\Ca: ~~\vdash \Gamma$ and $\Ca < \D$ \hfill by Lemma \ref{lm:ctxwf}}
\prf{Assume $\Gamma', \theta, \theta'.~\Gamma' \Vdash \theta = \theta' : \Gamma$}
\prf{$\Gamma' \Vdash t = s : \{\theta_i\}\ann\tau$ \hfill by sem. def. of $\Gamma' \Vdash \theta = \theta' : \Gamma$}
\prf{$\Gamma' \Vdash \{\theta\}y = \{\theta'\}y : \{\theta\}\ann\tau$ \hfill subst. def. where $\{\theta\}y = t$ and $\{\theta'\}y = s$}
\prf{$\Gamma \models y = y : \ann\tau$ \hfill by def. of validity}
\\
\pcase{ $\D = \ibnc
{\Gamma \vdash t : \tau'}{\Gamma \vdash \tau' \equiv \tau : u}
{\Gamma \vdash t : \tau}{}$}
\prf{$\Gamma \models t : \tau'$ \hfill by IH }
\prf{$\Gamma \models \tau = \tau' : u$ \hfill by IH }
\prf{$\Gamma \models t : \tau$ \hfill by Lemma \ref{lem:validtypconv}}
\\
\pcase{$\D = \ianc{\Gamma \vdash C : T}
                           {\Gamma \vdash \cbox C : \cbox T}{}$
}
\prf{$\Ca: ~~\vdash \Gamma$ and $\Ca < \D$ \hfill by Lemma \ref{lm:ctxwf}}
\prf{$\models \Gamma$ \hfill  by IH}
\prf{Let $C = \hatctx{\Psi} \vdash M$ and $T = \Psi \vdash A$.}
\prf{$\Gamma ; \Psi \vdash M : A$ \hfill by inversion on $\Gamma \vdash C : T$}
\prf{$\Gamma ; \Psi \models M : A$ \hfill by IH}
\prf{Assume $\Gamma',~\theta,~\theta'.~ \Gamma' \Vdash \theta = \theta'$}
\prf{$\Gamma ; \{\theta\}\Psi \Vdash \{\theta\}M = \{\theta'\}M : \{\theta\}A$ \hfill using $\Gamma ; \Psi \models M : A$}
\prf{$\Gamma \Vdash \{\theta\}(\hatctx{\Psi} \vdash M) = \{\theta'\}(\hatctx{\Psi} \vdash M) : \{\theta\}T$ \hfill by sem. def.}
\prf{$\Gamma \Vdash \{\theta\}\cbox C = \{\theta'\}\cbox C: \{\theta\}\cbox T$ \hfill by previous line}
\prf{$\Gamma \models \cbox C = \cbox C: \cbox T$ \hfill by abstraction, since $\Gamma', ~\theta,~\theta'$ were arbitrary}
\prf{$\Gamma \models \cbox C : \cbox T$ \hfill by def. of validity}
\\
\pcase{$\D = \ibnc
{\Gamma \vdash t : (y:\ann\tau_1) \arrow \tau_2}{ \Gamma \vdash s : \ann\tau_1}
{\Gamma \vdash t~s : \{s/y\}\tau_2}{}
$}
\prf{$\Gamma \models s : \ann\tau_1$ \hfill by IH }
\prf{$\Gamma \models t : (y:\ann\tau_1) \arrow \tau_2$ \hfill by IH}
\prf{$\Gamma \models t~s : \{s/y\}\tau_2$ \hfill by Lemma \ref{lem:ValidTypeApp}}
\pcase{$\D = \ianc{\Gamma, y:\ann\tau_1 \vdash t : \tau_2 }
                            {\Gamma \vdash \tmfn y t : (y:\ann\tau_1) \arrow \tau_2}{}$}
\\[0.2em]
\prf{$\Ca: ~~\vdash \Gamma, y{:}\ann\tau_1$ and $\Ca < \D$ \hfill by Lemma \ref{lm:ctxwf}}
\prf{$\models \Gamma, y{:}\ann\tau_1$ \hfill  by IH }
\prf{$\models \Gamma$ \hfill by def. of validity }
\prf{$\Gamma, y:\ann\tau_1 \models t : \tau_2$ \hfill by IH}
\prf{$\Gamma \models  (\tmfn y t) : (y:\ann\tau_1) \arrow \tau_2$ \hfill Lemma \ref{lem:ValidFun}}

\pcase{$\D = \ianc{\vdash \Gamma}{\Gamma \vdash u_1: u_2}{(u_1, u_2) \in \Ax}$}
\prf{$\Gamma \vdash u_1 \whnf u_1 : u_2$ \hfill by rules and typing assumption}
\prf{$u_1 \leq u_2$ \hfill by $(u_1, u_2) \in \Ax$}
\\
\pcase{$\D = \ibnc {\Gamma \vdash \ann\tau_1 : u_1}
       {\Gamma, y{:}\ann\tau_1 \vdash \tau_2 : u_2}
       {\Gamma \vdash (y:\ann\tau_1) \arrow \tau_2 : u_3}{(u_1,~u_2,~u_3) \in \Ru}$}
\prf{$\Gamma \models \ann\tau_1 : u_1$ \hfill by IH}
\prf{$\forall \Gamma'.~\Gamma' \Vdash \theta = \theta' : \Gamma \Longrightarrow \Gamma' \Vdash \{\theta\}\ann\tau_1 = \{\theta'\}\ann\tau_1 : u_1$ \hfill by sem. def.}
\prf{$\Gamma, y{:}\ann\tau_1 \models \tau_2 : u_2$ \hfill by IH}
\prf{$\forall ~\Gamma' \Vdash \theta = \theta' : \Gamma, y{:}\ann\tau_1 \Longrightarrow \Gamma' \Vdash  \{\theta\}\tau_2 = \{\theta'\}\tau_2 : u_1$ \hfill by sem. def.}
\prf{Assume $\Gamma' \Vdash \theta = \theta' : \Gamma$}
\prf{$\Gamma' \Vdash \{\theta\}\ann\tau_1 = \{\theta'\}\ann\tau_1 : u_1$ \hfill by previous lines}
\prf{Assume $\Gamma'' \leq_\rho \Gamma'$}
\prf{$\Gamma'' \Vdash \{\rho\}\{\theta\}\ann\tau_1 = \{\rho\}\{\theta'\}\ann\tau_1 : u_1$ \hfill sem. weakening for computations lemma \ref{lem:compsemweak}}
\prf{$\forall \Gamma''.~\Gamma'' \leq_\rho \Gamma'. \Gamma'' \Vdash \{\rho\}\{\theta\}\ann\tau_1 = \{\rho\}\{\theta'\}\ann\tau_1 : u_1$\hfill by previous lines}
\prf{Assume $\Gamma'' \leq_\rho \Gamma'$ and $\Gamma'' \Vdash t = t' : \{\{\rho\}\theta\}\ann\tau_1$}
\prf{$\Gamma'' \Vdash \{\rho\}\theta = \{\rho\}\theta' : \Gamma$ \hfill by sem. weakening lemma \ref{lem:weakcsub}}
\prf{$\Gamma '' \Vdash \{\rho\}\theta, t/y = \rho\{\theta'\}, t'/y : \Gamma, y{:}\ann\tau_1$ \hfill by sem. def.}
\prf{$\Gamma'' \Vdash \{\{\rho\}\theta, t/y\}\tau_2 = \{\rho\{\theta'\}, t'/y\}\tau_2 : u_2$ \hfill by $\Gamma, y{:}\ann\tau_1 \models \tau_2 : u_2$ \\
\mbox{$\quad$}\hfill choosing $\theta =  \{\rho\}\theta, t/y$ and $\theta' = \rho\{\theta'\}, t'/y$ }
\prf{$\Gamma'' \Vdash \{\rho, t/y\}\{\theta, y/y\}\tau_2 = \{\rho, t'/y\}\{\theta', y/y\}\tau_2 : u_2$\hfill by subst. def.}
\prf{$\forall \Gamma''.~\Gamma'' \leq_\rho \Gamma'. ~\Gamma'' \Vdash t = t' : \{\{\rho\}\theta\}\ann\tau_1 \Longrightarrow \Gamma'' \Vdash \{\rho, t/y\}\{\theta, y/y\}\tau_2 = \{\rho, t'/y\}\{\theta', y/y\}\tau_2 : u_2$ \hfill \\
\mbox{$\quad$}\hfill by previous lines}
\prf{$\Gamma' \Vdash \{\theta\}((y:\ann\tau_1) \arrow \tau_2 : u_3) = \{\theta'\}((y:\ann\tau_1) \arrow \tau_2 : u_2)$ \hfill by sem. def.}
\prf{$\Gamma \models ((y:\ann\tau_1) \arrow \tau_2 : u_3) = ((y:\ann\tau_1) \arrow \tau_2 : u_3)$ \hfill by sem. def.}
\prf{$\Gamma \models (y:\ann\tau_1) \arrow \tau_2 : u_3$ \hfill by def. of validity }
\\
\pcase{$\D = \ianc{\Gamma \vdash T}{\Gamma \vdash \cbox{T} : u}{}$}
\prf{$\models \Gamma$ \hfill }
\prf{Assume $\Gamma' \Vdash \theta = \theta' : \Gamma $}
\prf{$\Gamma \models_\LF T = T$ \hfill by IH}
\prf{$\Gamma' \Vdash_\LF \{\theta\} T = \{\theta'\}T$ \hfill since $\Gamma \models_\LF T = T$ }
\prf{$\Gamma' \vdash \cbox{\{\theta\}T}\whnf \cbox{\{\theta\}T} : u$ \hfill since $\norm \cbox{\{\theta\}T}$}
\prf{$\Gamma' \Vdash \{\theta\}\cbox{T} = \{\theta'\}\cbox{T} : u$ \hfill by sem. def.}
\prf{$\Gamma \models \cbox{T}  = \cbox{T} : \univ k$ \hfill since $\Gamma', \theta, \theta'$ are arbitrary}
\prf{$\Gamma \models \cbox{T} : \univ k$ \hfill by def. of validity}
}
\LONGVERSIONCHECKED{
\\[1em]
\pcase{$\D =
\ianc
{\Gamma \vdash t : \cbox{\Psi \vdash \tm} \quad
 \Gamma \vdash \IH : u \quad
 \Gamma \vdash b_v : \IH \quad
 \Gamma \vdash b_{\mathsf{app}} : \IH \quad
 \Gamma \vdash b_{\mathsf{lam}} : \IH}
{\Gamma \vdash \tmrec {\IH} {b_v} {b_{\mathsf{app}}} {b_{\tlam}} \rappto \Psi~t : \{{\Psi}/\psi,~t/y\}\tau}{}
$\\[0.25em]
\mbox{where}~$\IH = (\psi : {\tmctx}) \arrow (y:\cbox{\psi \vdash \tm}) \arrow \tau$}
\\[-1em]
\prf{$\Gamma \models \IH : u $ \hfill by IH}
\prf{$\Gamma \models t : \cbox{\Psi \vdash \tm}$ \hfill by IH}
\prf{$ \Gamma, \psi:\tmctx, p:\cbox{~\psi \vdash_\# \tm}  \vdash  t_v : \{p / y\}\tau$ \hfill by typing inversion}
\prf{$\Gamma,  \psi:\tmctx, m:\cbox{~\psi \vdash \tm}, n:\cbox{~\psi \vdash \tm}
         f_m: \{m/y\}\tau, f_n: \{n/y\}\tau$}
\prf{\mbox{$\qquad$}\hfill$\qquad \vdash  t_{\mathsf{app}} : \{\cbox{~\psi \vdash\tapp~\unbox{m}{\id}~\unbox{n}{\id}}/y\}\tau$$\qquad$ by typing inversion}
\prf{$\Gamma, \phi:{\tmctx},  m:\cbox{~\unboxc{\phi}, x:\tm \vdash \tm},
          f_m:\{\cbox{~\unboxc{\phi}, x:\tm}/\psi, m /y \} \tau$}
\prf{\mbox{$\qquad$}$\quad$\hfill $\vdash t_{\tlam} : \{\phi/\psi, \cbox{~\unboxc{\phi} \vdash \tlam~\lambda x.\unbox{m}{\id}~} / y\}\tau$ $\qquad$ by typing inversion}
\prf{$\Ca: ~~\vdash \Gamma, \psi:\tmctx, p:\cbox{~\unboxc{\psi} \vdash_\# \tm}$ and $\Ca < \D$ \hfill by Lemma \ref{lm:ctxwf}}
\prf{$\models \Gamma$ \hfill by def. of validity }
\prf{$ \Gamma, \psi:\tmctx, p:\cbox{~\unboxc{\psi} \vdash_\# \tm}  \models  t_v : \{p / y\}\tau$ \hfill by IH}
\prf{$\Gamma,  \psi:\tmctx, m:\cbox{~\unboxc{\psi} \vdash \tm}, n:\cbox{~\unboxc{\psi} \vdash \tm}
         f_m: \{m/y\}\tau, f_n: \{n/y\}\tau$}
\prf{\mbox{$\qquad$}\hfill$\qquad \models  t_{\mathsf{app}} : \{\cbox{~\unboxc{\psi} \vdash\tapp~\unbox{m}{\id}~\unbox{n}{\id}}/y\}\tau\qquad$ by IH}
\prf{$\Gamma, \phi:\tmctx,  m:\cbox{~\unboxc{\phi}, x:\tm \vdash \tm},
          f_m:\{\cbox{~\unboxc{\phi}, x:\tm}/\psi, m /y \} \tau$}
\prf{\mbox{$\quad$}\hfill $\models t_{\tlam} : \{\phi/\psi, \cbox{~\unboxc{\phi} \vdash \tlam~\lambda x.\unbox{m}{\id}~} / y\}\tau\qquad$  by IH}
\prf{$\Gamma \models \tmrec {\IH} {b_v} {b_{\mathsf{app}}} {b_{\tlam}} \rappto \Psi~t :   \IH$
 \hfill by Validity of Recursion Lemma \ref{lem:ValidRec}}
\\
\fbox{ If $\Gamma \vdash t \equiv t' : \tau$ then $\Gamma \models t =  t' : \tau$.}
\\[1em]
\pcase{
$\D = \ianc {\Gamma \vdash \tmfn y t : (y{:}\ann\tau_1) \arrow \tau_2
     \qquad \Gamma \vdash s : \ann\tau_1}
           {\Gamma  \vdash (\tmfn y t)~s \equiv \{s/y\}t : \{s/y\}\tau_2}{}$
                }
\prf{To Show: $\Gamma \models (\tmfn y t)~s = \{s/y\}t : \{s/y\}\tau_2
  $}
\prf{$\Ca: ~~\vdash \Gamma$ and $\Ca < \D$ \hfill by Lemma \ref{lm:ctxwf}}
\prf{$\models \Gamma$ \hfill by IH}
\prf{Assume $\Gamma' \Vdash \theta = \theta' : \Gamma$}
\prf{$\Gamma \models \tmfn y t : (y : \ann\tau_1) \arrow \tau_2$ \hfill by IH}
\prf{$\Gamma \models s : \ann\tau_1$ \hfill by IH}
\prf{$\Gamma' \Vdash \{\theta\}s = \{\theta'\}s : \{\theta\}\ann\tau_1$ \hfill
   by $\Gamma \models s : \ann\tau_1$}
\prf{$\Gamma \Vdash \{\theta\}(\tmfn y t) = \{\theta'\}(\tmfn y t) : \{\theta\}((y : \ann\tau_1) \arrow \tau_2)$ \hfill by previous lines (def. of $\models$)}
\prf{$\Gamma \vdash \{\theta\}(\tmfn y t) \whnf \{\theta\}(\tmfn y t) : \{\theta\}((y : \ann\tau_1) \arrow \tau_2)$ \hfill by sem. equ. def.}
\prf{$\Gamma \vdash \{\theta'\}(\tmfn y t) \whnf \{\theta'\}(\tmfn y t) : \{\theta\}((y : \ann\tau_1) \arrow \tau_2)$ \hfill by sem. equ. def.}
\prf{$\Gamma \Vdash \{\theta\}(\tmfn y t)~\{\theta\}s = \{\theta'\}(\tmfn y t)~\{\theta'\}s' : \{\theta, \{\theta\}s/y\}\tau_2$ \hfill by sem. equ. def.}
\prf{$\Gamma \vdash \{\theta\}((\tmfn y t)~s) \whnf w :  \{\theta, s/y\}\tau_2$ \hfill by sem. equ. def}
\prf{$\Gamma \vdash \{\theta, \{\theta\}s/y\}t \whnf w :  \{\theta, s/y\}\tau_2$ \hfill by inversion on $\whnf$}
\prf{$\Gamma \vdash \{\theta'\}((\tmfn y t)~s) \whnf w' :  \{\theta, s/y\}\tau_2$ \hfill by sem. equ. def}
\prf{$\Gamma \vdash \{\theta', \{\theta'\}s/y\}t \whnf w' :  \{\theta, s/y\}\tau_2$ \hfill by inversion on $\whnf$}
\prf{$\Gamma \Vdash (\tmfn y t)~s = \{s/y\}t : \{s/y\}\tau_2$ \hfill by Backwards Closure (Lemma \ref{lem:bclosed})}
\\[0.5em]
\pcase{
$\D =  \ianc{\Gamma \vdash t : \cbox{\Psi \vdash A}}
            {\Gamma \vdash \cbox{\hatctx \Psi \vdash \unbox{t}{\wk{\hatctx\Psi}}} \equiv t : \cbox{\Psi \vdash A}}{}$
}
\prf{To Show: $\Gamma \models \cbox{\hatctx \Psi \vdash \unbox{t}{\wk{\hatctx\Psi}}} = t : \cbox{\Psi \vdash A}$}
\prf{Assume $\Gamma' \Vdash \theta = \theta' : \Gamma$}
\prf{$\Gamma \models t : \cbox{\Psi \vdash A}$ \hfill by IH}
\prf{$\Gamma' \Vdash \{\theta\} t = \{\theta'\}t : \{\theta\} \cbox{\Psi \vdash A}$ \hfill by $\Gamma \models t : \cbox{\Psi \vdash A}$}
\prf{$\Gamma' \vdash \{\theta\} t \whnf w : \{\theta\} \cbox{\Psi \vdash A}$ \hfill by sem. def.}
\prf{$\Gamma' \vdash \{\theta'\} t \whnf w' : \{\theta\} \cbox{\Psi \vdash A}$ \hfill by sem. def.}
\prf{$\Gamma' ; \{\theta\}\Psi \Vdash \unbox w \id = \unbox{w'}\id : \{\theta\}A$ \hfill by sem. def.}
\prf{We note that $w$ is either $n$ where $\neut n$ or $\cbox{\hatctx\Psi \vdash M}$ }
\prf{\emph{Sub-case} $w$ is neutral, i.e. $\neut w$}
\prf{$\Gamma' ; \{\theta\}\Psi \vdash \unbox {\{\theta\}t} {\wk{\hatctx\Psi}} \lfwhnf \unbox w {\id} : \{\theta\}A$
  \hfill since $\Gamma' \vdash \{\theta\} t \whnf w : \{\theta\} \cbox{\Psi \vdash A}$ }
\prf{$\Gamma' ; \{\theta\}\Psi \vdash \unbox{\{\theta\}\cbox{\hatctx\Psi \vdash   \unbox{t}{\wk{\hatctx\Psi}} }}\id \lfwhnf \unbox w {\id} : \{\theta\}A$  \hfill by $\lfwhnf$ }
\prf{$\Gamma ; \{\theta\}\Psi \Vdash \unbox{\{\theta\}\cbox{\hatctx\Psi \vdash   \unbox{t}{\wk{\hatctx\Psi}} }}\id = \unbox{w'}\id : \{\theta\}A$
\hfill \\\mbox{\hspace{2cm}}\hfill using $\Gamma' ; \{\theta\}\Psi \Vdash \unbox w \id = \unbox{w'}\id : \{\theta\}A$ and Backwards Closure (Lemma \ref{lem:bclosed})}
\prf{$\Gamma ; \{\theta\}\Psi \vdash \{\theta\}\cbox{\hatctx \Psi \vdash \unbox{t}{\wk{\hatctx\Psi}}} \whnf \{\theta\}\cbox{\hatctx \Psi \vdash \unbox{t}{\wk{\hatctx\Psi}}} : \{\theta\}\cbox{\Psi \vdash A}$\hfill since $\norm \{\theta\}\cbox{\hatctx \Psi \vdash \unbox{t}{\wk{\hatctx\Psi}}}$}
\prf{$\Gamma' \Vdash \{\theta\}\cbox{\hatctx \Psi \vdash \unbox{t}{\wk{\hatctx\Psi}}} = \{\theta'\} t : \{\theta\}\cbox{\Psi \vdash A}$ \hfill by sem. def.}
\\
\prf{\emph{Sub-case} $w = \cbox{\hatctx \Psi \vdash M}$}
\prf{$\Gamma' ; \{\theta\}\Psi \vdash \unbox{\cbox{\hatctx\Psi \vdash M}}{\id} \lfwhnf N : \{\theta\}A$ \hfill
 by $\Gamma' ; \{\theta\}\Psi \Vdash \unbox w \id = \unbox{w'}\id : \{\theta\}A$ }
\prf{$\Gamma'; \{\theta\}\Psi \vdash  M \lfwhnf N : \{\theta\}A$ \hfill by $\lfwhnf$}
\prf{$\Gamma' ; \{\theta\}\Psi \vdash \unbox{\{\theta\}\cbox{\hatctx\Psi \vdash   \unbox{t}{\wk{\hatctx\Psi}} }}\id \lfwhnf N : \{\theta\}A $}
\prf{$\Gamma' ; \{\theta\}\Psi \Vdash N = \unbox{w'}\id : \{\theta\}A$ \hfill by sem. equ. } 
\prf{$\Gamma' ; \{\theta\}\Psi \Vdash \unbox{\{\theta\}\cbox{\hatctx\Psi \vdash   \unbox{t}{\wk{\hatctx\Psi}} }}\id = \unbox{w'}\id: \{\theta\}A$ \hfill by Backwards Closure (Lemma  \ref{lem:bclosed})}
\prf{$\Gamma' \Vdash \{\theta\}\cbox{\hatctx \Psi \vdash \unbox{t}{\wk{\hatctx\Psi}}} = \{\theta'\} t : \{\theta\}\cbox{\Psi \vdash A}$ \hfill by sem. def.}
}
\end{proof}


\begin{theorem}[Normalization and Subject Reduction]\label{lm:norm}
If $\Gamma \vdash t : \tau$ then $t \whnf w$ and $\Gamma \vdash t \equiv w : \tau$
\end{theorem}
\begin{proof}
By the Fundamental theorem (Lemma \ref{lm:fundtheo}), we have $\Gamma \Vdash t = t : \tau$ (choosing the identity substitution for $\theta$ and $\theta'$).

This includes a definition $t \whnf w$. Since $w$ is in whnf
(i.e. $\norm w$), we have $w \whnf w$. Therefore, we can easily show
that also $\Gamma \Vdash t = w : \tau$.  By
well-formedness\LONGVERSION{\ (Lemma \ref{lm:semwf})}, we also have that
$\Gamma \vdash t \equiv w : \tau$ and more specifically,
$\Gamma \vdash w : \tau$.
\end{proof}

Using the fundamental lemma, we can also show function type
injectivity, which is the basis for implementing a type checker.
\LONGVERSION{
Further, type checking relies on the following inversion lemmas:
\begin{lemma}[Inversion]
  \begin{enumerate}
  \item If $\Gamma \vdash x : \ann\tau$ then $x:\ann\tau' \in \Gamma$ for some $\ann\tau'$ and $\Gamma \vdash \ann\tau \equiv \ann\tau' : u$.
\item If $\Gamma \vdash \tmfn y t : \tau$ then $\Gamma, y:\ann\tau_1 \vdash t : \tau_2$ for some $\ann\tau_1$, $\tau_2$
  and If $\Gamma \vdash \tau \equiv (y : \ann\tau_1) \arrow \tau_2 : u$.
\item If $\Gamma \vdash t~s : \tau$ then
     $\Gamma \vdash t : (y : \ann\tau_1) \arrow \tau_2$ and
     $\Gamma \vdash s : \ann\tau_1$ for some $\ann\tau_1$ and $\tau_2$ and
     $\Gamma \vdash \tau \equiv \{s/y\}\tau_2 : u$.
\item If $\Gamma \vdash \cbox C : \tau$ then
     $\Gamma \vdash \cbox C : \cbox T$ for some contextual type $T$ and
     $\Gamma \vdash \tau \equiv \cbox T : u$.
\item If $\Gamma \vdash \titer {(b_v \mid b_\tapp \mid b_\tlam)}{}\IH~\Psi~t : \tau$ where
      $\IH = (\psi : \tmctx) \arrow (y : \cbox{\Psi \vdash \tm}) \arrow \tau$
     then
   $\Gamma \vdash t : \cbox{\Psi \vdash \tm}$ and
   $\Gamma \vdash \IH : u$ and
   $\Gamma \vdash b_v : \IH$ and $\Gamma \vdash b_\tapp : \IH$ and $\Gamma \vdash b_\tlam : \IH$ and
   $\Gamma \vdash \tau \equiv \{\Psi/\psi, t/y\}\tau : u$.
\item If $\Gamma \vdash u_1 : \tau$ then there some $u_2$ s.t. $(u_1, u_2) \in \Ax$ and $\Gamma \vdash \tau \equiv u_2 : u$.
\item If $\Gamma \vdash (y : \ann\tau_1) \arrow \tau_2 : \tau$ then there is some $u_1$, $u_2$, and $u_3$ s.t.
  $(u_1, u_2, u_3) \in \Ru$ and $\Gamma \vdash \ann\tau_1 : u_1$, $\Gamma, y:\ann\tau_1 \vdash \tau_2 : u_2$ and
$\Gamma \vdash u_3 \equiv \tau : u$.
  \end{enumerate}
\end{lemma}
\begin{proof}
By induction on the typing derivation.
\end{proof}
}

\begin{lemma}[Injectivity of Function Type]\label{lm:funinj}\quad\\
If $\Gamma \vdash (y : \ann\tau_1) \arrow \tau_2 \equiv (y : \ann\tau'_1) \arrow \tau'_2 :u$
then
$\Gamma\! \vdash \ann\tau_1 \equiv \ann\tau'_1 \!:\! u_1$ and
$\Gamma, y:\ann\tau_1 \vdash \tau_2 \equiv \tau_2 : u_2$ and $(u_1, u_2, u) \in \Ru$.
\end{lemma}
\begin{proof}
By the Fundamental theorem (Lemma \ref{lm:fundtheo}) we have
$\Gamma \Vdash (y : \ann\tau_1) \arrow \tau_2 \equiv (y : \ann\tau'_1) \arrow \tau'_2 :u$
(choosing the identity substitution for $\theta$ and $\theta'$).
By the sem. equality def., we have $\Gamma \Vdash \ann\tau_1 = \ann\tau'_1 : u_1$
and $\Gamma, y:\ann\tau_1 \Vdash \tau_2 = \tau'_2 : u_2$ and
$(u_1, u_2, u) \in \Ru$.
By well-formedness of semantic typing\LONGVERSION{\ (Lemma \ref{lm:semwf})}, we have
$\Gamma \vdash \ann\tau_1 \equiv \ann\tau'_1 : u_1$ and
and $\Gamma, y:\ann\tau_1 \vdash \tau_2 \equiv \tau'_2 : u_2$.
\end{proof}


Last but not least, the fundamental lemma allows us to show that not every type is inhabited and thus \cocon can be used as a logic. To establish this stronger notion of consistency, we first prove that we can discriminate type constructors.

\begin{lemma}[Type Constructor Discrimination]\label{lm:typdisc}
Neutral types, sorts, and function types can be discriminated.
\LONGVERSION{\begin{enumerate}
\item If $\Gamma \vdash u_1 \equiv u_2 : u_3$  then $u_1 = u_2$ (they are the same).
\item $\Gamma \vdash x~\vec t \neq u : u'$.
\item $\Gamma \vdash x~\vec t  \neq (y : \ann\tau_1) \arrow \tau_2 $.
\item $\Gamma \vdash x~\vec t \neq \cbox T$.
\item $\Gamma \vdash u \neq (y : \ann\tau_1) \arrow \tau_2$.
\end{enumerate} }
\end{lemma}
\begin{proof}
\LONGVERSION{Proof by contradiction.}
To show for example that $\Gamma \vdash x~\vec t  \neq (y : \ann\tau_1) \arrow \tau_2 $, we
assume $\Gamma \vdash x~\vec t \equiv (y : \ann\tau_1) \arrow \tau_2 : u$. By the fundamental lemma (Lemma \ref{lm:fundtheo}), we have
$\Gamma \Vdash  x~\vec t \equiv (y : \ann\tau_1) \arrow \tau_2 : u$ (choosing the identity substitution for $\theta$ and $\theta'$); but this is impossible given the semantic equality definition. 
\end{proof}

\begin{theorem}[Consistency]\label{lm:consistency}
$x:u_0 \not\vdash t : x$.
\end{theorem}
\begin{proof}
Assume $\Gamma \vdash t : x$ where $\Gamma = (x{:}u_0)$.  By subject
reduction (Lemma \ref{lm:norm}), there is some $w$ such that $t \whnf w$
and $\Gamma \vdash t \equiv w : x$ and in particular, we must have
$\Gamma \vdash w : x$. As $x$ is neutral, it cannot be equal to $u$,
$(y : \ann\tau_1) \arrow \tau_2$, or $\cbox T$ (Lemma
\ref{lm:typdisc}). Thus $w$ can also not be a sort, function, or
contextual object. Hence, $w$ can only be neutral, i.e. given the
assumption $x:u_0$, the term $w$ must be $x$. This implies that
$\Gamma \vdash x : x$ and implies $\Gamma \vdash x \equiv u_0 : u_0$
by inversion on typing. But this is impossible by Lemma
\ref{lm:typdisc}.
\end{proof}




An extended version with the full technical development is available
at \citet{cocon:arxiv19}.

\section{Conclusion}
\cocon is a first step towards integrating LF methodology into
Martin-L{\"o}f style dependent type theories and and bridges the
longstanding gap between these two worlds. We have established
normalization and consistency. The next immediate step is
to derive an equivalence algorithm based on weak head reduction and
show its completeness. We expect that this will follow a similar
Kripke-style logical relation as the one we described. This would
allow us to justify that type checking \cocon programs is decidable.

It should be possible to implement \cocon as an extension to \beluga
-- from a syntactic point of view, it would be a small change, however
in practice this requires revisiting type reconstruction, type
checking, unification, and conversion. It also  seems possible to
extend existing implementation of Agda, however this might be more
work, as in this case one needs to implement the LF  infrastructure.

\section*{Acknowledgments}
Brigitte Pientka was supported by NSERC (Natural Science and
Engineering Research Council) Grant 206263.
David Thibodeau was supported by NSERC's Alexander Graham Bell Canada
Graduate Scholarships -- Doctoral Program (CGS D).
Andreas Abel was supported by the Swedish Research Council through VR
Grant 2014-04864 \emph{Termination Certificates for Dependently-Typed
  Programs and Proofs via Refinement Types} and the EU COST Action CA
15123 \emph{EUTYPES: Types for Programming and Verification}.
Francisco Ferreira wants to acknowledge the support received from
EPSRC grants EP/K034413/1 and EP/K011715/1.
%

%




%
\bibliographystyle{apalike}

\end{document}

%% file: macros.tex
\usepackage{mathpartir}

\usepackage{todonotes}

\newcommand{\highlight}[1]{{\color{green}{#1}}}

\newcommand{\beluga}{\textsc{Beluga}\xspace}
\newcommand{\cocon}{\textsc{Cocon}\xspace}

\newtheorem{definition}{Definition}[section]

\newcommand{\hatctx}[1]{\hat{#1}}
\newcommand\floor[1]{\lfloor#1\rfloor}
\newcommand\ceil[1]{\lceil#1\rceil}
\newcommand{\cbox}[1]{\ceil {#1}}
\newcommand{\unbox}[2]{\floor {#1}_{#2}}
\newcommand{\unboxc}[1]{#1}

\newcommand{\csubclo}[2]{\{#1 / #2\}_{clo}}
\newcommand{\lfss}[2]{[#1 / #2]}
\newcommand{\lfs}[2]{[#1 / \hatctx #2]}

\definecolor{DimGrey}{rgb}{0.8,0.8,0.8}

\lstdefinelanguage{ContextualML}
{
  morekeywords={and, block, case, of, mlam, fn, impossible, let, in, schema,
    some, rec, type, ctype, prop, stratified, inductive, coinductive, LF, if, then,
    else, total, with},
  keepspaces=true,
  sensitive,
  morecomment=[l]{\%},
  morecomment=[n]{\%\{}{\}\%},
  morestring=[b]"
}[keywords,comments,strings]

\newcommand{\tlam}{\mathsf{lam}}
\newcommand{\tcopy}{\mathsf{copy}}

\newcommand{\tapp}{\mathsf{app}}

\newcommand{\pcase}[1]{\emph{Case}. #1 \\[0.75em]}

\newcommand{\prf}[1]{#1\\[0.15em]}

\newcommand{\LF}{\mathsf{LF}}

\newcommand{\der}{\vdash}
\newcommand{\sem}{\Vdash}
\newcommand{\semlf}{\Vdash_{\mathsf{LF}}}

\newcommand{\D}{{\mathcal{D}}}
\newcommand{\Ca}{{\mathcal{C}}}

\newcommand{\M}{{\mathcal{M}}}

\newcommand{\IH}{{\mathcal{I}}}
\newcommand{\J}{{\mathcal{J}}}

\newcommand{\JLF}{\J_{\mathsf{LF}}}

\newcommand{\So}{\mathcal{S}}
\newcommand{\Ax}{\mathcal{A}}
\newcommand{\Ru}{\mathcal{R}}

\newcommand{\rappto}{~}
\newcommand{\cappto}{\ll}
\newcommand{\typeof}{\mathsf{typeof}}

\newcommand{\ann}[1]{ \breve{#1}}

\newcommand{\pos}{\mathrm{lookup}~}
\newcommand{\posno}{\mathrm{lookup}}
\newcommand{\trunc}{\mathrm{trunc}}

\newcommand{\clam}{\ensuremath{\mathtt{lam}}~\xspace}
\newcommand{\capp}{\ensuremath{\mathtt{app}}~\xspace}
\newcommand{\cletv}{\ensuremath{\mathtt{letv}}~\xspace}

\newcommand{\bnfas}{\;\mathrel{::=}\;}
\newcommand{\bnfalt}{\, \mid \,}
\newcommand{\wvec}[1]{\overrightarrow{#1}}

\newcommand{\Pityp}[3]{\Pi #1{:}#2.#3}
\newcommand{\const}[1]{\ensuremath{\mathsf{#1}}}
\newcommand{\wwk}[1]{\textbf{wk}}
\newcommand{\wk}[1]{\ensuremath{\mathsf{wk}_{#1}}}
\newcommand{\wkempty}{\wk{(\cdot)}}
\newcommand{\sclo}[2]{\wk{#1} \circ #2}
\newcommand{\id}{\ensuremath{\mathsf{id}}}
\renewcommand{\arrow}{\Rightarrow}
\newcommand{\univ}[1]{\ensuremath{\mathsf{U}_{#1}}}

\newcommand{\lftype}[0]{\mathsf{type}}
\newcommand{\lfkind}[0]{\mathsf{kind}}
\newcommand{\ctx}[0]{\mathsf{ctx}}

\newcommand{\NN}{\mathbb{N}}

\newcommand{\FV}{\mathsf{FV}}

\newcommand{\tmctx}{\ensuremath{\mathsf{ctx}}}
\newcommand{\tm}{\ensuremath{\mathsf{tm}}}

\newcommand{\R}{\mathcal{B}}
\newcommand{\trec}[3]{\mathsf{rec}^{#3}~\ensuremath{{#1}}}
\newcommand{\titer}[3]{\mathsf{rec}^{#3}~\ensuremath{{#1}}}

\newcommand{\mto}{\Rightarrow}

\newcommand{\tmfn}[2]{\mathsf{fn}\; #1 \Rightarrow #2}

\newcommand{\tmrec}[4]{\mathsf{rec}^{#1} (#2 \mid #3 \mid #4)~}
\newcommand{\tmrecctx}[3]{\mathsf{rec}^{#1} (#2 \mid #3)~}

\newcommand{\whnf}{\searrow}
\newcommand{\lfwhnf}{\searrow_{\mathsf{LF}}}

\newcommand{\tightoverset}[2]{%
  \mathop{#2}\limits^{\vbox to -.5ex{\kern-0.95ex\hbox{$#1$}\vss}}}

\newcommand{\neut}{\small\mathsf{wne}~}

\newcommand{\norm}{\small\mathsf{whnf}~}

